\documentclass[aps,
               prx,
               showpacs,
               amssymb,
               superscriptaddress,
               twocolumn,
              nofootinbib, 10pt,
              longbibliography]{revtex4-2}
\pdfoutput=1
\usepackage{pifont}
\usepackage[utf8]{inputenc}
\usepackage[english]{babel}
\usepackage[T1]{fontenc}
\usepackage{amsmath,amssymb,mathtools}
\usepackage{mathpazo} 
\usepackage{amsfonts}
\usepackage{amsthm}
\usepackage[dvipsnames]{xcolor}
\usepackage[pdftex,colorlinks=true,urlcolor= blue,linkcolor=Blue,citecolor=RedViolet]{hyperref}
\usepackage{bm}
\usepackage{tikz}
\usepackage{cancel}
\usepackage{soul}
\usepackage{booktabs} 
\usepackage{siunitx}  
\usepackage{relsize}
\usepackage{physics}
\usepackage{mathrsfs} 
\usepackage{comment}
\usepackage{ragged2e}
\usepackage{graphicx}
\usepackage{float}
\usepackage{tabularx}
\usepackage{array}
\usepackage{subcaption} 
\usepackage{caption}
\usepackage{orcidlink}
\usepackage[normalem]{ulem} 
\usepackage{cleveref}
\usepackage{amsthm}

\newtheorem{theorem}{Theorem}
\newtheorem*{theorem*}{Theorem}

\newtheorem{lemma}{Lemma}

\newtheorem{prop}{Proposition}
\newtheorem{definition}{Definition}

\newcommand{\M}[1]{\mathcal{#1}}

\newcommand{\id}{\mathbb{I}}

\newcommand{\subtiny}[3]{\ensuremath{_{\hspace{#1 pt}\protect\raisebox{#2 pt}{\tiny{$ #3$}}}}}
\newcommand{\suptiny}[3]{\ensuremath{^{\hspace{#1 pt}\protect\raisebox{#2 pt}{\tiny{$ #3$}}}}}

\newcommand{\deff}{d\ensuremath{_{\hspace{-0.5pt}\protect\raisebox{0pt}{\tiny{eff}}}}}

\newcommand{\deffspec}{d\ensuremath{_{\hspace{-0.5pt}\protect\raisebox{0pt}{\tiny{spec}}}}}

\newcommand{\eqst}{{\mu}}
\newcommand{\leakP}{\mathscr{L}\subtiny{-1.3}{0.3}{{P}}}
\newcommand{\leakz}{\mathscr{L}\subtiny{0}{0}{{0}}}
\newcommand{\avginfty}[1]{\left\langle #1 \right\rangle\subtiny{-1}{-2}{T\rightarrow\infty}}

\newcommand{\deffdyn}{d\ensuremath{_{\hspace{-0.5pt}\protect\raisebox{0pt}{\tiny{dyn}}}}}

\newcommand{\Hpow}{H\subtiny{0}{0}{\mathrm{pow}}}

\newcommand{\qnmdyn}{q\ensuremath{_{\hspace{-0.5pt}\protect\raisebox{0pt}{\tiny{nm}}}^{\protect\raisebox{-1pt}{\tiny{dyn}}}}}

\newcommand{\pmu}{p\subtiny{-1.8}{-0.5}{\mu}}

\newcommand{\avgT}[1]{\left\langle #1 \right\rangle\subtiny{0}{0}{T}}
\newcommand{\leak}{\mathscr{L}}
\newcommand{\Ne}{\mathcal{N}}

\newcommand{\up}{\texttt{Up}}
\newcommand{\dw}{\texttt{Dw}}
\newcommand{\mypm}{\texttt{Pm}}

\newcommand{\spd}{\mathcal{G}\subtiny{0}{0}{\mathscr{L}}}

\begin{document}

\title{Spectral analysis of equilibration: information leakage in isolated quantum systems}

\author{Andr\'e T. Ces\'{a}rio\,\texorpdfstring{\orcidlink{0000-0002-6972-2576}}{}}\email{andretcs@ufmg.br}
\affiliation{Instituto de F\'{i}sica, Universidade Federal Fluminense, Niter\'{o}i, RJ, 24210-346, Brazil}
\affiliation{Departamento de F\'{i}sica, Universidade Federal de Pernambuco, Recife, PE, 50670-901, Brazil}

\author{Marcos G. Alpino\,\texorpdfstring{\orcidlink{0009-0000-2390-4086}}{}}
\affiliation{Departamento de F\'{\i}sica - ICEx - Universidade Federal de Minas Gerais,
Av.~Pres.~Ant\^{o}nio Carlos 6627 - Belo Horizonte - MG - 31270-901 - Brazil.}

\author{Reinaldo O. Vianna\,\texorpdfstring{\orcidlink{0000-0003-2857-8552}}{}}
\affiliation{Departamento de F\'{\i}sica - ICEx - Universidade Federal de Minas Gerais,
Av.~Pres.~Ant\^{o}nio Carlos 6627 - Belo Horizonte - MG - 31270-901 - Brazil.}
\author{Tiago Debarba\,\texorpdfstring{\orcidlink{0000-0001-6411-3723}}{}}\email{debarba@utfpr.edu.br}
\affiliation{Departamento Acad{\^ e}mico de Ci{\^ e}ncias da Natureza - Universidade Tecnol{\'o}gica Federal do Paran{\'a}, Campus Corn{\'e}lio Proc{\'o}pio - Paran{\'a} -  86300-000 - Brazil.}

\begin{abstract}
\noindent 
We develop a unified dynamical-spectral framework for equilibration in isolated quantum systems based on a subspace coarse-graining approach. Central to our formulation is the Leakage Fidelity Function (LFF), expressed as the probability that a unitarily evolving state escapes the support of its initial subspace. Whereas in open systems irreversibility is associated with a flow of information from a subsystem to its environment, here the information flows between two complementary subspaces of the Hilbert space of a closed system, the support of the initial state and its orthogonal complement, and the LFF makes this flow explicit: the pair $(H,\rho\subtiny{0}{0}{0})$ equilibrates when, for most times, the information does not return completely to the initial subspace. Using the LFF as a diagnostic of equilibration, we derive finite-time bounds on its mean-square fluctuations in terms of the spectral gap structure, the equilibrium weight of the initial support, and the effective dimension. For arbitrary mixed states, the resulting variance bound scales at most as the inverse of the effective dimension, while for initial pure states, it admits an inverse-square bound. By introducing spectral power distributions, entropic measures, and the associated dynamical and spectral dimensions, we establish a quantitative link between phase mixing, gap participation, and dynamical stability. We further show that the Mandelstam–Tamm relation bounds the first-passage time required for the leakage probability to reach a prescribed threshold, providing a complementary constraint on the onset of information leakage. Our results provide a state-dependent, geometrically transparent perspective on how spectral complexity and subspace information leakage jointly govern irreversibility in closed quantum many-body systems.

\end{abstract}

\maketitle
\begin{figure*}[t]
    \centering  \includegraphics[width=1\linewidth]{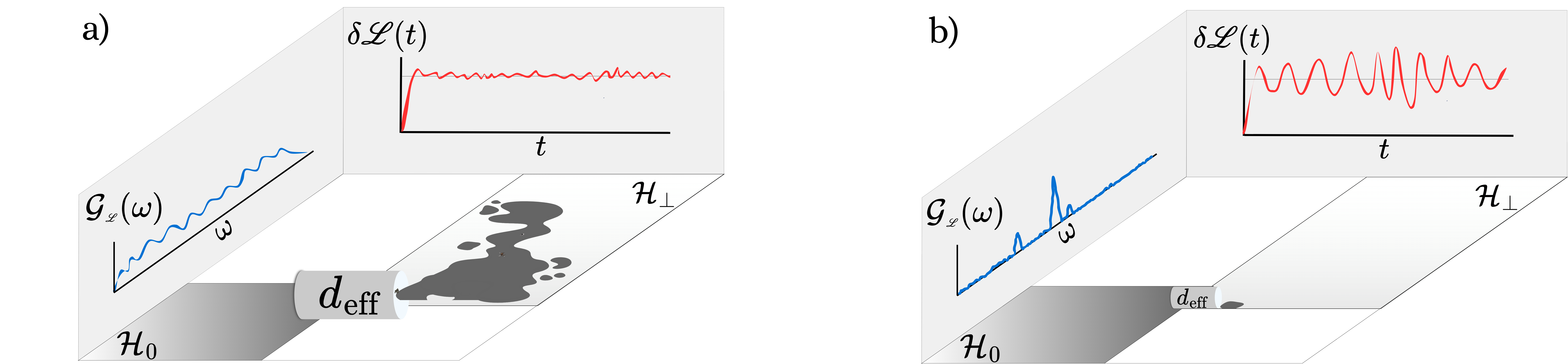}
\caption{\justifying\textbf{Illustrative image of LFF.}
{The Hilbert space is decomposed as
$\mathcal H=\mathcal H\subtiny{0}{0}{0}\oplus\mathcal H_\perp$,
with $\mathcal H\subtiny{0}{0}{0}$ the support of the initial state.
The front wall schematically represents the LFF, while the
side wall depicts its spectral power over Bohr frequencies. Panel (a) depicts a situation in which the LFF has small fluctuations about its long-time average after the initial transient. Panel (b) illustrates persistent revivals and larger fluctuations associated with a more concentrated active-gap spectrum. In particular, the state need not remain entirely inside $\mathcal H\subtiny{0}{0}{0}$ in panel (b): if it did,
$\leak(t)$ would vanish identically.}}
    \label{fig:leakage_image}
\end{figure*}

\section{Introduction}
Across classical and quantum statistical mechanics, equilibration emerges not from fundamental dynamical irreversibility, but from a restriction of description: coarse-graining in phase space \cite{boltzmann1970,ehrenfest1990}, projection onto relevant variables \cite{zwanzig1961,zwanzig1970concept,mori1965}, tracing out environmental degrees of freedom \cite{zurek1982environment}, or limiting attention to experimentally accessible observables \cite{reimann2008foundation,linden2009quantum,reimann2010canonical,short2011equilibration}. Equilibration in isolated quantum systems does not require external baths: on average, unitary dynamics alone typically drives experimentally relevant observables toward steady behavior on intermediate and long timescales~\cite{Reimann2012,Short_2012,gogolin2016equilibration}. In this context, the second law of thermodynamics emerges as a consequence of entropic measures coarse-grained according to macroscopic observables~\cite{meier2025emergence}.~Based on such entropic measures, it is possible to quantify the statistical complexity associated with the order-disorder behavior of a unitary equilibration process~\cite{alpino2025}. The emergence of classicality from purely quantum equilibration has also yielded insight into measurement theory without invoking an {\it ad hoc} Born-rule postulate~\cite{schwarzhans2023quantum,engineer2024equilibration,demelo2024finite}.

In closed quantum systems, the Eigenstate Thermalization Hypothesis (ETH) \cite{deutsch1991,srednicki1994} provides a complementary mechanism whereby thermodynamic behavior arises at the level of few-body observables despite global unitary evolution. While ETH offers one mechanism for equilibration, a broader dynamical route emphasizes phase dispersion, dephasing, and the structure of energy gaps generated by the Hamiltonian~\cite{Anza2017,AnzaHuber2018}. In this perspective, equilibration results from interference among many oscillatory contributions in the energy eigenbasis, whose cumulative effect suppresses long-time fluctuations of observables. The relationship between the equilibration timescale and the spectral properties of the Hamiltonian has been recently explored in some specific many-body settings~\cite{del2017scrambling,hara2026timescale,schwarzhans2026average}.

Under unitary evolution, return amplitudes and fidelities exhibit universal structures. At very short times, one observes the quadratic Zeno regime governed by the energy uncertainty of the initial state \cite{Gorin2006}. At intermediate times, mixing characteristics and decay profiles, on average, depend on the density of accessible states and on the connectivity induced by the Hamiltonian. These regimes reflect the progressive deterioration of initial coherences in the energy eigenbasis~\cite{Emerson2002,Gorin2006,Prosen2002,Kowalewska2009,Weinstein2005}.
The survival probability is directly related to the local density of states (LDOS), defined as the overlap of the initial state with the energy eigenstates~\cite{Zarate2023}. Since the survival probability is the squared modulus of the Fourier transform of the LDOS, equilibration properties are tightly linked to how broadly the initial state spreads over the spectrum. In particular, long-time fluctuations are controlled by participation ratios and by the effective dimension of the explored subspace~\cite{TorresHerrera2018,TimeFluct2013}.
In parallel, the recently proposed \textit{observable statistical mechanics} framework predicts stationary outcome distributions for few-outcome observables from a maximum-entropy principle without requiring microscopic reconstruction~\cite{Scarpa2025}. These developments suggest that equilibration can be understood operationally in terms of restricted observables rather than full-state convergence. 

We study a quantitative measure of the degree of equilibration that captures the suppression of temporal fluctuations and its dependence on system parameters: the \textit{Leakage Fidelity Function} (LFF), a measure of the \textit{infidelity} relative to the initial support. Beyond static participation measures, we examine how the Hamiltonian's spectral structure- the portion of Hilbert space effectively accessed by the initial state- governs the dynamics of equilibration. Equilibration, then, emerges as a consequence of interference among many incommensurate frequencies, producing destructive phase mixing that suppresses persistent oscillations.

To provide a geometric intuition for this mechanism, we introduce in Fig.~\ref{fig:leakage_image}, a schematic representation of the Leakage Fidelity Function ($\leak$). The floor depicts the distribution of the evolving state within the decomposed Hilbert space $\mathcal{H}=\mathcal{H}\subtiny{0}{0}{0} \oplus \mathcal{H}\subtiny{0}{0}{\perp}$ (with $\mathcal{H}\subtiny{0}{0}{0}$ representing the subspace spanned by the initial states' support, and $\mathcal{H}\subtiny{0}{0}{\perp}$ its orthogonal complement). The front wall shows the temporal signal, for a given time window $T$, that is, ${\delta\mathscr{L}(t)=\mathscr{L}(t)-\mathscr{L}\suptiny{0}{0}{\infty}}$ obtained from unitary evolution, while the side wall represents the spectral power density (SPD) profile $\spd(\omega)=~\left|{\delta\leak}(\omega)\right|\suptiny{0}{0}{2}$ determined by the Hamiltonian spectrum, {where $\delta\leak(\omega)=\int\subtiny{0}{0}{-\infty}\suptiny{0}{0}{\infty}\delta\leak(t) e^{-i\omega t} dt$, for $\omega > 0$,} is the Fourier transform of the LFF signal $\delta\leak(t)$.

The contrast between the two panels in Fig.~\ref{fig:leakage_image} highlights how broad spectral participation leads to destructive phase mixing and small temporal fluctuations, whereas restricted spectral support results in persistent oscillatory behavior. These considerations motivate a systematic investigation of the interplay between spectral structure, dynamical constraints, geometric bounds, and the emergent equilibration behavior of isolated quantum systems. Rather than invoking ensemble-typical chaotic diagnostics, we focus on universal bounds and operational quantities that hold for fixed Hamiltonians and initial states, thereby providing a structurally controlled framework for equilibration theory.

\begin{table*}[t]
\caption{Relation between the results of this work and previous literature. The first column lists our result; the second, its established counterpart; the third, what the present work adds. Upper bounds on finite-time fluctuations are compared separately in Table~\ref{tab:LLFs}.}  
\label{tab:known_vs_new}
\small
\newcommand{\ca}[1]{\parbox[t]{4.4cm}{\raggedright #1\strut}}
\newcommand{\cb}[1]{\parbox[t]{5.6cm}{\raggedright #1\strut}}
\newcommand{\cc}[1]{\parbox[t]{5.6cm}{\raggedright #1\strut}}
\begin{tabular}{lll}
\hline\hline
\ca{Result in this work} & \cb{Established counterpart} & \cc{What this work adds}\\
\hline
\ca{$\leakP(t)=1-\Tr(P\rho(t))$, Eq.~\eqref{eq:leakage_fidelity}} &  \cb{For $\rho_0$  pure: $\leakz(t)=1-\mathcal{S}(t)$, Eq.~\eqref{eq:leakz_pure}, where $\mathcal S(t)$ is the Survival Probability (SP)} & \cc{LFF quantifies the coarse-grained statistics $\{P,Q\}$ that diagnoses equilibration of the pair $(H,\rho\subtiny{0}{0}{0})$.}\\[2pt]\hline
\ca{$\leakP\suptiny{0}{0}{\infty}=1-\Tr(P\eqst)$, where $\eqst$ is the diagonal ensemble, Eq.~\eqref{eq:alpha0_leak}} & \cb{
For pure initial states, $\leakz\suptiny{0}{0}{\infty}
=1-1/\deff = 1-\mathrm{IPR}$, Eq.~\eqref{eq:leak_linearentropy}, where $\mathrm{IPR}=\sum_n p_n^2=1/\deff$ is the inverse participation ratio of the initial energy populations ~\cite{CamposVenuti2010,TimeFluct2013}} & \cc{The long-time LFF quantifies how much of the equilibrium state lies in the initial subspace in the long-time limit. }\\[2pt]\hline
\ca{Theorem~\ref{th:leakage_bounds}: Finite-time fluctuations of LFF} & \cb{Finite-time observable bounds~\cite{reimann2008foundation,Short_2012} with $O=P:=\operatorname{supp}(\rho\subtiny{0}{0}{0})$} & \cc{Finite-time variance bounds scaling as
$1/\deff$ for arbitrary mixed initial states
and $1/\deff\suptiny{0}{0}{2}$
for pure initial states, cf.~Table~\ref{tab:LLFs}.}\\[2pt]\hline
\ca{Theorem~\ref{thm:leakage_concentration_dir}: Leakage concentration from variance bound} & \cb{Chebyshev argument for temporal typicality~\cite{reimann2008foundation}} & \cc{Dirichlet law of Haar populations.
Explicit finite-$d_E$ tail bound for the LFF over the Haar ensemble.}\\[2pt]\hline
\ca{Lemma~\ref{lem:exact_variance}: Exact long-time variance of $\leakP$} & \cb{Exact variance of an observable for non-degenerate gaps~\cite{Peres1984ergodicity,reimann2008foundation}} & \cc{Extension to degenerate gaps by coherent grouping of the amplitudes.}\\[2pt]\hline
\ca{Theorem~\ref{thm:LFF_variance_clean}: Long-time fluctuations of the LFF} & \cb{IPR-controlled fluctuations~\cite{TimeFluct2013,Lezama2023}} & \cc{Lower bound $(\leakP\suptiny{0}{0}{\infty})\suptiny{0}{0}{2}/(2\deffdyn)$; equality and closed form for pure states and $\rho\subtiny{0}{0}{0}=P/r$; role of the phases for mixed states.}\\[2pt]\hline
\ca{HUB states: $\leak\subtiny{0}{0}{\mathrm{HUB}}=1-\mathcal{K}(t)$, Eq.~\eqref{eq:HUB_SFF}} & \cb{SFF as the SP of an infinite-temperature thermofield-double state~\cite{del2017scrambling,del2018decay,Das2025}} & \cc{Characterization through Hamiltonian-unbiased bases; $K$-HUB family as a benchmark interpolating with support size.}\\[2pt]\hline
\ca{$\deffspec$ and $\Hpow$, Eqs.~\eqref{eq:def_deffspec} and \eqref{eq:Shannon_power_entropy}} & \cb{No direct counterpart known to us; related use of higher population moments in the generalized SP~\cite{Zarate2023}} & \cc{Descriptors of the spectral organization of the fluctuations, independent of $\deff$ and $\deffdyn$.}\\[2pt]\hline
\ca{Mixed initial states, $P=\operatorname{supp}(\rho\subtiny{0}{0}{0})$} & \cb{Peres mixed-state fidelity (equal to $1-\leakP$ for $\rho\subtiny{0}{0}{0}=P/r$)~\cite{Peres1984}; allegiance~\cite{Sokolov2007}; finite-temperature LE for the XY chain~\cite{CamposVenuti2011}; state leakage~\cite{Wood2018}} & \cc{Support fixed by the pair $(H,\rho\subtiny{0}{0}{0})$; transitions within the support not counted as leakage; Theorems~\ref{th:leakage_bounds} and~\ref{thm:LFF_variance_clean} for arbitrary rank.}\\[2pt]\hline
\ca{Robustness to noise (Sec.~\ref{sec:noise})} & \cb{SFF and SP under decoherence~\cite{xu2019,MatsoukasRoubeas2024quantumchaos}} & \cc{Equilibrium level unchanged by energy-diagonal noise; criterion separating dynamical and kinematic dephasing.}\\[2pt]
\hline\hline
\end{tabular}
\end{table*}

\begin{table*}[t]
\centering
\renewcommand{\arraystretch}{1.3}
\setlength{\tabcolsep}{10pt}
\begin{tabular}{@{}l|c|c|l@{}}
\hline\hline
\textbf{Coarse-graining} & {\textbf{Measurement}} & {\textbf{Mean-square bound}} & \textbf{Reference}\\ 
\hline
Observables &
$O$ &
{$\|O\|\suptiny{0}{0}{2}/\deff$} &
{Reimann~\cite{reimann2008foundation}; Short and Farrelly~\cite{Short_2012}} \\

Probabilities &
{POVM $\{A\subtiny{0}{0}{i}\}_{i=1}\suptiny{0}{0}{s}$} &
{$1/\deff$} &
Meier et al.~\cite{meier2025emergence} \\

Subspaces (LFF) &
{$\{P,Q\}$} &
{$\sum\subtiny{0}{0}{n}p\subtiny{0}{0}{n}\suptiny{0}{0}{2}P\subtiny{0}{0}{nn}(1-P\subtiny{0}{0}{nn})\leq 1/(4\deff)$} &
{Eq.~\eqref{eq:leak_unconditional_deff_bound}} \\
Pure initial states (LFF) &
{$\{P,Q\}$, $P=\ketbra{\psi\subtiny{0}{0}{0}}{\psi\subtiny{0}{0}{0}}$} &
{$(1-1/d\subtiny{0}{0}{E})/\deff\suptiny{0}{0}{2}$} &
Eq.~\eqref{eq:leakZ_bound} \\
\hline\hline
\end{tabular}
\caption{\justifying {Upper bounds on the time-averaged mean-square deviation from equilibrium for different coarse-grained descriptions. The common spectral factor $f(\varepsilon,T)$, defined in Eq.~\eqref{eq:spectral_factor}, is omitted from the table entries. The first two rows summarize established results, while the last two present the LFF bounds derived here. For an $s$-outcome POVM, the bound refers to the squared Euclidean distance between the outcome distribution and its equilibrium value, $\avgT{\|\vec p(t)-\vec p_{\eqst}\|\subtiny{0}{0}{2}\suptiny{0}{0}{2}}$. As discussed in Sec.~\ref{sec:variance_bound}, Eqs.~\eqref{eq:comparacao_meier_1}~and~\eqref{eq:leak_unconditional_deff_bound}, the general POVM estimate to the binary projective measurement $\{P,Q\}$ yields $1/(2\deff)$, whereas Theorem~\ref{th:leakage_bounds} exploits the geometry of this binary partition to obtain a support-dependent bound, universally bounded by $1/(4\deff)$. For pure initial states, Eq.~\eqref{eq:leakZ_bound} further establishes a scaling of order $1/\deff\suptiny{0}{0}{2}$.}}
\label{tab:LLFs}
\end{table*}

\subsection{Main results and Paper outline}

{This work addresses the equilibration of isolated quantum systems from the perspective of the emergence of classicality. We partition the Hilbert space into two complementary subspaces. Let $P$ be the orthogonal projector onto the support of the initial state, $\M{H}\subtiny{0}{0}{0}=\operatorname{supp}(\rho_0)=P\M{H}$, where all the information is initially contained, and let $Q=\id-P$ be the projector onto its orthogonal complement $\M{H}\subtiny{0}{0}{\perp}$, into which information can flow under the unitary dynamics.
Moreover, rather than introducing a new quantity, we investigate the population of the complementary subspace $\{P,Q\}$ -- named here as \textit{Leakage Fidelity Function (LFF)} -- as an equilibration diagnostic of the pair $(H,\rho\subtiny{0}{0}{0})$. Coarse-graining the Hilbert space in a direct sum $\M{H}=\M{H}\subtiny{0}{0}{0}\oplus\M{H}\subtiny{0}{0}{\perp}$ reveals the irreversibility of the dynamics as a flow of information between two subspaces of the Hilbert space of a closed system, and the LFF quantifies this flow. The system equilibrates when this information does not fully return to the initial subspace, except for rare recurrences. For pure initial states, the LFF reduces to $1$ minus the survival probability, and for Hamiltonian-unbiased initial states, it yields the spectral form factor.}

Section~\ref{sec:framework} establishes the conceptual and mathematical basis of our approach. There, we formalize equilibration in closed quantum systems in terms of a fixed subspace decomposition $\mathcal{H}=\mathcal{H}\subtiny{0}{0}{0}\oplus\mathcal{H}_\perp$ determined by the support of the initial state of interest and its orthogonal complement. Within this coarse-grained structure, we motivate LFF as an operational measure of information flux between dynamically coupled subspaces. In Subsection~\ref{sec:fig_of_merit}, we {define} the leakage fidelity function and revise the important figures of merit related to equilibration of isolated quantum systems. 

In Section~\ref{sec:equil_LFF},{we present two results approaching LFF in the context of equilibration.} 
In Subsection~\ref{sec:variance_bound}, we derive the following universal variance bounds: Theorem~\ref{th:leakage_bounds} {bounds finite-time LFF fluctuations for arbitrary initial density operators by both a support-dependent quantity and a rank-independent term of order $\deff\suptiny{0}{0}{-1}$}; Theorem~\ref{thm:leakage_concentration_dir} further establishes concentration of the LFF around its equilibrium value, both for fixed initial states evaluated at random times and for fixed times with initial states drawn from the Haar measure. 

Section~\ref{sec:leak_spec_power} analyzes the spectral power of the LFF and its cumulative weights as a function of the Hamiltonian dimension, and introduces two spectral quantifiers for the description of equilibration: the Shannon power entropy $\Hpow$ and the spectral dimension $\deffspec$, which quantify how the fluctuation signal is distributed over its Bohr frequencies. These quantities characterize the degree of spectral delocalization of the dynamics in the frequency domain, quantifying the effective number of active energy gaps and the resulting phase-mixing processes that govern the suppression of revivals and the approach to equilibration.

Section~\ref{sec:diagnosing_equilibration} is dedicated to introducing the dynamical dimension ($\deffdyn$), which quantifies the fluctuations of the LFF time signal. It is computed as a function of the number of cosine combinations that compose the LFF signal. Theorem~\ref{thm:LFF_variance_clean} shows the exact value of the LFF time-averaged variance in the long-time regime, expressing an exact relation for the role played by $\deffdyn$ in the emergence of equilibration in isolated quantum systems. The physical interpretation of the dynamical dimension and its spectral counterpart is also provided, along with two numerical examples. {Subsection~\ref {sec:noise}, we present and discuss how the approach is robust under dynamical noise, and also its stability under initial-state perturbations. Subsection~\ref {sec:MT_bound_leak} establishes a speed-limit relation for the time-averaged LFF, thereby revealing the minimum
average time required for the system to evolve from $\M{H}\subtiny{0}{0}{0}$ to $\M{H}\subtiny{0}{0}{\perp}$.}

Throughout, all numerical illustrations refer to the same physical system: an $N$-spin-$1/2$ Ising-like model with transverse field. In particular, Figs.~\ref{fig:fig2} to \ref{fig:dephasing} and Table~\ref{tab:curvature_states} illustrate the spectral and averaged equilibration characteristics of this model. Details on the physical system and the numerical methods employed are omitted from the main text; we refer the reader to Appendix~\ref{appsec:appE}.

We conclude in Section~\ref{sec:conclusions} with a summary of our results, {the relation with previous literature} and an outlook on future directions. Technical proofs and supplementary results are presented in the Appendices.

\section{Framework}\label{sec:framework}
In this section, we briefly review the theoretical framework describing the equilibration process in closed quantum systems and, based on these results, motivate and introduce the LFF. A general summary of equilibration in isolated quantum systems under different coarse-graining schemes is presented in Table \ref{tab:LLFs}. 

\subsection{Equilibration of Isolated Quantum Systems}
{Consider an isolated quantum system described by the Hamiltonian
$H_d=\sum\subtiny{0}{0}{n} E\subtiny{0}{0}{n}\Pi\subtiny{0}{0}{n}$ on a $d$-dimensional Hilbert space
$\M{H}_d$, where $\Pi\subtiny{0}{0}{n}$ are the projectors onto its eigenspaces. The
system is prepared in an initial state with spectral decomposition
$\rho\subtiny{0}{0}{0}=\sum_{\alpha=0}^{r-1}\lambda_\alpha\ketbra{\varphi_\alpha}{\varphi_\alpha}$,
$\lambda_\alpha>0$, whose support is the range of
$P=\sum_{\alpha}\ketbra{\varphi_\alpha}{\varphi_\alpha}$, with rank $r=\Tr(P)$. Under the
unitary dynamics $\rho(t)=e^{-iH_dt}\rho\subtiny{0}{0}{0}e^{iH_dt}$ (with $\hbar=1$), the
state never leaves the subspace
$\M{H}=\mathrm{span}\{\Pi\subtiny{0}{0}{n}\ket{\varphi_\alpha}\}_{n,\alpha}$, the smallest
$H_d$-invariant subspace containing the support of $\rho\subtiny{0}{0}{0}$.}

{Although the formalism developed here holds in general, including for degenerate Hamiltonians, in what follows we restrict ourselves to the nondegenerate case, $\Pi\subtiny{0}{0}{n}=\ketbra{E\subtiny{0}{0}{n}}{E\subtiny{0}{0}{n}}$, since our focus is
the study of equilibration in isolated quantum systems. In this case $\M{H}$ is simply spanned by
the populated eigenvectors, i.e., those with $\bra{E\subtiny{0}{0}{n}}\rho\subtiny{0}{0}{0}
\ket{E\subtiny{0}{0}{n}}>0$; we denote their number by $d\subtiny{0}{0}{E}=\dim\M{H}$ and
relabel them as $n=0,\dots,d\subtiny{0}{0}{E}-1$.}

These dynamical features are governed by the spectral structure of the initial state $\rho\subtiny{0}{0}{0}$. 
The relevant equilibrium state for unitary dynamics is obtained from the long-time average of the state. We first define the time average of $\rho(t)$ over a finite interval $T$ as
\begin{equation}
\label{eq:avgt}
\langle \rho(t)\rangle\subtiny{0}{0}{T}
:=
\frac{1}{T}
\displaystyle\int_{0}^{T}
\rho(t)\,\mathrm{d}t.
\end{equation}
{The equilibrium state $\eqst$ is defined as the long-time average}
\begin{equation}
\label{eq:equilibrium_state}
\eqst
:=
\lim_{T\rightarrow\infty}
\langle \rho(t)\rangle\subtiny{0}{0}{T}.
\end{equation}
For $H =
\sum_{n=0}^{d\subtiny{0}{0}{E}-1}E\subtiny{0}{0}{n}\ketbra{E\subtiny{0}{0}{n}}{E\subtiny{0}{0}{n}}$, this state takes the explicit form $\eqst\equiv\sum\subtiny{0}{0}{n}\ketbra{E\subtiny{0}{0}{n}}{E\subtiny{0}{0}{n}}\rho\subtiny{0}{0}{0}\ketbra{E\subtiny{0}{0}{n}}{E\subtiny{0}{0}{n}},$ which corresponds to the diagonal ensemble in the energy eigenbasis \cite{Short_2012}. {The effective dimension is defined as the inverse of the purity of the equilibrium state,
quantifying the effective number of energy eigenstates contributing to the long-time dynamics,
\begin{equation}
\label{eq:deff}
\deff
:=
\frac{1}{\Tr\!\left(\eqst\suptiny{0}{0}{2}\right)}
= \frac{1}{\sum_n \rho\subtiny{0}{0}{nn}\suptiny{0}{0}{2}},
\end{equation}
where $\rho\subtiny{0}{0}{nn}=\bra{E\subtiny{0}{0}{n}}\rho\subtiny{0}{0}{0}\ket{E\subtiny{0}{0}{n}}$ are the diagonal elements of $\rho\subtiny{0}{0}{0}$ in the energy eigenbasis.}
In particular, for pure initial states $\ket{\psi\subtiny{0}{0}{0}}=\sum\subtiny{0}{0}{n} c\subtiny{0}{0}{n} \ket{E\subtiny{0}{0}{n}}$, $\deff 
= {1}/{\sum\subtiny{0}{0}{n} p\subtiny{0}{0}{n}\suptiny{0}{0}{2}}$, with $p\subtiny{0}{0}{n}= |c\subtiny{0}{0}{n}|\suptiny{0}{0}{2}$. 

{Since long-time fluctuations of few-body observables scale inversely with $\deff$, the effective dimension provides a link between the initial state's energy
distribution and the suppression of temporal fluctuations in closed quantum  systems~\cite{reimann2008foundation}.}
In this context, for a given observable $O$, a fundamental bound on finite-time  observable equilibration in closed quantum systems is provided by~\cite{Short_2012},
\begin{equation}\label{eq:bound_Reimann}
\avgT{|\Tr[\rho(t)O]-\Tr[\eqst O]|\suptiny{0}{0}{2}}
\leq\frac{f(\varepsilon,T)\|O\|\suptiny{0}{0}{2}}{\deff}.
\end{equation}
\begin{equation}\label{eq:spectral_factor}
f(\varepsilon,T):=\Ne(\varepsilon)
\left(1+\frac{8\log d_E}{\varepsilon T}\right).
\end{equation}
Here, $\|O\|=\max_{\ket{\zeta}\in\M{H}} |\bra{\zeta}O\ket{\zeta}|$ denotes the operator norm, $f(\varepsilon,T)$ is {a spectral quantity}, with $\mathcal{N}(\varepsilon)$ denoting the maximal number of energy gaps contained in any interval with width $\varepsilon$, {and $d_E$ denotes the number of distinct energy levels \cite{Short_2012}.} 

The inequality in Eq.~\eqref{eq:bound_Reimann} captures a key structural feature of equilibration: deviations from the equilibrium expectation value are suppressed by the effective dimension, while the operator norm acts as the scale that converts state-space delocalization into observable fluctuations. Importantly, the bound only has physical content when $\|O\|$ grows at most polynomially with the system size $N$. If $\|O\|$ were allowed to increase exponentially with $N$, the right-hand side of Eq.~\eqref{eq:bound_Reimann} would dominate the $1/\deff$ suppression, rendering the inequality unable to diagnose equilibration and effectively masking the decay of fluctuations. 

We introduce a coarse-graining procedure based on the subspace structure determined by the initial state. As discussed previously, equilibration emerges from restrictions on the physical description, whereby information can flow between distinct macroscopic subspaces of the Hilbert space $\M{H}$. For a given initial state $\rho\subtiny{0}{0}{0}$, let $P$ denote the projector onto its support, defining a subspace $\M{H}\subtiny{0}{0}{0} \subset \M{H}$ such that $P\rho\subtiny{0}{0}{0} P = \rho\subtiny{0}{0}{0}$. By construction, there exists an orthogonal projector $Q = \id - P$ spanning the complementary subspace $\M{H}\subtiny{0}{0}{\perp} \subset \M{H}$, so that the total Hilbert space decomposes as
\begin{equation}
    \M{H} = \M{H}\subtiny{0}{0}{0} \oplus \M{H}\subtiny{0}{0}{\perp}.
\end{equation}
We, thus, restrict the dynamical description to the coarse-grained subspace $\M{H}_\perp$, complementary to the initial state $\rho\subtiny{0}{0}{0}\in\M{H}\subtiny{0}{0}{0}$. {Fig.~\ref{fig:leakage_image} illustrates the redistribution
of probability between the initial support and its complement.
The magnitude of LFF depends on the support geometry,
whereas its fluctuations obey the bounds derived below.} {In the long-time limit, and assuming nondegenerate
nonzero energy gaps, setting $O=P$ gives} 
\begin{equation}\label{eq:bound_weak}
     \big\langle\big|\Tr(\rho(t)P) - \Tr(\eqst P)\big|\suptiny{0}{0}{2}\big\rangle_{T\rightarrow\infty} 
    \;\le\; \frac{1}{\deff}, 
\end{equation}
where $\Tr(\rho(t)P)$ represents the  projection of $\rho(t)$ that is still localized in $\M{H}\subtiny{0}{0}{0}$.  Notice that this fidelity provides a meaningful quantitative characterization of the emergence of equilibration \cite{del2017scrambling,del2018decay}. Indeed, if the state remains confined to the initial subspace throughout the entire evolution, it remains invariant under the action of $P$ at all times, resulting in $\Tr(\rho(t)P)=1$ for every $t$. One may, therefore, construct an indicator of equilibration based on deviations from this condition, namely, when the state acquires support in the complementary subspace $\M{H}\subtiny{0}{0}{\perp}$. 

\subsection{Leakage Fidelity Function and Related Figures of Merit}\label{sec:fig_of_merit}
Within the framework presented above, equilibration can be analyzed in terms of the information flux between $\M{H}\subtiny{0}{0}{0}$ and its orthogonal complement $\M{H}\subtiny{0}{0}{\perp}$, as generated by the  Hamiltonian. In particular, deviations from the invariance of $\M{H}\subtiny{0}{0}{0}$ under the dynamics quantify the information flux, populations and coherences, into $\M{H}\subtiny{0}{0}{\perp}$, providing a natural measure of coarse-grained irreversibility. Quantitatively, we denote this measure of information flux as presented in Def.~\ref{def:leakage}
{
\begin{definition}[Leakage Fidelity Function (LFF)]
\label{def:leakage}
{Let $\rho\subtiny{0}{0}{0}$ be an arbitrary density operator
on $\mathcal{H}$, and let
$P=\operatorname{supp}(\rho\subtiny{0}{0}{0})$ be the
orthogonal projector onto its support
$\mathcal{H}\subtiny{0}{0}{0}\subseteq\mathcal{H}$.}
For a state $\rho(t)$ evolving unitarily under
a time-independent Hamiltonian $H$,
the LFF with respect to $\mathcal{H}\subtiny{0}{0}{0}$
is defined as
\begin{equation}
\label{eq:leakage_fidelity}
\leakP(t):=\Tr[(\id-P)\rho(t)],
\end{equation}
{i.e., the probability that the system is
found in the orthogonal complement
$\mathcal{H}\subtiny{0}{0}{\perp}$ of its initial support
at time $t$.}
\end{definition}
}
{For a mixed initial state with spectral decomposition $\rho\subtiny{0}{0}{0}=\sum_{\alpha=0}^{r-1}
\lambda_\alpha
\ketbra{\varphi_\alpha}{\varphi_\alpha}$,
the LFF can be expressed as a convex average
of the escape probabilities of its pure-state components,}
\begin{equation}
\label{eq:leakage_mixed_convex}
{
\leakP(t)=
\sum_{\alpha=0}^{r-1}
\lambda_\alpha
\left\|
(\id-P)U(t)\ket{\varphi_\alpha}
\right\|\suptiny{0}{0}{2}.
}
\end{equation}
{Each component contributes an escape probability
$\leak\subtiny{0}{0}{\alpha}\suptiny{0}{0}{(P)}(t)
=\|(\id-P)U(t)\ket{\varphi_\alpha}\|^2$,
evaluated with respect to the same initial support projector $P$. These probabilities generally differ from the complements of the individual survival probabilities, since transitions between states within the initial support do not contribute
to leakage.} {Our analysis concerns initial states for which $\leakP(t)$ is not identically zero, i.e., for which there is a flow of population between the initial subspace and its orthogonal complement. Since $\Tr[P\rho(t)]=\Tr[P\rho_0]$ for all $t$ whenever $[H,P]=0$, this excludes the case in which the support of $\rho_0$ is spanned by energy eigenstates, which includes stationary states, $\rho_0=\eqst$.}

Notice that we are not proposing a new quantifier; LFF is a natural measure of \textit{infidelity} relative to the initial state. The main goal is to approach the progressive loss of memory in the initial-state subspace as a coarse-graining mechanism for capturing a physically transparent diagnosis underlying equilibration.
The information spread across different subspace sectors reflects, on average, the dynamical irreversibility relative to the Hamiltonian spectral structure. {LFF at equilibrium also depends on the orientation of the
initial support relative to the energy eigenbasis: for mixed states,
$\leakP\suptiny{0}{0}{\infty}=1-p\subtiny{-1.8}{-0.5}{\mu}$, while for pure states
$\leakz\suptiny{0}{0}{\infty}=1-1/\deff$. In particular, for a full-rank initial state, with $P=\id$, no information leakage occurs even for a large effective dimension.} Thus, the LFF provides a simple and transparent indicator of the extent to which the unitary dynamics have diluted the information initially encoded in the state. 

The LFF, defined in Eq.~\eqref{eq:leakage_fidelity},
quantifies the probability that the system leaves the monitored subspace (providing a direct, perturbation-free measure of irreversibility). Expanding $\rho(t)$ at short times gives $\leakP(t)=t\suptiny{0}{0}{2}\Tr(\rho_{0}PHQHP)+O(t\suptiny{0}{0}{3})$, which shows that the quadratic Zeno onset is determined by the matrix elements $H\subtiny{0}{0}{mn}=\matrixel{E\subtiny{0}{0}{m}}{PHQHP}{E\subtiny{0}{0}{n}}$, i.e., the couplings between the $P$ and $Q=\id-P$ sectors that mediate the initial escape from the monitored subspace. These couplings define the local generator of equilibration and encode the same microscopic structure that drives fidelity decay \cite{Emerson2002,Gorin2006,Li2007,Karch2025}.

For pure initial states, the support projector has rank one, and the LFF of Def.~\ref{def:leakage} reduces to the complement of the well-established \textit{survival probability} (SP)~\cite{Fonda1978}. Specifically, for $\rho\subtiny{0}{0}{0}=P=\ketbra{\psi\subtiny{0}{0}{0}}{\psi\subtiny{0}{0}{0}}$,
\begin{equation}
\label{eq:leakz_pure}
\leakz(t)=1-|\braket{\psi(t)}{\psi\subtiny{0}{0}{0}}|\suptiny{0}{0}{2}=1-S(t),
\end{equation}
where $S(t)=|\braket{\psi(t)}{\psi\subtiny{0}{0}{0}}|\suptiny{0}{0}{2}$ is the SP. For more details on this quantity, see the review in Ref.~\cite{Zarate2023} and the references therein. Writing
$\ket{\psi\subtiny{0}{0}{0}}=\sum\subtiny{0}{0}{n} c\subtiny{0}{0}{n}\ket{E\subtiny{0}{0}{n}}$, one obtains
\begin{equation}\label{eq:survprob}
    \mathcal{S}(t)
    =\left|\sum\subtiny{0}{0}{n} |c\subtiny{0}{0}{n}|\suptiny{0}{0}{2}e^{-iE\subtiny{0}{0}{n}t}\right|\suptiny{0}{0}{2}
    =\left|\int \rho_{\mathrm{ini}}(E)
    e^{-iEt}\,dE\right|\suptiny{0}{0}{2},
\end{equation}
where
$\rho_{\mathrm{ini}}(E)=\sum\subtiny{0}{0}{n}|c\subtiny{0}{0}{n}|\suptiny{0}{0}{2}\delta(E-E\subtiny{0}{0}{n})$
is the local density of states (LDOS).
The survival probability is widely employed in studies of quantum equilibration and ergodicity
~\cite{Zarate2023,SantosTavora2016,TorresHerrera2014,TorresHerrera2015},
many-body localization~\cite{Schiulaz2019},
and quantum chaos~\cite{Santos2020}.
The LDOS governs its short-time decay,
whereas long-time behavior can reveal spectral correlations and rigidity
~\cite{TorresHerrera2017,TorresHerrera2018,Lerma2019,Lezama2021}.

A further specialization arises when the pure initial state is Hamiltonian-unbiased, with $|c\subtiny{0}{0}{n}|\suptiny{0}{0}{2}=1/d_E$ for all $n$. In this case, the survival probability $\M{S}\subtiny{0}{0}{\mathrm{HUB}}(t)$ coincides with the normalized \textit{spectral form factor} (SFF), $\mathcal{K}(t)$
\begin{equation}\label{eq:SFF}
    \mathcal{K}(t)
    =\frac{1}{d_E\suptiny{0}{0}{2}}
    \left|\Tr e^{-iHt}\right|\suptiny{0}{0}{2}
    =\mathcal{S}\subtiny{0}{0}{\mathrm{HUB}}(t),
\end{equation}
yielding $\leak\subtiny{0}{0}{\mathrm{HUB}}(t)=1-\mathcal{K}(t)$. The SFF is a standard diagnostic of spectral correlations and quantum chaos, whose characteristic dip-ramp-plateau structure is typically resolved by averaging over an ensemble of Hamiltonians 
~\cite{mehta2004random,haake2010quantum,TorresHerrera2017,Das2025,Dong2025}.
{These successive specializations connect the LFF $\leakP(t)$, applicable to mixed states of arbitrary rank, to the complements of the SP and SFF, extending the analysis of equilibration and temporal fluctuations beyond these well-established and important special cases.}

\section{Equilibration of Leakage Fidelity Function}
\label{sec:equil_LFF}

In this section, we characterize equilibration through the long-time average of the LFF and its temporal fluctuations. The former quantifies the average information leakage from the initial support, while the latter reveals the persistence of dynamical memory and revivals.

\subsection{Leakage Fidelity Function variance bounds}
\label{sec:variance_bound}

For arbitrary initial states, the long-time LFF depends on both the energy populations and the orientation of the initial support projector relative to the Hamiltonian
eigenbasis. To characterize this dependence, we introduce the equilibrium support probability
\begin{equation}
\label{eq:alpha0_leak}
p\subtiny{-1.8}{-0.5}{\mu}
:=\lim_{T\to\infty}\avgT{\Tr(P\rho(t))} =
\Tr(P\eqst)
=
\sum\subtiny{0}{0}{n}p\subtiny{0}{0}{n}P\subtiny{0}{0}{nn},
\end{equation}
where
$p\subtiny{0}{0}{n}
=\bra{E\subtiny{0}{0}{n}}
\rho\subtiny{0}{0}{0}
\ket{E\subtiny{0}{0}{n}}$
and $P\subtiny{0}{0}{nn} =\bra{E\subtiny{0}{0}{n}}
P\ket{E\subtiny{0}{0}{n}}$. The quantity $p\subtiny{-1.8}{-0.5}{\mu}$ measures the probability retained by the equilibrium state $\eqst$ within the initial support. Consequently, the long-time averaged LFF will be
\begin{equation}
\label{eq:leakP_infty}
\leakP\suptiny{0}{0}{\infty} = 1-\Tr(P\eqst) = 1-p\subtiny{-1.8}{-0.5}{\mu}.
\end{equation}
While this equilibrium value characterizes the long-time average LFF, it does not determine the magnitude of its temporal fluctuations. These depend on the off-diagonal spectral coefficients and the interference among
Bohr frequencies $\omega\subtiny{0}{0}{nm} =E\subtiny{0}{0}{n}-E\subtiny{0}{0}{m}$. To quantify these deviations, we introduce the centered signal
\begin{equation}
\label{eq:deltaLeakP}
\delta\leakP(t) = \leakP(t)-\leakP\suptiny{0}{0}{\infty},
\end{equation}
which describes the temporal fluctuations around the long-time LFF for a given Hamiltonian--initial-state pair $(H,\rho\subtiny{0}{0}{0})$.
To make this dependence explicit, we write $\rho\subtiny{0}{0}{nm}=\bra{E\subtiny{0}{0}{n}}\rho\subtiny{0}{0}{0}\ket{E\subtiny{0}{0}{m}}$ and $P\subtiny{0}{0}{mn}=\bra{E\subtiny{0}{0}{m}}P\ket{E\subtiny{0}{0}{n}}$ in the energy eigenbasis and introduce the complex coefficients $z\subtiny{0}{0}{nm}=\rho\subtiny{0}{0}{nm}P\subtiny{0}{0}{mn}=b\subtiny{0}{0}{nm}e^{i\phi\subtiny{0}{0}{nm}}$, with $b\subtiny{0}{0}{nm}=|\rho\subtiny{0}{0}{nm}P\subtiny{0}{0}{mn}|$. In terms of these amplitudes, we define the spectral moments
\begin{equation}
\label{eq:spectral_moments}
S\subtiny{0}{0}{k}
:=
\sum\subtiny{0}{0}{n<m}
b\subtiny{0}{0}{nm}\suptiny{0}{0}{k},
\qquad k=1,2,\ldots.
\end{equation}
The first four moments characterize different
orders of the amplitude distribution: $S\subtiny{0}{0}{1}$ gives its total weight, $S\subtiny{0}{0}{2}$ determines the long-time variance for nondegenerate energy gaps, as seen in Lemma~\ref{lem:exact_variance}. 
{For nondegenerate energy levels,
the centered LFF admits the exact real-phasor expansion}
\begin{equation}
\label{eq:leak_mixed_phasor}
{\delta\leakP(t)
=-2\sum\subtiny{0}{0}{n<m}b\subtiny{0}{0}{nm}\cos(\omega\subtiny{0}{0}{nm}t-\phi\subtiny{0}{0}{nm}),}
\end{equation}
{where the sum is over ordered pairs, and $\omega\subtiny{0}{0}{nm}=E\subtiny{0}{0}{n}-E\subtiny{0}{0}{m}$, and $\phi\subtiny{0}{0}{nm} = \mathrm{arg}(\rho\subtiny{0}{0}{nm}P\subtiny{0}{0}{mn})$. 
Unlike the pure-state case, the coefficients $\rho\subtiny{0}{0}{nm}P\subtiny{0}{0}{mn}$ generally carry nonzero phases and cannot be replaced by products of energy populations.}

The bound derived in Theorem~\ref{th:leakage_bounds} establishes a universal constraint on such fluctuations in terms of three spectral quantities: (i) the gap-density function $\Ne(\varepsilon)$, which measures near-degeneracies \cite{Short_2012}; (ii) the effective dimension $\deff$, which encodes the delocalization of the initial state in the energy eigenbasis; and (iii) {the Hilbert space dimension $d\subtiny{0}{0}{\mathrm{E}}$, for non-degenerate Hamiltonians.} It shows that large temporal oscillations are suppressed whenever the initial state is highly delocalized or the energy gaps are sufficiently irregular to induce strong dephasing.


\begin{figure*}[t!]
    \centering
    \begin{subfigure}{0.49\textwidth}
        \centering
        \includegraphics[width=\linewidth]{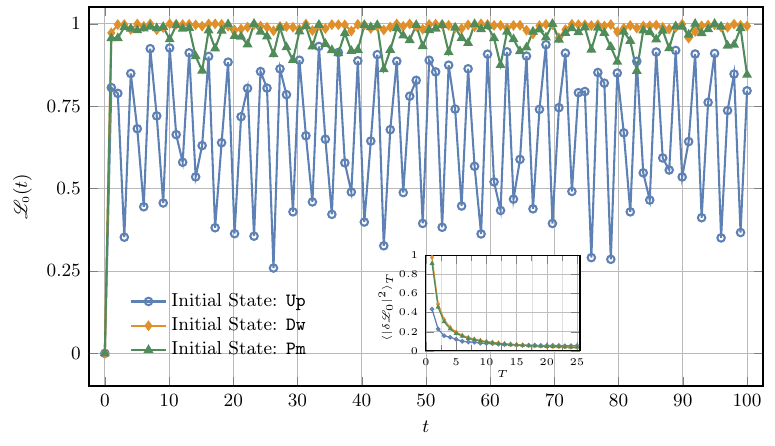}
        \caption{Pure case.}
        \label{fig:fig2a}
    \end{subfigure}
    \begin{subfigure}{0.49\textwidth}
        \centering
        \includegraphics[width=\linewidth]{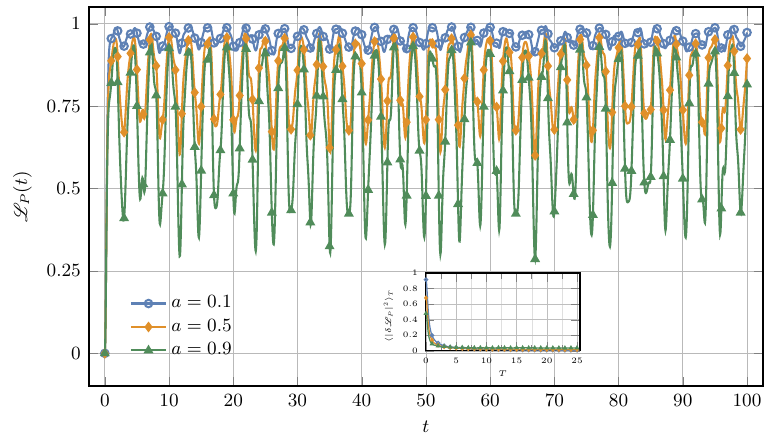}
        \caption{Mixed case.}
        \label{fig:fig2b}
    \end{subfigure}
    \caption{\justifying Time evolution of LFF in the
    $N=10$ spin Ising chain.
    Left: $\leakz(t)$ for the pure initial
    states $\up$, $\dw$, and $\mypm$.
    Right: $\leakP(t)$ for the mixed initial states $\rho_0(a)=a\ketbra{\up}{\up}+(1-a)\ketbra{\dw}{\dw}$,
    with $a=0.1,0.5,0.9$, sharing the same rank-two support projector
    $P=\ketbra{\up}{\up}+\ketbra{\dw}{\dw}$.
    The insets show the running mean-square deviation from the corresponding long-time value,
    \(
    \avgT{|\delta\!\mathscr{L}\subtiny{0}{0}{i}|\suptiny{0}{0}{2}},
    \)
    where $\mathscr{L}\subtiny{0}{0}{i}=\{\leakz, \leakP\}$. The pure $\up$ state exhibits large and persistent oscillations,
    whereas $\dw$ and $\mypm$ display substantially smaller temporal
    fluctuations after the initial transient.
    For the mixed states, increasing the weight of $\up$ lowers {long-time LFF} and enhances temporal fluctuations, even though
    the initial support remains unchanged.}
\label{fig:fig2}
\end{figure*}
{\begin{theorem}[Finite-time fluctuations of LFF]
\label{th:leakage_bounds}
Let $H$ have $d\subtiny{0}{0}{E}$ distinct energy levels, let $\rho\subtiny{0}{0}{0}{\neq\mu}$ be an arbitrary initial state with support projector $P$, and write $p\subtiny{0}{0}{n}=\bra{E\subtiny{0}{0}{n}}\rho\subtiny{0}{0}{0}\ket{E\subtiny{0}{0}{n}}$ and $P\subtiny{0}{0}{nn}=\bra{E\subtiny{0}{0}{n}}P\ket{E\subtiny{0}{0}{n}}$. Then, for every $\varepsilon,T>0$:\\
(i) the fluctuations are controlled by the equilibrium value of the LFF,
\begin{equation}\label{eq:leak_bound}
\avgT{|\delta\leakP(t)|\suptiny{0}{0}{2}}\leq f(\varepsilon,T)\left(1-\frac{1}{d\subtiny{0}{0}{E}}\right)\pmu\suptiny{0}{0}{2},
\end{equation}
where $\pmu=\Tr(P\eqst)=1-\leakP\suptiny{0}{0}{\infty}$ is the equilibrium support probability, defined in Eq.~\eqref{eq:alpha0_leak};\\
(ii) irrespective of the equilibrium value,
\begin{equation}\label{eq:leak_unconditional_deff_bound}
\avgT{|\delta\leakP(t)|\suptiny{0}{0}{2}}\leq f(\varepsilon,T)\sum\subtiny{0}{0}{n}p\subtiny{0}{0}{n}\suptiny{0}{0}{2}P\subtiny{0}{0}{nn}(1-P\subtiny{0}{0}{nn})\leq\frac{f(\varepsilon,T)}{4\deff}.
\end{equation}
The spectral factor $f(\varepsilon,T)$ is defined in Eq.~\eqref{eq:spectral_factor}.
\end{theorem}}

A complete proof of Theorem~\ref{th:leakage_bounds} is given in Appendix~\ref{appsec:theorem_leak_variance}, {where we also derive two alternative bounds in terms of the smallest nonzero eigenvalue $\lambda\subtiny{0}{0}{\min}$ and the rank $r$ of the initial state $\rho_0$.} {In the trivial case $[H,P]=0$, one has $P_{nn}\in\{0,1\}$ and $\pmu=1$. The right-hand side of Eq.~(18) then vanishes, and the bound is saturated, consistently with $\leakP(t)\equiv0$, whereas Eq.~(17) is trivially satisfied.}

{Bound (i) constrains the temporal fluctuations by the equilibrium value of the LFF, since the right-hand side of Eq.~\eqref{eq:leak_bound} is proportional to $\pmu^2=(1-\leakP\suptiny{0}{0}{\infty})^2$. When a large fraction of the equilibrium population remains in the initial subspace, as for initial states of high rank whose support covers most of the populated energy levels, bound (i) becomes of order $f(\varepsilon,T)$, and bound (ii) takes over, guaranteeing $\mathcal{O}(\deff^{-1})$ fluctuations for every initial state whenever $f(\varepsilon,T)$ remains bounded. Bound (ii) has the same order as the observable bound of Eq.~\eqref{eq:bound_Reimann} with $O=P$; the gain lies in the prefactor $P\subtiny{0}{0}{nn}(1-P\subtiny{0}{0}{nn})\leq1/4$, which reflects the two-outcome structure of the coarse-graining.}

{For pure initial states, $\rho\subtiny{0}{0}{0}=P=\ketbra{\psi\subtiny{0}{0}{0}}{\psi\subtiny{0}{0}{0}}$ with $\ket{\psi\subtiny{0}{0}{0}}=\sum\subtiny{0}{0}{n}c\subtiny{0}{0}{n}\ket{E\subtiny{0}{0}{n}}$, the LFF is the complement of the survival probability, $\leakz(t)=1-\mathcal{S}(t)$, resulting \cite{CamposVenuti2010,TimeFluct2013}
\begin{align}
\leakz(t)=
1-\frac{1}{\deff}
-
2\sum\subtiny{0}{0}{n<m}p\subtiny{0}{0}{n}p\subtiny{0}{0}{m}\cos(\omega\subtiny{0}{0}{nm}t).\label{eq:cos_only}
\end{align}
Moreover, $P\subtiny{0}{0}{nn}=p\subtiny{0}{0}{n}=|c\subtiny{0}{0}{n}|\suptiny{0}{0}{2}$, and the equilibrium support probability reduces to the inverse participation ratio, $\pmu=\sum\subtiny{0}{0}{n}p\subtiny{0}{0}{n}\suptiny{0}{0}{2}=1/\deff$: the population that the equilibrium state retains in the initial subspace is fixed by the number of energy eigenstates effectively occupied by $\ket{\psi\subtiny{0}{0}{0}}$. 
LFF in the long-time mean, given by 
$\leakz\suptiny{0}{0}{\infty}$, coincides with the linear entropy $H\subtiny{0}{0}{\leak}(\eqst)=1 - \Tr(\eqst\suptiny{0}{0}{2})$ of $\eqst$:
\begin{equation}
\label{eq:leak_linearentropy}
\leakz\suptiny{0}{0}{\infty} =  H\subtiny{0}{0}{\leak}(\eqst)
  = 1 - \frac{1}{\deff}.
\end{equation}  
Consequently, the mean-square temporal deviations
of the LFF and the survival probability coincide,
\begin{equation}
\avgT{|\delta\leakz(t)|\suptiny{0}{0}{2}}
=
\avgT{|\delta\mathcal{S}(t)|\suptiny{0}{0}{2}}.
\end{equation}
Since $P\subtiny{0}{0}{nn}=p\subtiny{0}{0}{n}$
and
$\pmu
=\sum\subtiny{0}{0}{n}
p\subtiny{0}{0}{n}\suptiny{0}{0}{2}
=1/\deff$,
Eq.~\eqref{eq:leak_bound} of
Theorem~\ref{th:leakage_bounds}
yields the following bound:
\begin{align}
  &\avgT{\lvert \delta\!\leakz(t)\rvert\suptiny{0}{0}{2}}
  \;\le\;
   f(\varepsilon,T)\,\left(1-\frac{1}{d\subtiny{0}{0}{E}}\right)\,\frac{1}{\deff\suptiny{0}{1}{2}}.
  \label{eq:leakZ_bound}
\end{align}
For pure initial states, the mean-square deviation bound scales as $1/\deff\suptiny{0}{0}{2}$, resulting in a tighter bound if compared to the variance bound in Eq.~\eqref{eq:bound_weak}, which is a direct application of the proposed coarse-graining applied to the observable equilibration bound \cite{Reimann2012,Reimann2012,Short_2012}.  }
For nondegenerate energy levels, Eqs~\eqref{eq:cos_only} and \eqref{eq:leak_linearentropy} result in the centered LFF 
\begin{equation}
\label{eq:cosine_deltaL}
\delta\!\leakz(t)
= \leakz(t)- \leakz\suptiny{0}{0}{\infty} =
-2\sum\subtiny{0}{0}{n<m} p\subtiny{0}{0}{n} p\subtiny{0}{0}{m} \cos(\omega\subtiny{0}{0}{nm} t).
\end{equation}
As an interesting remark, the cosine-only decomposition in Eq.~\eqref{eq:cosine_deltaL} is a direct manifestation of the classical structure of almost-periodic functions in the sense of Harald Bohr~\cite{Bohr2018}. The unitary evolution generated by a time-independent Hamiltonian produces phases $e^{-i\omega\subtiny{0}{0}{nm}t}$ associated with the discrete set of Bohr frequencies $\omega\subtiny{0}{0}{nm}$, so any quantity built from overlaps with the initial state admits a representation as a real superposition of harmonic components with exactly those frequencies. This almost-periodic structure shows that the relaxation of $\delta\!\leakz(t)$ is governed by the number of active incommensurate Bohr frequencies. A small number of participating gaps leads to quasi-periodic dynamics and persistent fluctuations, whereas a large number of incommensurate frequencies enhances phase mixing and drives equilibration.

{It is natural to compare Theorem~\ref{th:leakage_bounds} with the population-equilibration bound of Meier et al.~\cite{meier2025emergence}, which applies Reimann-type arguments to the outcome probabilities $p\subtiny{0}{0}{i}(t)=\Tr[A\subtiny{0}{0}{i}\rho(t)]$ of an arbitrary $s$-outcome POVM $\{A\subtiny{0}{0}{i}\}$. Their general result controls the distance between the outcome distribution $\vec p(t)$ and its equilibrium value, with a scaling proportional to $\sqrt{s f(\varepsilon,T)/\deff}$. The LFF corresponds to the particular binary projective measurement $\{P,Q=\id-P\}$, where $P$ is the support projector of the initial state. In this case, $p_Q(t)=\leakP(t)$ and $\delta p_P(t)=-\delta p_Q(t)$, so that $\|\vec p(t)-\vec p^\infty\|\subtiny{0}{0}{2}\suptiny{0}{0}{2}
=2|\delta\leakP(t)|\suptiny{0}{0}{2}.$ The mean-square estimate underlying the general POVM bound therefore gives
\begin{align}
\label{eq:comparacao_meier_1}
\avgT{|\delta\leakP(t)|\suptiny{0}{0}{2}}\leq
\frac{f(\varepsilon,T)}{2}
\left[
\Tr(P\eqst\suptiny{0}{0}{2})
+\Tr(Q\eqst\suptiny{0}{0}{2})
\right]=\frac{f(\varepsilon,T)}{2\deff},
\end{align}
where the two contributions are
$\Tr(P\eqst\suptiny{0}{0}{2})
=\sum\subtiny{0}{0}{n}\rho\subtiny{0}{0}{nn}\suptiny{0}{0}{2}P\subtiny{0}{0}{nn}$
and
$\Tr(Q\eqst\suptiny{0}{0}{2})
=\sum\subtiny{0}{0}{n}\rho\subtiny{0}{0}{nn}\suptiny{0}{0}{2}(1-P\subtiny{0}{0}{nn})$.
By contrast, Eq.~\eqref{eq:leak_unconditional_deff_bound} yields a support-dependent estimate
$\avgT{|\delta\leakP(t)|\suptiny{0}{0}{2}}\leq\frac{f(\varepsilon,T)}{4\deff}$.
Theorem~\ref{th:leakage_bounds} follows from exploiting the geometry of the binary partition defined by the complementary projectors $P$ and $Q$ relative to the energy eigenbasis, as encoded in the factors $P\subtiny{0}{0}{nn}(1-P\subtiny{0}{0}{nn})$.}


\begin{figure*}[t!]
    \centering
    \begin{subfigure}[b]{0.48\textwidth}
        \centering
        \includegraphics[width=\linewidth]{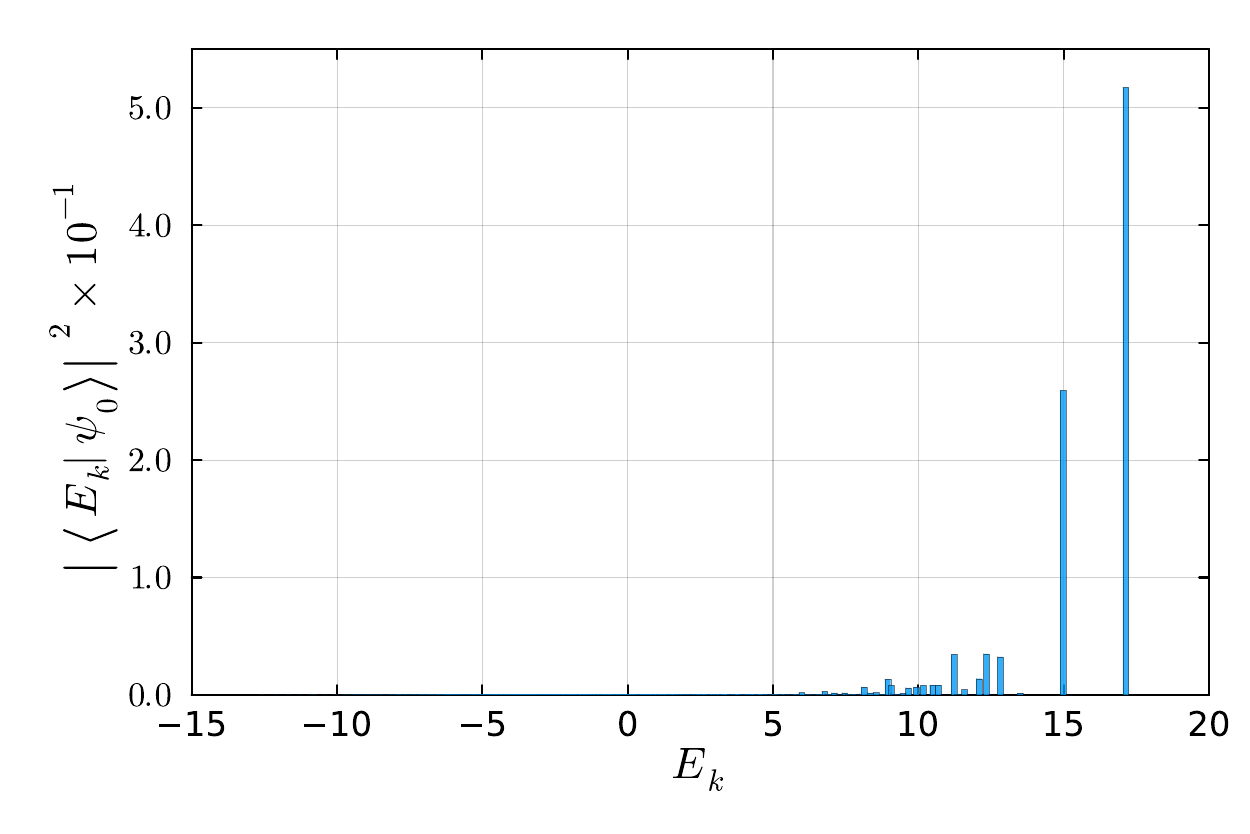}
        \caption{Initial state $\up$.}
        \label{fig:fig3a}
    \end{subfigure}
    \begin{subfigure}[b]{0.48\textwidth}
        \centering
        \includegraphics[width=\linewidth]{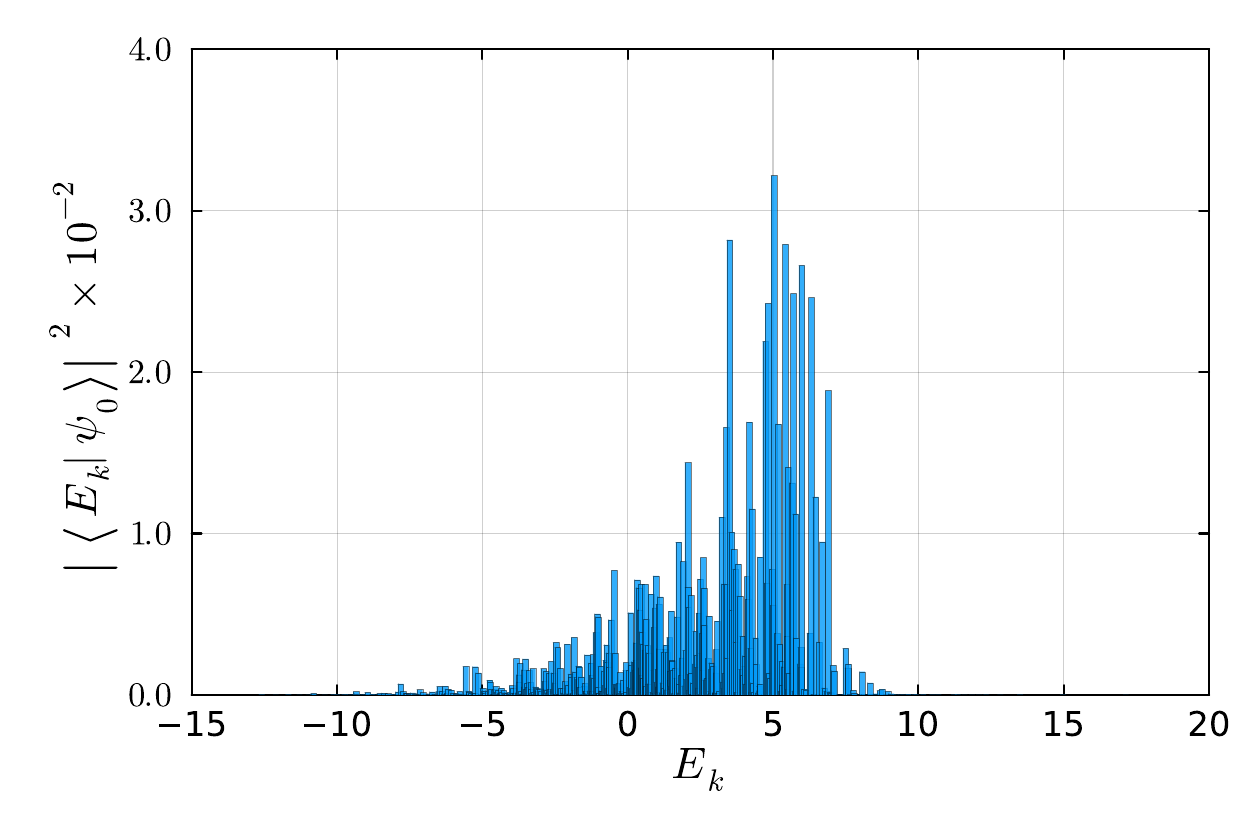}
        \caption{Initial state $\dw$.}
        \label{fig:fig3b}
    \end{subfigure}
    \caption{\justifying Spectral occupations
        $p\subtiny{0}{0}{k}=\abs{\ip{E\subtiny{0}{0}{k}}{\psi\subtiny{0}{0}{0}}}\suptiny{0}{0}{2}$
        as functions of the Hamiltonian eigenenergies $E\subtiny{0}{0}{k}$ for the $N=10$
        spin chain. Panels (a) and (b) correspond, respectively, to the initial
        states $\up$ and $\dw$. The $\up$ state is concentrated on a small number
        of energy eigenstates, resulting in a small effective dimension
        $\deff\simeq 2.95$. In contrast, the $\dw$ state occupies a much broader
        portion of the spectrum and has $\deff\simeq 93.74$. These different
        occupation profiles underlie the distinct leakage dynamics displayed
        in Fig.~\ref{fig:fig2}.}
    \label{fig:fig3}
\end{figure*}

{Fig.~\ref{fig:fig2} illustrates the spectral-dynamical contrast for pure and mixed initial states. In panel~(a), the pure $\up$ state exhibits pronounced, persistent oscillations, consistent with its concentrated energy distribution and reduced effective dimension, whereas for $\dw$ and $\mypm$, LFF remain closer to their equilibrium values after the initial transient. Panel~(b) extends this comparison to mixtures of $\up$ and $\dw$: increasing the weight of $\up$ lowers LFF and enhances temporal fluctuations, despite the common rank-two initial support. Together, these results illustrate how LFF dynamics probe equilibration and memory effects through their dependence on spectral structure and initial-state composition. Theorem~\ref{th:leakage_bounds} provides a quantitative characterization of these fluctuations by bounding their time-averaged mean-square magnitude in terms of the energy populations, the geometry of the initial support, and the spectral factor $f(\varepsilon,T)$.}

\subsection{Leakage Fidelity Function Fluctuations}

LFF equilibration should not be formulated only as a statement about the smallness of the mean deviation. One also needs a concentration statement that controls the likelihood that the function exhibits large fluctuations around its equilibrium value. This leads to two complementary probabilistic formulations: one in which the initial state is fixed and the time is sampled uniformly from the interval $[0,T]$, and another in which the time window is fixed and the initial state is sampled from the Haar ensemble. Theorem~\ref{thm:leakage_concentration_dir} makes this distinction explicit by combining the finite-time spectral factor with the effective dimension of the initial state and the Haar-induced distribution of energy populations.  

\begin{theorem}[Leakage concentration from variance bound]\label{thm:leakage_concentration_dir}
{Let $[0,T]$ be a fixed time window and let the pure initial state be written as
$\ket{\psi\subtiny{0}{0}{0}}=\sum\subtiny{0}{0}{i}c\subtiny{0}{0}{i}\ket{E\subtiny{0}{0}{i}}$,
where $\{\ket{E\subtiny{0}{0}{i}}\}\subtiny{0}{0}{i}$ denotes the energy eigenbasis and
$p\subtiny{0}{0}{i}=|c\subtiny{0}{0}{i}|\suptiny{0}{0}{2}$
are the corresponding energy populations.
When the initial state is sampled from the Haar ensemble, as in part (B), these populations follow the symmetric $\mathrm{Dirichlet}(1,\ldots,1)$ distribution.
Then, for the spectral factor $f(\varepsilon,T)$ defined in Eq.~\eqref{eq:spectral_factor}, the following bounds hold:}\\
\noindent\emph{(A) Time probability (fixed state).}
For any fixed state $\ket{\psi\subtiny{0}{0}{0}}$, for any $\epsilon>0$ and $t\sim\mathrm{Unif}[0,T]$,
\begin{equation}\label{eq:A-time}
\Pr_t\!\Big(\,|\leakz(t)-\leakz\suptiny{0}{0}{\infty}|\ \ge\ \epsilon\,\Big)
\ \le\
\frac{f(\varepsilon,T)}{\epsilon\suptiny{0}{0}{2}}\,
\Bigl(1-\frac{1}{d\subtiny{0}{0}{E}}\Bigr)\,
\frac{1}{\deff\suptiny{0}{1}{2}}.
\end{equation}\\
\noindent\emph{(B) Ensemble probability (random state).}
Let $\ket{\psi\subtiny{0}{0}{0}}$ be distributed according to the Haar measure. For any fixed averaging time $T$ and any $\eta>0$, 
\begin{align}\label{eq:B-ensemble}
&\Pr_{\ket{\psi\subtiny{0}{0}{0}}}\!\Big(\,\Big\langle\big|\leakz(t)-\leakz\suptiny{0}{0}{\infty}\big|\suptiny{0}{0}{2}\Big\rangle\subtiny{0}{0}{T}\ \ge\ \eta\,\Big)
\le F(\eta,T,d\subtiny{0}{0}{E}),
\end{align}
where $F(\eta,T,d\subtiny{0}{0}{E}) = \frac{f(\varepsilon,T)}{\eta}\,
\Bigl(1-\frac{1}{d\subtiny{0}{0}{E}}\Bigr)g(d\subtiny{0}{0}{E})$, and 
$g(d_E)=\frac{4(d_E+5)}{(d_E+1)(d_E+2)(d_E+3)}$.
\end{theorem} 
Part 2(A) provides a \emph{local-in-time} control of the LFF dynamics: for any fixed initial state, it bounds the probability that the instantaneous LFF deviates significantly from its long-time value. (i)~States with small effective dimension, as a coherent superposition supported on only a few energy levels, may display large and persistent oscillations of $\leakz(t)$, and the bound reflects this characteristic. In contrast, (ii)~states with large effective dimension undergo strong dephasing across many frequencies, causing LFF to remain close to $\leakz\suptiny{0}{0}{\infty}$ for most times in $[0,T]$. 

Part 2(B) provides a \emph{global-in-ensemble} control: (iii) for Haar-random initial states the time-averaged deviation $\avgT{\,|\leakz(t)-\leakz\suptiny{0}{0}{\infty}|\suptiny{0}{1}{2}\,}$ is strongly suppressed for almost all initial conditions. Rare, highly structured states may exhibit atypical dynamics, including identically vanishing leakage or pronounced temporal fluctuations, without contradicting the Haar-ensemble estimate.

The formal proof of the theorem can be found in Appendix~\ref{appsec:General_moments_dirichlet}. Theorem~\ref{thm:leakage_concentration_dir} \textit{does not} imply ergodicity or ergodicity breaking: part (B) constrains temporal fluctuations but does not enforce the equality of time and ensemble averages, which may differ for specially prepared or low--effective-dimension states. Instead, parts (A) and (B) jointly show that LFF equilibration is temporally stable for high--effective-dimension states and typical within the Haar ensemble. The results of Theorem~\ref{thm:leakage_concentration_dir} motivated us to study the behavior of the LFF and its variance in Section \ref{sec:leak_HUB} for the family of orthonormal states forming a \textit{Hamiltonian-unbiased basis} (HUB) with respect to the energy eigenbasis, of dimension $d\subtiny{0}{0}{\mathrm{E}}$ (HUB states) and of dimension $K$ (K-HUB states).

{The extension of the local-in-time control (part 2(A)) of Theorem~\ref{thm:leakage_concentration_dir} to arbitrary mixed initial states is established in Appendix~\ref{app:leakage_mixed_concentration}.
An analogous extension of the ensemble control 2(B) requires specifying a probability measure over mixed states, since different ensembles generally lead to different statistical averages. Unlike the pure-state case, no unique unitarily invariant measure exists on the full space of mixed quantum states. Thus, we briefly outline the corresponding procedure and relevant references in the same appendix, leaving ensemble-dependent concentration bounds outside the scope of this work.} 

In the following sections, we show that the effective dimension $\deff$ is not the only important aspect of equilibration on average, showing how its description depends on observable variation bounds, in the spirit of Refs.~\cite{Reimann2012,Short_2012}.

\begin{figure*}[t]
    \centering
    \begin{subfigure}[b]{0.32\textwidth}
        \centering
        \includegraphics[width=\textwidth]{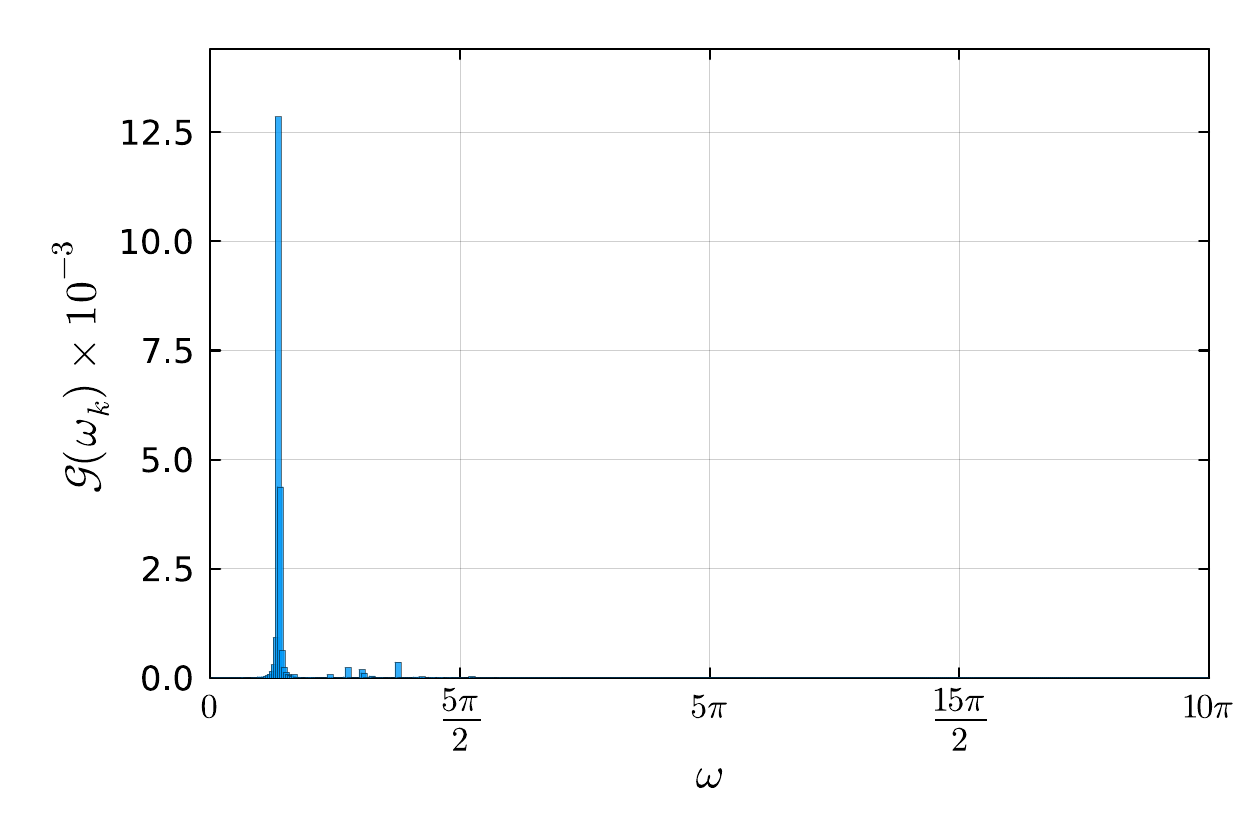}
        \caption{Initial state $\up$.}
        \label{fig:fig4a}
    \end{subfigure}
        \hfill
    \begin{subfigure}[b]{0.32\textwidth}
        \centering
        \includegraphics[width=\textwidth]{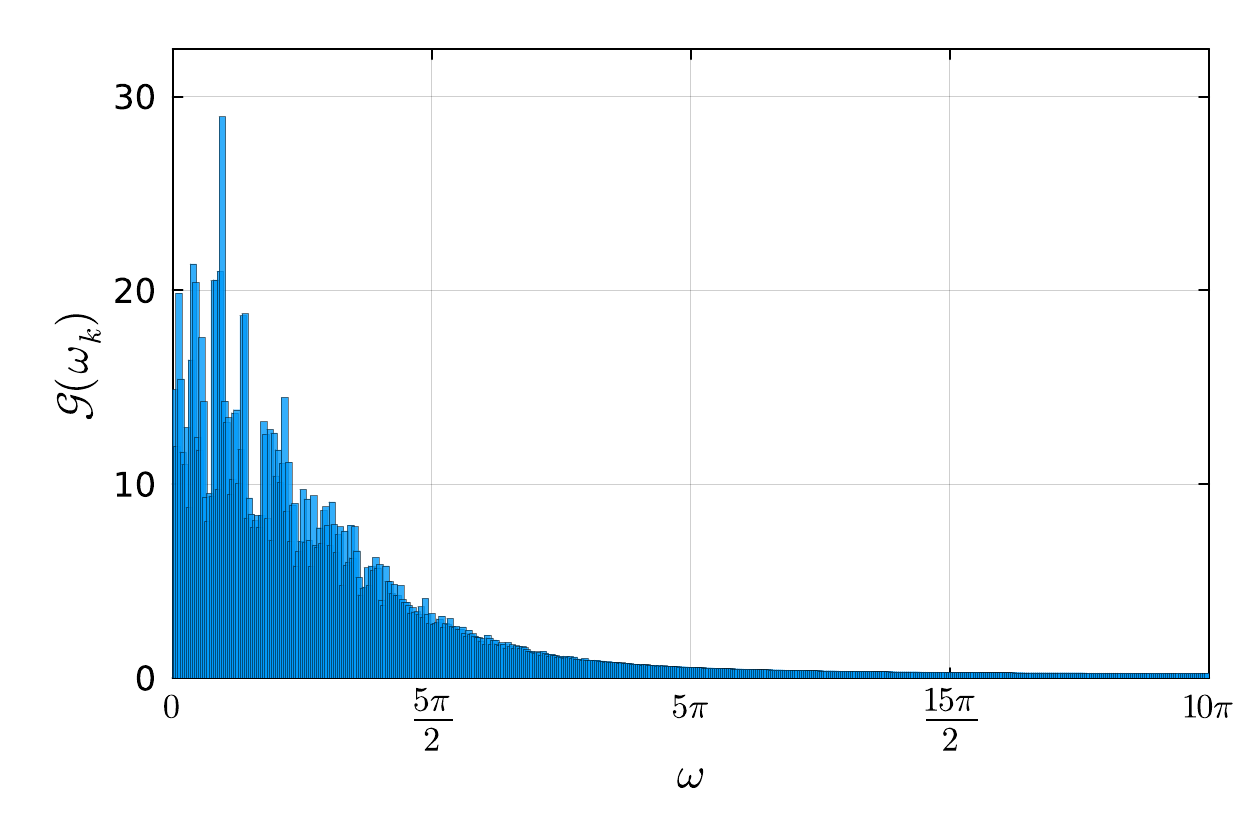}
        \caption{Initial state $\dw$.}
        \label{fig:fig4b}
    \end{subfigure}
    \hfill
    \begin{subfigure}[b]{0.32\textwidth}
        \centering
        \includegraphics[width=\textwidth]{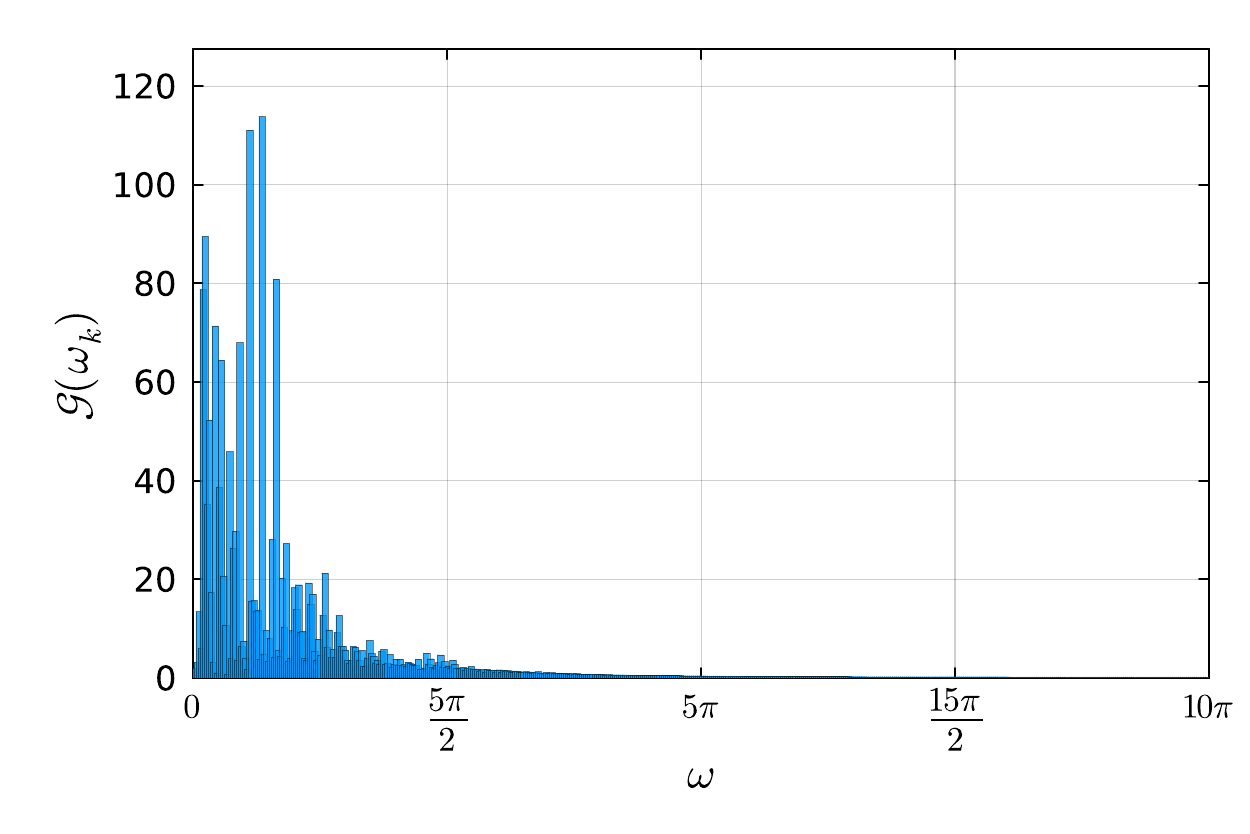}
        \caption{Initial state $\mypm$.}
        \label{fig:fig4c}
    \end{subfigure}
    \caption{\justifying  The spectral power density (SPD): $\mathcal{G}(\omega_{k})    =|\delta\!\leakz(\omega\subtiny{0}{0}{k})|\suptiny{0}{0}{2},$
            of the centered signal
            $\delta\!\leakz(t)=\leakz(t)-\leakz\suptiny{0}{0}{\infty}$
            for the initial states (a) $\up$, (b) $\dw$, and (c) $\mypm$. The discrete angular frequencies are
            $\omega\subtiny{0}{0}{k}=2\pi k/(N_t\Delta t)$, where $N_t$ is the number of sampled times and $\Delta t$ is the sampling interval. Each resolved peak represents a Fourier-frequency bin containing one or more active Bohr
            gaps, $\omega\subtiny{0}{0}{nm}=E\subtiny{0}{0}{n}-E\subtiny{0}{0}{m}$, weighted by the corresponding pair
            populations $p\subtiny{0}{0}{n}p\subtiny{0}{0}{m}$. The concentrated spectrum of the $\up$ state reveals dynamics dominated by a few frequency channels, whereas the broader spectra of the other states indicate stronger phase mixing. Because the signal is observed over a finite time window, the Bohr gaps are resolved only up to the Fourier resolution 
            $\Delta\omega=2\pi/(N_t\Delta t)$.}
            \label{fig:fig4}
\end{figure*}

\section{Spectral Analysis} \label{sec:leak_spec_power}

{In this section, we address the question: {\it what spectral properties determine whether a closed quantum system exhibits equilibration behavior on average?} We begin with a static diagnostic: the projection of the initial state onto the Hamiltonian eigenbasis. This information is illustrated in Fig.~\ref{fig:fig3}, which shows the spectral occupation probabilities $\abs{\ip{E\subtiny{0}{0}{k}}{\psi\subtiny{0}{0}{0}}}\suptiny{0}{0}{2}$ for the two fully polarized initial states $\ket{\uparrow\uparrow\cdots\uparrow}$ ($\up$) and $\ket{\downarrow\downarrow\cdots\downarrow}$ ($\dw$), for the The contrast between these two profiles already indicates different dynamical regimes. The state $\up$ is strongly localized in the energy eigenbasis, with most of its weight concentrated on a small subset of eigenstates, leading to a small effective dimension $\deff\ll d\subtiny{0}{0}{E}$. In contrast, the state $\dw$ is distributed over a much larger portion of the spectrum, yielding a substantially larger effective dimension. However, the occupation profile in Fig.~\ref{fig:fig3} is only a static characterization of the initial state. It specifies which eigenstates are populated, but not which energy gaps are effectively activated in the time-dependent LFF signal. In Sec.~\ref{subsec:SPD_LFF}, we formulate this frequency-domain description in terms of the spectral power density of the LFF fluctuation signal. In Sec.~\ref{sec:d_spec} we introduce the \textit{spectral dimension} $\deffspec$, which quantifies delocalization in the frequency domain by means of the effective number of positive Bohr-frequency channels contributing to the fluctuation power of the LFF. In Sec.~\ref{sec:Hpow} we introduce the \textit{power entropy} $\Hpow$, which characterizes the degree of concentration or dispersion of the fluctuation power among the available Bohr-frequency channels.}

\subsection{Spectral Power Density of LFF}
\label{subsec:SPD_LFF}

The temporal fluctuations of the LFF provide information about the spectral structure of the underlying unitary dynamics. For an arbitrary initial state $\rho\subtiny{0}{0}{0}$, with support projector $P$, the real-phasor decomposition of the centered signal presented in Eq.~\eqref{eq:leak_mixed_phasor}, can be expressed in general as
\begin{equation}
\delta\leakP(t)
=
-\sum\subtiny{0}{0}{n<m}
\left[
z\subtiny{0}{0}{nm}
e\suptiny{0}{0}{i\omega\subtiny{0}{0}{nm}t}
+
z\subtiny{0}{0}{nm}\suptiny{0}{0}{*}
e\suptiny{0}{0}{-i\omega\subtiny{0}{0}{nm}t}
\right],
\end{equation}
where
$z\subtiny{0}{0}{nm}
=\rho\subtiny{0}{0}{nm}P\subtiny{0}{0}{mn}$
and
$\omega\subtiny{0}{0}{nm}
=E\subtiny{0}{0}{m}-E\subtiny{0}{0}{n}>0$
for $n<m$.
The signal is real and has zero long-time
average. Its spectral structure is determined
by the active Bohr frequencies and the
complex coefficients
$z\subtiny{0}{0}{nm}$, whose moduli
$b\subtiny{0}{0}{nm}
=|z\subtiny{0}{0}{nm}|$
specify the amplitudes of the corresponding
oscillatory contributions.

When many oscillatory contributions are
present, their interference can produce
a complex temporal signal. It is, therefore,
instructive to examine its Fourier transform,
which identifies the active frequency
channels and their relative spectral
contributions. The spectral amplitudes
reveal which energy gaps contribute
most strongly to the LFF fluctuations,
while the values and spacings of these
frequencies provide additional information
about dephasing, beating, recurrences,
and the characteristic timescales
of the dynamics. 
Since $\delta\leakP(t)$ is real, its Fourier
transform satisfies
$\delta\leakP(-\omega)
=\delta\leakP(\omega)\suptiny{0}{0}{*}$.
Thus, the negative-frequency components
contain no independent information, and
we restrict the spectral analysis to
positive Bohr frequencies. Using the Fourier convention $\delta\leakP(\omega)
= \int\subtiny{0}{0}{-\infty}\suptiny{0}{0}{\infty} \delta\leakP(t) e\suptiny{0}{0}{-i\omega t}\,dt,$ we obtain, for $\omega>0$,
\begin{equation}
\label{eq:fourier_deltaL}
\delta\leakP(\omega)
=
-2\pi
\sum\subtiny{0}{0}{n<m}
z\subtiny{0}{0}{nm}
\delta(\omega-\omega\subtiny{0}{0}{nm}),
\end{equation}
where
$z\subtiny{0}{0}{nm}
=\rho\subtiny{0}{0}{nm}P\subtiny{0}{0}{mn}$.
Each positive Bohr frequency therefore
contributes a Fourier peak whose complex
amplitude determines its spectral power.

{Considering the energy spectrum composed of discrete energy gaps, let $\omega\subtiny{0}{0}{k}$ denote the $N\subtiny{0}{0}{\omega}$ distinct positive Bohr frequencies of the Hamiltonian. For each frequency, we define the set
\begin{equation}
\label{eq:Gamma_alpha}
\Gamma\subtiny{0}{0}{k} := \left\{(n,m):n<m,\; E\subtiny{0}{0}{m}-E\subtiny{0}{0}{n} =\omega\subtiny{0}{0}{k}\right\}.
\end{equation}
Its cardinality $g\subtiny{0}{0}{k}
=|\Gamma\subtiny{0}{0}{k}|$ is the multiplicity of the corresponding positive energy gap. The complex amplitude associated with this frequency is the coherent sum
\begin{equation}
C\subtiny{0}{0}{k} = \sum\subtiny{0}{0}{(n,m)\in\Gamma\subtiny{0}{0}{k}}z\subtiny{0}{0}{nm},
\end{equation}
which combines all transition amplitudes associated with the positive Bohr frequency $\omega\subtiny{0}{0}{k}$.}

We define the \emph{Spectral Power Density} (SPD) of the LFF as the fluctuation power carried by each distinct positive Bohr-frequency channel. The SPD therefore provides a frequency-domain description of how the power of temporal fluctuations is distributed among the active gaps: 
\begin{equation}
\label{eq:SPD_definition_degenerate}
\spd(\omega\subtiny{0}{0}{k}) = |C\subtiny{0}{0}{k}|\suptiny{0}{0}{2} = \left| \sum\subtiny{0}{0}{
(n,m)\in\Gamma\subtiny{0}{0}{k}} \rho\subtiny{0}{0}{nm}
P\subtiny{0}{0}{mn}\right|\suptiny{0}{0}{2}.
\end{equation}
The coherent aggregation is essential: contributions with the same frequency cannot be treated as independent
Fourier channels. For mixed initial states, their complex phases may interfere constructively or destructively, 
enhancing or suppressing the power of the corresponding spectral peak. In particular, an active energy-gap
sector may have vanishing total spectral power if its constituent amplitudes cancel exactly.

In the presence of degenerate energy gaps, the relevant frequency-domain contribution is not associated with a single pair $(n,m)$, but with the coherent sum of all pair contributions sharing the same positive Bohr frequency. For nondegenerate positive gaps, each frequency is associated with a single unordered pair of energy
eigenstates. Consequently,
Eq.~\eqref{eq:SPD_definition_degenerate}
reduces to
\begin{equation}
\label{eq:SPD_definition_nondegenerate}
\spd(\omega\subtiny{0}{0}{nm}) = b\subtiny{0}{0}{nm}\suptiny{0}{0}{2} = |\rho\subtiny{0}{0}{nm} P\subtiny{0}{0}{mn}|\suptiny{0}{0}{2}, \qquad n<m.
\end{equation}

The use of spectral power, rather than spectral amplitude, is directly related to the characterization of equilibration
through quadratic temporal fluctuations. For distinct positive Bohr-frequency channels, the orthogonality of the
Fourier modes under long-time averaging gives
\begin{equation}
\avginfty{|\delta\leakP(t)|\suptiny{0}{0}{2}}
=
2\sum\subtiny{0}{0}{k}
|C\subtiny{0}{0}{k}|\suptiny{0}{0}{2}
=
2\sum\subtiny{0}{0}{k}
\spd(\omega\subtiny{0}{0}{k}).
\end{equation}
This identity is the discrete-frequency counterpart of Parseval's theorem: the mean-square fluctuation in the time domain is determined by the total spectral power. The factor of two accounts for the contributions from positive and negative frequencies \cite{oppenheim1996signals,bracewell1999fourier}. Therefore, the SPD identifies not only the frequencies present in the dynamics but also their individual contributions to the long-time fluctuation variance. To characterize the distribution of spectral power over different frequency ranges, we introduce the
\emph{cumulative spectral weight}
\begin{equation}
\label{eq:cumulative_weight_correct}
W\subtiny{0}{0}{\leak}(\Omega)
:=
\sum\subtiny{0}{0}{
\omega\subtiny{0}{0}{k}\leq\Omega}
\spd(\omega\subtiny{0}{0}{k}).
\end{equation}
This quantity measures the total fluctuation power carried by the active positive Bohr frequencies up to the cutoff $\Omega$. Its dependence on $\Omega$ reveals the frequency ranges in which the fluctuation power is concentrated.

For pure initial states,
$\rho\subtiny{0}{0}{0}
=P=\ketbra{\psi\subtiny{0}{0}{0}}
{\psi\subtiny{0}{0}{0}}$,
these coefficients reduce to
$z\subtiny{0}{0}{nm}
=p\subtiny{0}{0}{n}p\subtiny{0}{0}{m}$,
where
$p\subtiny{0}{0}{n}
=|\braket{E\subtiny{0}{0}{n}}
{\psi\subtiny{0}{0}{0}}|\suptiny{0}{0}{2}$. Each pairwise population factor
$p\subtiny{0}{0}{n}p\subtiny{0}{0}{m}$
determines the amplitude of the oscillation
associated with the corresponding energy gap.
Thus, the relevant dynamical information
is carried not only by the occupied energy
eigenstates but also by the pairs of
eigenstates generating the active Bohr
frequencies. In this case Eq.~\eqref{eq:fourier_deltaL} reduces to
\begin{equation}
\delta\leakz(\omega)
=
-2\pi
\sum\subtiny{0}{0}{n<m}
p\subtiny{0}{0}{n}p\subtiny{0}{0}{m}
\delta(\omega-\omega\subtiny{0}{0}{nm}),
\qquad \omega>0.
\end{equation}
The spectral coefficients are therefore
real, nonnegative, and entirely determined
by the initial energy populations.
This recovers the direct connection
between the spectral organization
of the LFF and the survival probability.
The distinction between pure and mixed states thus becomes particularly relevant when the Hamiltonian possesses degenerate energy gaps. For pure initial states with nondegenerate positive energy gaps the SPD reduces to
\begin{equation}
\spd(\omega\subtiny{0}{0}{nm})
=
p\subtiny{0}{0}{n}\suptiny{0}{0}{2}
p\subtiny{0}{0}{m}\suptiny{0}{0}{2},
\qquad n<m,
\end{equation}
recovering the original pair-resolved spectral-power expression employed in the analysis of the Hamiltonian considered in this work.

In numerical simulations, the signal is
sampled at discrete times
$t\subtiny{0}{0}{j}=j\Delta t$,
with $j=0,\ldots,N-1$.
The discrete Fourier transform of
the LFF fluctuations is
\begin{equation}
\label{eq:LeakFFT}
\delta\leakP(\omega\subtiny{0}{0}{k})
=
\sum\subtiny{0}{0}{j=0}\suptiny{0}{0}{N-1}
\delta\leakP(t\subtiny{0}{0}{j})
e\suptiny{0}{0}{-i\omega\subtiny{0}{0}{k}
t\subtiny{0}{0}{j}},
\end{equation}
where the discrete angular frequencies are
\begin{equation}
\label{eq:omegak}
\omega\subtiny{0}{0}{k}
=
\frac{2\pi k}{N\Delta t},
\qquad
k=0,\ldots,N-1.
\end{equation}
For a one-sided spectrum, the Nyquist interval is $0\leq\omega\leq\pi/\Delta t$. When the horizontal axis is expressed
in dimensionless units, $\varpi=\omega\Delta t$, this becomes
$\varpi\in[0,\pi]$. The sampling interval must be sufficiently
small to avoid aliasing of the relevant Bohr frequencies, while the total observation time $N\Delta t$ determines the Fourier frequency spacing, $2\pi/(N\Delta t)$.
Consequently, the numerical frequency grid need not coincide exactly with the Bohr frequencies of the Hamiltonian. Finite observation times can broaden spectral peaks and prevent nearby frequencies from being resolved. These finite-time effects must be distinguished from exact degeneracies of the energy gaps.
In numerical spectra, a resolved peak may originate from a single dominant transition, from several nearby gaps falling within the same frequency-resolution window, or from
an exactly degenerate gap sector. These possibilities may produce similar numerical spectral features, although their underlying mechanisms are different. Distinguishing them
requires information about the energy-gap structure and the
available frequency resolution.

{In the discrete setting used in the numerical analysis, this yields the spectral power density defined in Eq.~\eqref{eq:SPD_definition_degenerate} and shown in Fig.~\ref{fig:fig4}, which is precisely the frequency-domain representation. It identifies the Bohr-frequency channels that contribute significantly to fluctuations and characterizes whether the corresponding spectral power is localized to a few dominant frequencies or spread across many dynamically active modes. When the initial spectral distribution is concentrated, only a sparse set of gaps contributes appreciably, producing quasi-periodic behavior and persistent recurrences. When the distribution is broad, many incommensurate gaps contribute, enhancing phase mixing and suppressing temporal fluctuations. Therefore, Fig.~\ref{fig:fig3} should be understood as a static energy-space diagnostic, whereas Fig.~\ref{fig:fig4} provides the corresponding dynamical gap-space diagnostic. The former determines the populated energy levels and hence $\deff$, while the latter resolves the frequency channels activated in the LFF signal. This distinction justifies the need for two complementary descriptors: $\deffdyn$, associated with the effective number of time-domain cosine modes, and $\deffspec$, associated with the distribution of spectral power.}

\subsection{Spectral dimension}\label{sec:d_spec}

The SPD introduced in Eq.~\eqref{eq:SPD_definition_degenerate} determines the distribution of fluctuation power among the positive Bohr-frequency
channels. Normalizing these nonnegative spectral weights gives a probability distribution over the active frequencies,
\begin{equation}
\label{eq:pow_distribution_LFF}
p\subtiny{0}{0}{\leak} (\omega\subtiny{0}{0}{k}) = \frac{
\spd(\omega\subtiny{0}{0}{k})}{\displaystyle\sum\subtiny{0}{0}{j}\spd(\omega\subtiny{0}{0}{j})},
\end{equation}
provided that the total spectral power is nonzero. Each normalized weight represents the fraction of the total positive-frequency fluctuation power carried by the
corresponding Bohr-frequency channel. The distribution is normalized, $\sum\subtiny{0}{0}{k} p\subtiny{0}{0}{\leak}
(\omega\subtiny{0}{0}{k})=1$, and provides a statistical description of the spectral organization of the LFF fluctuations.
Under the assumption of nondegenerate positive energy gaps, each pair $n<m$ defines a unique frequency channel. Using Eq.~\eqref{eq:SPD_definition_nondegenerate},
we obtain
\begin{equation}
p\subtiny{0}{0}{\leak} (\omega\subtiny{0}{0}{nm})
= \frac{b\subtiny{0}{0}{nm}\suptiny{0}{0}{2}}{
S\subtiny{0}{0}{2}} = \frac{ |\rho\subtiny{0}{0}{nm}
P\subtiny{0}{0}{mn}|\suptiny{0}{0}{2}}{\displaystyle
\sum\subtiny{0}{0}{i<j} |\rho\subtiny{0}{0}{ij} P\subtiny{0}{0}{ji}|\suptiny{0}{0}{2}}, \qquad n<m,
\end{equation}
where $S_2\geq 0$ is introduced in Eq.~\eqref{eq:spectral_moments}.

In analogy with the standard effective dimension, in Eq.~\eqref{eq:deff},
$\deff
=(\sum\subtiny{0}{0}{n}
\rho\subtiny{0}{0}{nn}\suptiny{0}{0}{2})
\suptiny{0}{0}{-1}$, with $\rho\subtiny{0}{0}{nn}=\bra{E\subtiny{0}{0}{n}}\rho\subtiny{0}{0}{0}\ket{E\subtiny{0}{0}{n}}$, 
we define the \emph{spectral dimension} $(\deffspec)$ as the inverse participation ratio of the normalized SPD,
\begin{equation}
\label{eq:def_deffspec}
\deffspec
=
\frac{1}{
\displaystyle
\sum\subtiny{0}{0}{k}
p\subtiny{0}{0}{\leak}
(\omega\subtiny{0}{0}{k})
\suptiny{0}{0}{2}
}.
\end{equation}
While $\deff$ measures the distribution of the initial energy populations, $\deffspec$ quantifies the effective number of positive Bohr-frequency channels contributing to the fluctuation power of the LFF across the active energy gaps. The two quantities therefore capture distinct aspects of the initial state and its dynamics.
The normalization is important because $\deffspec$ characterizes the distribution of fluctuation power among the active Bohr frequencies, independently of its overall magnitude. Consequently, two initial states may have the same total spectral power but substantially different spectral dimensions. For nondegenerate positive gaps, the normalized SPD is constructed from the squared amplitudes, yielding $\deffspec=S\subtiny{0}{0}{2}\suptiny{0}{0}{2}/S\subtiny{0}{0}{4}$, with the moments defined in Eq.~\eqref{eq:spectral_moments}. For degenerate positive gaps, the corresponding inverse participation ratio is instead constructed from the coherently aggregated spectral powers $|C\subtiny{0}{0}{k}|\suptiny{0}{0}{2}$. 

Large values of $\deffspec$ indicate that the fluctuation power is broadly distributed among many distinct Bohr frequencies, with several channels contributing comparably. Such spectral organization can favor dephasing through the interference of multiple oscillatory components. Conversely, small values indicate that a few dominant frequencies carry most of the spectral power, potentially producing more regular oscillations, pronounced beating patterns, and recurrences. Nevertheless, $\deffspec$ alone does not determine the recurrence structure or guarantee the suppression of revivals, since the frequency values, their separations, and possible commensurabilities also influence the temporal interference.  

For pure initial states, $\deffspec$ can be expressed entirely in terms of the energy populations $p\subtiny{0}{0}{n}=|\braket{E\subtiny{0}{0}{n}}{\psi\subtiny{0}{0}{0}}|\suptiny{0}{0}{2}$, given by 
\begin{equation}
p\subtiny{0}{0}{\leak}
(\omega\subtiny{0}{0}{nm})
=
\frac{
p\subtiny{0}{0}{n}\suptiny{0}{0}{2}
p\subtiny{0}{0}{m}\suptiny{0}{0}{2}
}{
\displaystyle
\sum\subtiny{0}{0}{i<j}
p\subtiny{0}{0}{i}\suptiny{0}{0}{2}
p\subtiny{0}{0}{j}\suptiny{0}{0}{2}
},
\qquad n<m.
\end{equation}
For nondegenerate positive gaps,
substituting the pure-state weights into  Eq.~\eqref{eq:def_deffspec} and expressing the pairwise sums
through population moments, as demonstrated in
Proposition~\ref{prop:deffspec}
of Appendix~\ref{appsec:deffdyn},
gives
\begin{equation}
\label{eq:deffspec_exact_general_moments}
\deffspec
=
\frac{
\left(
\frac{1}{\deff\suptiny{0}{0}{2}}
-
\sum\subtiny{0}{0}{n}
p\subtiny{0}{0}{n}\suptiny{0}{0}{4}
\right)\suptiny{0}{0}{2}
}{
2\left[
\left(
\sum\subtiny{0}{0}{n}
p\subtiny{0}{0}{n}\suptiny{0}{0}{4}
\right)\suptiny{0}{0}{2}
-
\sum\subtiny{0}{0}{n}
p\subtiny{0}{0}{n}\suptiny{0}{0}{8}
\right]
}.
\end{equation}
Whereas the standard effective dimension $\deff=(\sum\subtiny{0}{0}{n}p\subtiny{0}{0}{n}\suptiny{0}{0}{2})
\suptiny{0}{0}{-1}$ depends only on the second population
moment, $\deffspec$ also depends on the fourth and eighth moments. Consequently, two pure initial states
with the same $\deff$ may exhibit different spectral-power distributions and spectral dimensions. Thus, $\deffspec$ refines the information provided by $\deff$ by incorporating
higher-order statistical properties that determine the relative weights of pairwise population products.

\subsection{Shannon power entropy}\label{sec:Hpow}

The normalized spectral-power distribution $p\subtiny{0}{0}{\leak}(\omega\subtiny{0}{0}{k})$ provides a statistical description of the relative contributions of the active Bohr frequencies to the LFF fluctuations.
Accordingly, it measures the spectral participation of that channel in the dynamical redistribution of probability outside the initial support.
From this distribution, we define the discrete \emph{Shannon power entropy}
\begin{equation}
\label{eq:Shannon_power_entropy}
\Hpow = - \sum\subtiny{0}{0}{k} p\subtiny{0}{0}{\leak}
(\omega\subtiny{0}{0}{k}) \ln
p\subtiny{0}{0}{\leak} (\omega\subtiny{0}{0}{k}),
\end{equation}
with the convention $0\ln0=0$. The quantity $\Hpow$ characterizes the degree of concentration or dispersion of the fluctuation power among the available Bohr-frequency channels. 

If $K$ is the number of active positive-frequency channels,
then $0\leq\Hpow\leq\ln K$. The minimum occurs when all
spectral power is concentrated in a single channel, whereas
the maximum is attained when the power is uniformly
distributed among the $K$ active frequencies.

A small value of $\Hpow$ indicates that a few dominant
frequencies carry most of the fluctuation power. The resulting temporal dynamics may then exhibit comparatively simple oscillatory structures, including pronounced beating
patterns and recurrences. A large value of $\Hpow$ signals the participation of many frequency channels with more evenly distributed spectral power. Such spectral organization provides the conditions for more extensive
phase mixing and complex temporal interference.
However, as for $\deffspec$, a large spectral entropy does
not by itself guarantee equilibration or the absence
of revivals, since these properties also depend on
the relations among the active frequencies.

$\Hpow$ and $\deffspec$ are both derived from the same normalized SPD; they measure its concentration in different ways. The spectral dimension is an inverse participation ratio, whereas $\Hpow$ is the Shannon entropy of the spectral-power distribution. In particular, the entropy depends logarithmically on the individual weights,
while the inverse participation ratio involves their squares. Their relationship can be expressed through the
Rényi entropy of order two, which satisfies 
\begin{equation}
\ln\deffspec
=
-\ln\left[
\sum\subtiny{0}{0}{k}
p\subtiny{0}{0}{\leak}
(\omega\subtiny{0}{0}{k})
\suptiny{0}{0}{2}
\right]
\leq
\Hpow.
\end{equation}
Thus, the two quantities are mathematically related but
retain different sensitivities to the distribution of
spectral power among dominant and weakly participating channels. 

Together, $\deffspec$ and $\Hpow$ characterize the
spectral organization of the LFF fluctuations. They complement $\deffdyn$, which controls the effective
participation of oscillatory amplitudes, and the standard
effective dimension $\deff$, which characterizes the
distribution of energy populations. This distinction
separates the overall fluctuation scale from
the finer frequency-domain structure of the dynamics.
%

\begin{figure*}[t]
    \centering
    \begin{subfigure}[b]{0.33\textwidth}
        \centering
        \includegraphics[width=\textwidth]
        {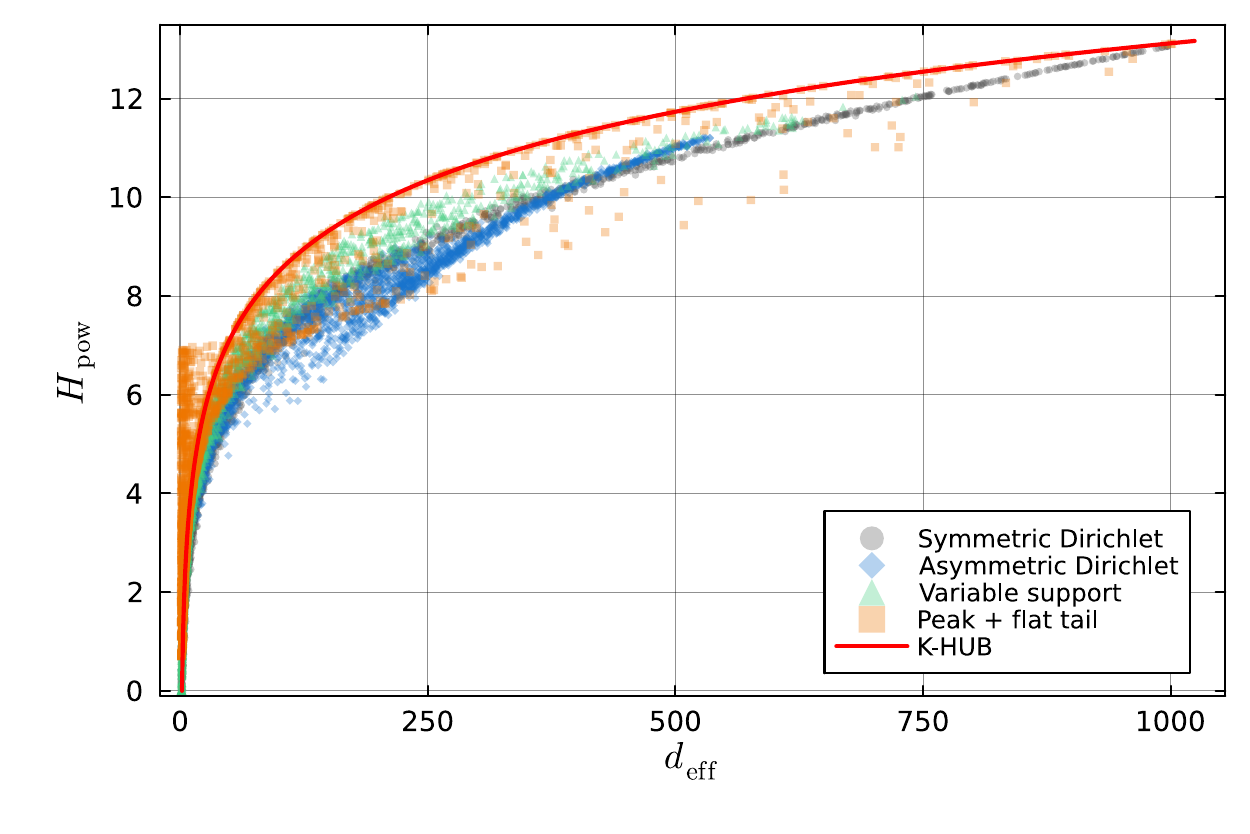}
        \caption{$\Hpow$ versus $\deff$.}
        \label{fig:fig5a}
    \end{subfigure}
    \hfill
    \begin{subfigure}[b]{0.33\textwidth}
        \centering
        \includegraphics[width=\textwidth]
        {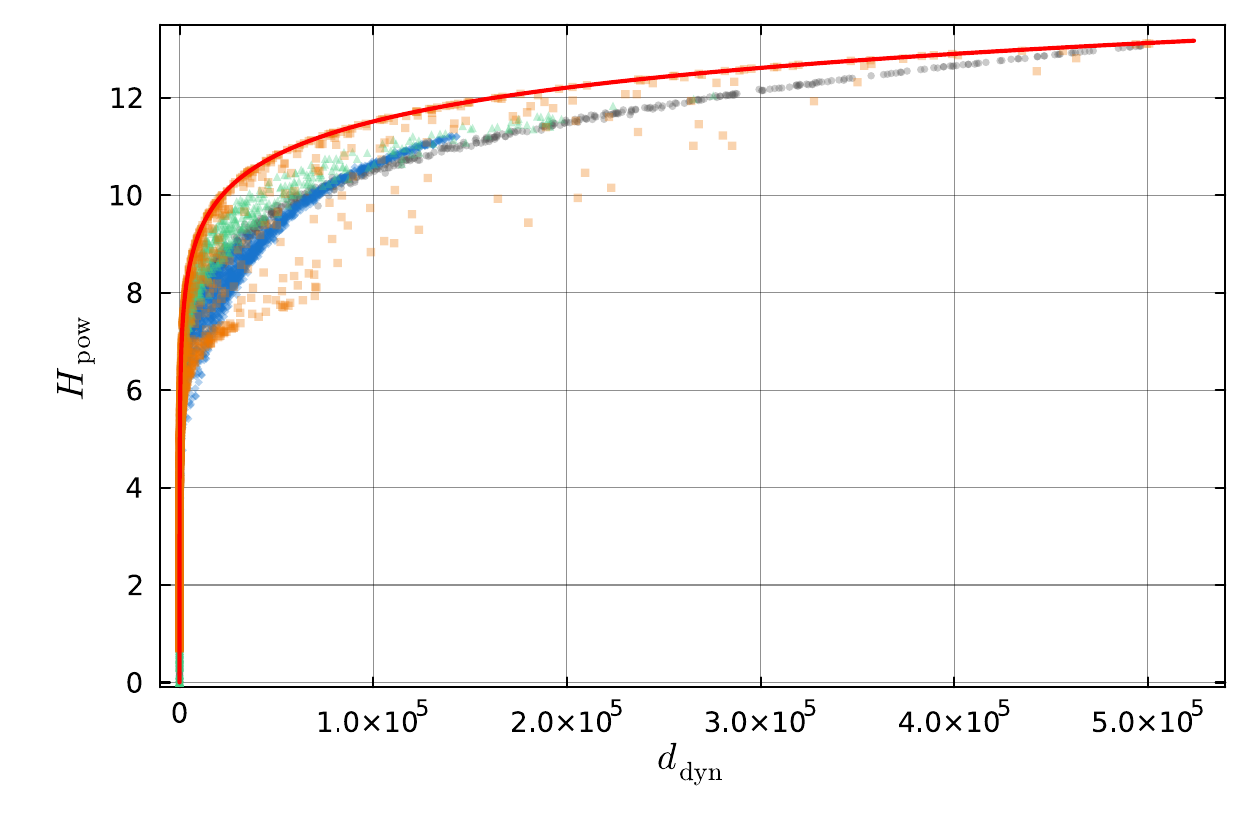}
        \caption{$\Hpow$ versus $\deffdyn$.}
        \label{fig:fig5b}
    \end{subfigure}
    \hfill
    \begin{subfigure}[b]{0.33\textwidth}
        \centering
        \includegraphics[width=\textwidth]
        {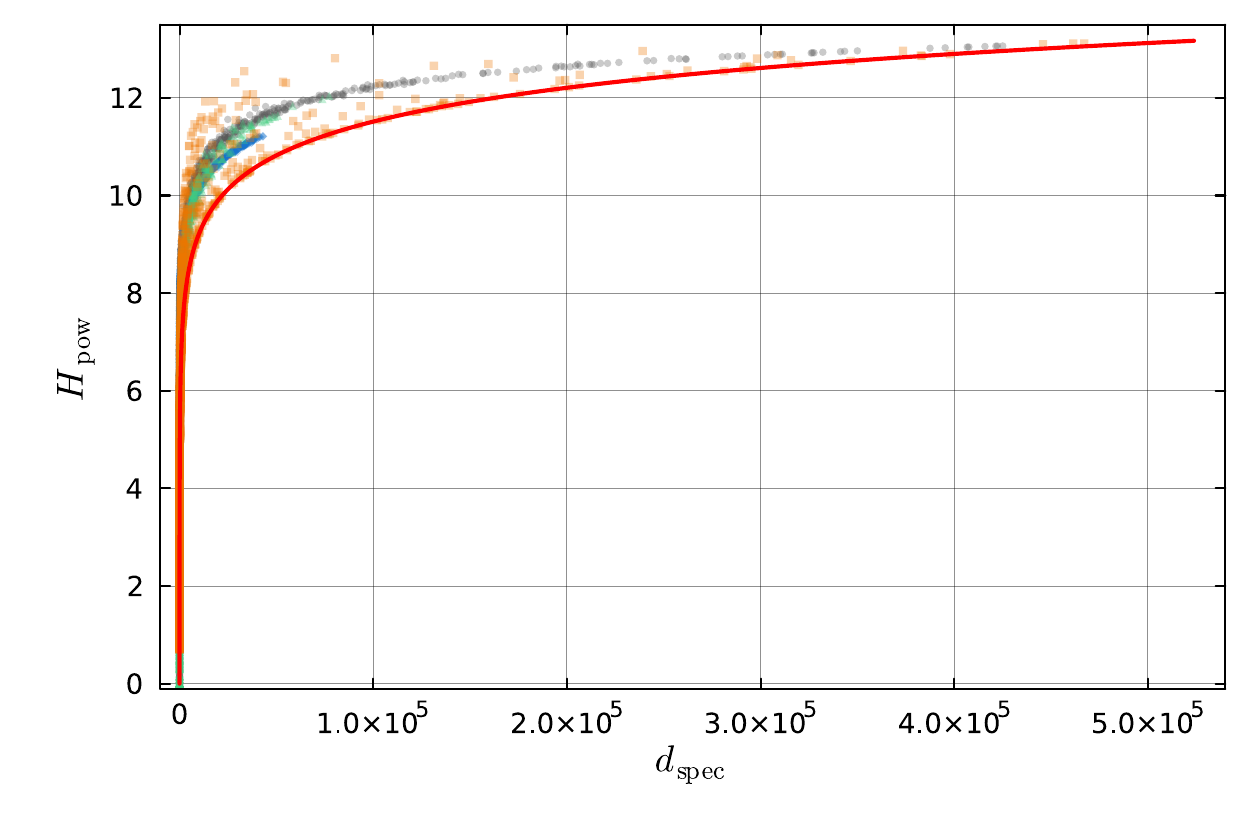}
        \caption{$\Hpow$ versus $\deffspec$.}
        \label{fig:fig5c}
    \end{subfigure}

\caption{\justifying Shannon power entropy $\Hpow$ as a function of (a) the effective dimension $\deff$, (b) the dynamical dimension $\deffdyn$, and (c) the spectral dimension $\deffspec$, for $8000$ energy-population distributions in a Hilbert space of dimension $d\subtiny{0}{0}{\mathrm{E}}=2^{10}$. Markers denote symmetric Dirichlet distributions over all levels (gray circles, $1500$ states), Haar-random pure states (purple stars, $500$ states), asymmetric Dirichlet distributions (blue diamonds), symmetric Dirichlet distributions on random supports (green triangles), and one-peak plus flat-tail distributions (orange squares), the last three with $2000$ states each. The sampling parameters are given in Appendix~\ref{secapp:synthetic_families}. The solid red curve is the $K$-HUB family, uniform over $K$ energy eigenstates ($2\leq K\leq d\subtiny{0}{0}{\mathrm{E}}$), for which $\deff=K$, $\deffdyn=\deffspec=\binom{K}{2}$, and $\Hpow=\ln\binom{K}{2}$; in panel (c), it marks the equality case of $\Hpow\geq\ln\deffspec$.}
    \label{fig:fig5}
\end{figure*}


\subsubsection{Hamiltonian Unbiased Basis}\label{sec:leak_HUB}
A Hamiltonian Unbiased Basis (HUB) is an orthonormal family 
$\{\ket{\varphi_{\alpha}}\}_{\alpha=0}^{d\subtiny{0}{0}{E}-1}$ that is mutually unbiased with the
energy eigenbasis $\{\ket{E\subtiny{0}{0}{n}}\}_{n=0}^{d\subtiny{0}{0}{E}-1}$ of the Hamiltonian $H$: 
\begin{equation}\label{eq:huo_rho_0}
    |\braket{\varphi_{\alpha}}{E\subtiny{0}{0}{n}}|\suptiny{0}{0}{2} = \frac{1}{d\subtiny{0}{0}{E}},\text{ for all }\alpha,\,n\in[0,d\subtiny{0}{0}{\mathrm{E}}].
\end{equation}
Hamiltonian Unbiased Observables (HUOs) are operators $O=\sum_{\alpha}
o_{\alpha}\ketbra{\varphi_{\alpha}}{\varphi_{\alpha}}$ whose eigenstates form a HUB.  A key property is that their equilibrium expectation value coincides with the
microcanonical prediction~\cite{Anza2017,Anza2018}:
\begin{equation}
    \Tr[O\,\eqst] = \frac{1}{d\subtiny{0}{0}{E}}\sum_{\alpha} o_{\alpha}.
\end{equation}
If the initial state is itself an element of a HUB, $\ket{\psi\subtiny{0}{0}{0}}=\ket{\varphi_{\alpha}}$,
hence the effective dimension has maximum value, $\deff\suptiny{0}{0}{\text{HUB}}=d\subtiny{0}{0}{E}$, and LFF becomes
\begin{equation}
   \label{eq:HUB_SFF}
    \leak\subtiny{0}{0}{\mathrm{HUB}}(t)
    = 1 - \frac{1}{d\subtiny{0}{0}{E}\suptiny{0}{1}{2}}\sum_{m,n} e^{-i(E\subtiny{0}{0}{n}-E\subtiny{0}{0}{m})t}.
\end{equation}
{Using the normalized spectral form factor defined in Eq.~\eqref{eq:SFF}, Eq.~\eqref{eq:HUB_SFF} becomes
\begin{equation}
    \leak\subtiny{0}{0}{\mathrm{HUB}}(t)
    = 1 - \mathcal{K}(t).
\end{equation}}
Therefore, at equilibrium,
\begin{equation}
    \leak\subtiny{0}{0}{\infty}\suptiny{0}{0}{\mathrm{HUB}} = 1 - \frac{1}{d\subtiny{0}{0}{E}}.
\end{equation}
Thus, the long-time LFF of a HUB initial state depends only on $d\subtiny{0}{0}{E}$, whereas its temporal behavior remains sensitive to the energy-gap structure~\cite{Scarpa2025}.

In the case of Eq.~\eqref{eq:huo_rho_0}, the effective dimension achieves its largest value $\deff = d\subtiny{0}{0}{E}$, {which yields the smallest bound within the pure-state family at fixed spectral factor, without by itself fixing the equilibration timescale.} It can be expressed by the LFF variation bound, in Eq.~\eqref{eq:leak_bound}, as
\begin{align}
  &\avgT{\lvert \delta\leak\subtiny{0}{0}{\mathrm{HUB}}(t)\rvert\suptiny{0}{0}{2}}
  \;\le\;
  \displaystyle f(\varepsilon,T)\,\frac{d\subtiny{0}{0}{E}-1}{d\subtiny{0}{0}{E}\suptiny{0}{0}{3}},
  \label{eq:leakage-variance_HUO}
\end{align}
where $f(\varepsilon,T)$ is the spectral factor, defined in Eq.~\eqref{eq:spectral_factor}. In the context of quantum chaos and ETH, it implies the same upper bound for the spectral form factor variation, in Eq.~\eqref{eq:SFF}, as 
\begin{align}
  \avgT{\lvert \delta\leak\subtiny{0}{0}{\mathrm{HUB}}(t)\rvert\suptiny{0}{0}{2}}
 =
 \avgT{\lvert \delta \mathcal{K}(t)\rvert\suptiny{0}{0}{2}}.
\end{align}

{In Fig.~\ref{fig:fig5}, we illustrate the distribution of $\Hpow$ as a function of $\deff$, $\deffdyn$, and $\deffspec$ for ${8000}$ states sampled from the parent families described in Appendix~\ref{secapp:synthetic_families}, in a $2^{10}$-dimensional Hilbert space, with a parametrization chosen to cover the range $\deff\in[2,1024]$. The purpose of this figure is to show that $\Hpow$ is not redundant with the effective-dimensional descriptors, in particular with $\deffspec$. The spectral dimension $\deffspec$ is mainly sensitive to the quadratic concentration of spectral power in dominant frequency channels, whereas $\Hpow$ measures the overall spread of the spectral-power distribution. As a result, two states may have similar $\deffspec$ but different $\Hpow$, and there is no universal function $\Hpow=f(\deffspec)$.}
The red curve corresponds to the structured $K$-HUB reference family. For a fixed integer $K\in[2,d\subtiny{0}{0}{\mathrm{E}}]$, a $K$-HUB state is uniformly supported on a $K$-dimensional subset $S\subtiny{0}{0}{K}$ of Hamiltonian eigenstates, satisfying $\left|\braket{E\subtiny{0}{0}{n}}{\varphi\subtiny{0}{0}{\alpha}}\right|\suptiny{0}{0}{2}=1/K$ for $n\in S\subtiny{0}{0}{K}$, with $|S\subtiny{0}{0}{K}|=K$. For non-degenerate positive gaps, this produces a uniform spectral-power distribution over $N\subtiny{0}{0}{K}=\binom{K}{2}$ active Bohr-frequency channels. Consequently, $\deff=K$, $\deffdyn=\deffspec=N\subtiny{0}{0}{K}$, and $\Hpow=\log N\subtiny{0}{0}{K}$. Thus, the red curve represents the uniform spectral-power benchmark, for which $\Hpow=\log\deffspec$, rather than a universal relation between $\Hpow$ and $\deffspec$. The HUB state is recovered as the extremal case $K=d\subtiny{0}{0}{\mathrm{E}}$. 

The apparently anomalous orange points above the $K$-HUB curve in
Fig.~\ref{fig:fig5a} are not a numerical artifact.
Consider the one-peak plus flat-tail family
\begin{equation}
    p(a,M)=
    \left(
        a,
        \underbrace{\frac{1-a}{M},\ldots,\frac{1-a}{M}}_{M\ \mathrm{terms}}
    \right).
\end{equation}
{This family has support size $M+1$ with
$3\leq M+1\leq d_E$. Its support can be much larger than
its effective dimension, even though a $K$-HUB state has
support size exactly $K=\deff$.}
In the plot, $M\leq d\subtiny{0}{0}{E}-1 \approx 2^{10}$, with effective
dimension $\deff\leq a\suptiny{0}{0}{-2}$; increasing $M$ generates a growing number
of nearly equipowered peak--tail Bohr-frequency channels, yielding
$\Hpow\simeq\ln M\leq 6.9$.
For $K$-HUB, by contrast, $\deff=K$ and $\Hpow=\ln\binom{K}{2}$,
which is logarithmically smaller than the effective dimension.
Consequently, the $K$-HUB curve does not constitute an upper bound on
$\Hpow$ at fixed $\deff$: distributions with larger support can exhibit
small effective dimension while attaining greater spectral entropy.

\subsubsection*{Maximum Shannon power entropy.}
The discrete Shannon entropy of the spectral power distribution quantifies how widely the spectral weight is distributed across the accessible transition frequencies. In a finite-dimensional system with a non-degenerate spectrum, the Bohr frequencies come in pairs $\pm\omega$ arising from unordered pairs of distinct energies. It is, therefore, natural to work with \emph{non-oriented} gaps, or equivalently with the spectrum restricted to $\omega>0$. For an effective Hamiltonian $H=\sum\subtiny{0}{0}{n} E\subtiny{0}{0}{n} \ketbra{E\subtiny{0}{0}{n}}{E\subtiny{0}{0}{n}}$; {assuming nondegenerate energy levels and distinct positive energy gaps,} each unordered pair $\{n,m\}$ with $n<m$, defines a unique positive gap. Hence, the number of distinct positive gaps is $N_{\mathrm{gaps}}= \binom{d\subtiny{0}{0}{E}}{2}.$ For a HUB state, characterized by $|c\subtiny{0}{0}{n}|\suptiny{0}{0}{2}=1/d\subtiny{0}{0}{E}$, each unordered pair $\{n,m\}$ contributes equally at the transition level; therefore, the aggregated spectral weight is uniformly distributed over the distinct positive gaps, and one has $p\subtiny{0}{0}{\leak}(\omega\subtiny{0}{0}{k}) = \frac{1}{N_{\mathrm{gaps}}} =
\frac{2}{d\subtiny{0}{0}{E}(d\subtiny{0}{0}{E}-1)},
\quad \omega\subtiny{0}{0}{k}>0.$ Hence, the Shannon power entropy is maximized under this gap structure, since the spectral power distribution is uniform; thus, $\Hpow\suptiny{0}{0}{\max}= \ln N_{\mathrm{gaps}}$, 
\begin{equation}
\Hpow\suptiny{0}{0}{\max}=
\ln\!\left(\frac{d\subtiny{0}{0}{E}(d\subtiny{0}{0}{E}-1)}{2}\right),
\end{equation}
For large $d\subtiny{0}{0}{E}$, this becomes $\Hpow\suptiny{0}{0}{\max}
\simeq
2\ln d\subtiny{0}{0}{E}-\ln 2
=
\ln\!\Big(\frac{d\subtiny{0}{0}{E}\suptiny{0}{1}{2}}{2}\Big).$ For a HUB state one has $\deff=d\subtiny{0}{0}{E}$, hence, $ \Hpow\suptiny{0}{0}{\max}
\simeq 2\ln \deff-\ln 2.$

\subsubsection*{Minimum Shannon power entropy.}
At the opposite extreme, $H{\subtiny{0}{0}{\mathrm{pow}}}$ is minimized when the whole spectral power is concentrated on a single positive energy gap. This occurs, for instance, when all positive transition frequencies that contribute to the spectrum coincide at a single value $\omega\subtiny{0}{0}{0}>0$. In this case, the relevant gap degeneracy is maximal, and the spectral power distribution satisfies
\begin{equation}
p\subtiny{0}{0}{\leak}(\omega\subtiny{0}{0}{0})=1,
\quad
p\subtiny{0}{0}{\leak}(\omega)=0,
\quad
\text{for all } \omega\neq \omega\subtiny{0}{0}{0} .
\end{equation}
Therefore, the spectral power distribution is perfectly localized on one frequency, and the Shannon power entropy attains its exact discrete minimum,
\begin{equation}
H\suptiny{0}{0}{\min}_{\subtiny{0}{0}{\mathrm{pow}}}=0.
\end{equation}
This is the discrete analog of a monochromatic SPD: all the relevant oscillatory power is carried by a single Bohr frequency. Hence, in the present discrete formulation, the lower bound is well-defined and is attained precisely by a completely localized spectral power distribution.

It is important, however, not to overinterpret this statement as a generic claim about the full Hamiltonian spectrum. The condition above only requires that all positive gaps contributing to nonzero populations coincide at $\omega\subtiny{0}{0}{0}$. It does not require that all energy-level spacings of the Hamiltonian be identical, nor that the full set of spectral gaps have maximal degeneracy. Thus, the relevant degeneracy is that of the active gap sector, selected jointly by the initial state, its projective subspace, and the Hamiltonian dynamics.

\section{Equilibration diagnostics}\label{sec:diagnosing_equilibration}

In this section, we discuss how LFF can serve as a diagnostic of equilibration on average. We introduce spectral descriptors that separate three distinct layers of equilibration: the equilibrium LFF level, the amplitude of temporal fluctuations, and the signal's spectral complexity. 

\subsection{dynamical dimension}
\label{subsec:deffdyn} 
The cosine-only decomposition in
Eq.~\eqref{eq:leak_mixed_phasor} shows that
the temporal fluctuations of the LFF are governed by the amplitudes $b\subtiny{0}{0}{nm}=|\rho\subtiny{0}{0}{nm}P\subtiny{0}{0}{mn}|$. For nondegenerate positive Bohr gaps, each pair $(n,m)$, with $n<m$, corresponds to a distinct harmonic mode. Using the first moment defined in Eq.~\eqref{eq:spectral_moments}, these amplitudes induce the normalized distribution
\begin{equation}
\label{eq:q_amp_mixed}
\qnmdyn
=
\frac{b\subtiny{0}{0}{nm}}{S\subtiny{0}{0}{1}},
\qquad n<m,
\end{equation}
where $S\subtiny{0}{0}{1}>0$ as in Eq.~\eqref{eq:spectral_moments}, and $\sum\subtiny{0}{0}{n<m}\qnmdyn=1$.
The weights $\qnmdyn$ quantify the relative
contribution of each pair-resolved mode to
the temporal fluctuation signal.

The corresponding inverse participation ratio
defines the \textit{dynamical dimension},
\begin{equation}
\label{eq:deffspecamp_clean}
\deffdyn(\rho\subtiny{0}{0}{0},P)
:=
\frac{1}{
\displaystyle\sum\subtiny{0}{0}{n<m}(\qnmdyn)\suptiny{0}{0}{2}}.
\end{equation}
This quantity measures the effective number of active harmonic modes, accounting for both their number and the distribution of their amplitudes. It is maximized when all active modes contribute equally and decreases when a few modes dominate
the signal.

For degenerate positive Bohr gaps, contributions
sharing the same frequency must be combined
coherently. Using the sets $\Gamma\subtiny{0}{0}{k}$ defined in
Eq.~\eqref{eq:Gamma_alpha}, the corresponding
amplitudes become
$A\subtiny{0}{0}{k}=|\sum_{(n,m)\in\Gamma\subtiny{0}{0}{k}}\rho\subtiny{0}{0}{nm}P\subtiny{0}{0}{mn}|$.
The normalized frequency weights are then
$q\subtiny{0}{0}{k}\suptiny{0}{0}{\mathrm{dyn}}=A\subtiny{0}{0}{k}/\sum\subtiny{0}{0}{j} A\subtiny{0}{0}{j}$, and the
dynamical dimension is obtained from their
inverse participation ratio. This construction
accounts for interference between contributions
with different phases within each frequency
sector and recovers Eq.~\eqref{eq:deffspecamp_clean}
when the positive gaps are nondegenerate.
%

\begin{figure*}[t!]
\centering
\begin{subfigure}{0.49\textwidth}
    \centering
    \includegraphics[width=\linewidth]{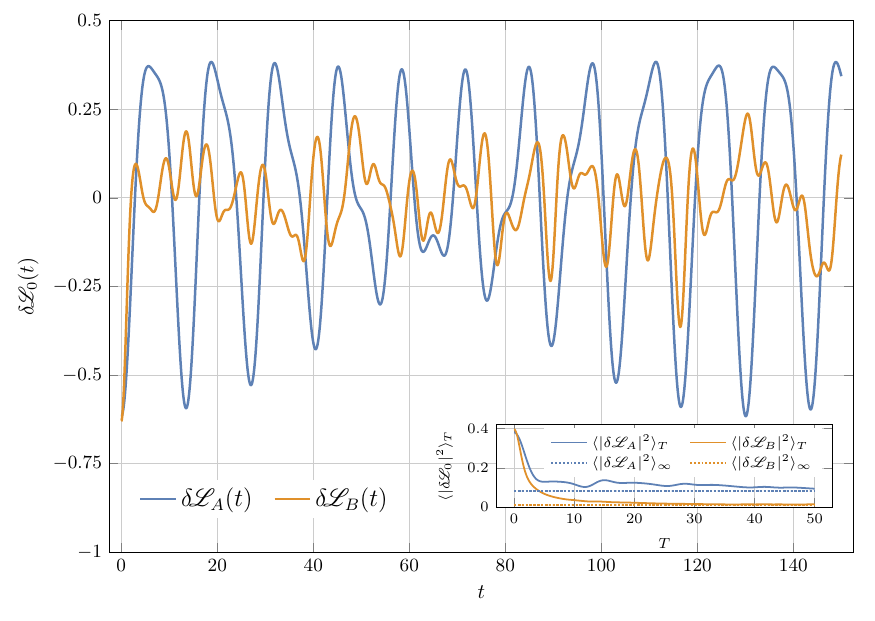}
    \caption{Initial states A and B with: $\deffdyn(A)\simeq\deff(A)\simeq\deff(B)< \deffdyn(B)$.} 
    \label{fig:fig6}
\end{subfigure}
\hfill
\begin{subfigure}{0.49\textwidth}
    \centering
    \includegraphics[width=\linewidth]{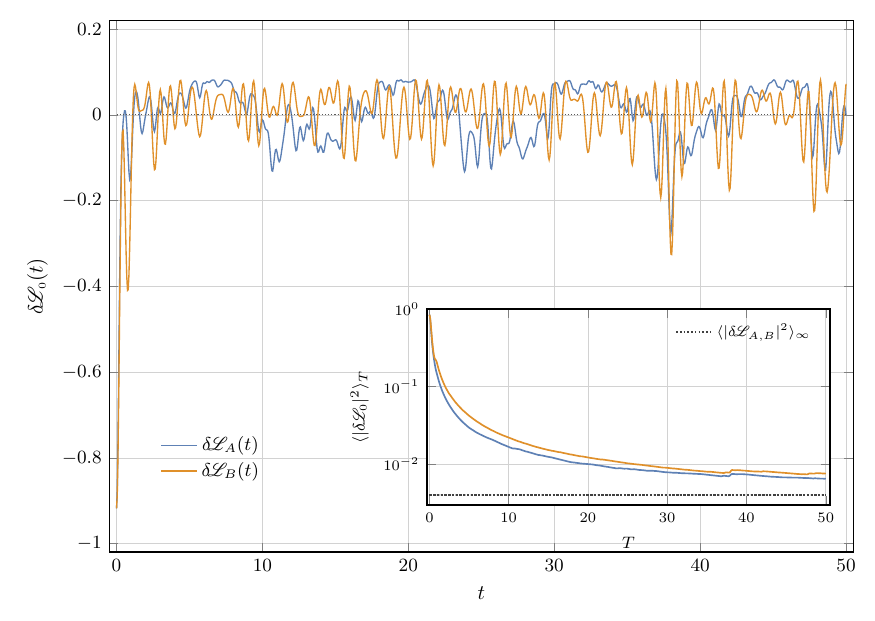}
    \caption{Initial states A and B with:
    $\deff(A)\simeq\deff(B)\ll\deffdyn(A)\simeq\deffdyn(B)$.}
    \label{fig:fig7}
\end{subfigure}
\caption{\justifying Controlled pure initial states $A$ and $B$ that separate the roles of the spectral descriptors (see Appendix~\ref{secapp:controlled_examples}). In both panels, the main plot shows the centered signal $\delta\!\leakz(t)=\leakz(t)-\leakz\suptiny{0}{0}{\infty}$, and the inset shows its running mean-square value $\avgT{|\delta\!\leakz|\suptiny{0}{0}{2}}$. (a) The states have nearly equal effective dimensions, $\deff(A)\simeq2.61$ and $\deff(B)\simeq2.70$, and hence similar equilibrium values, $\leakz\suptiny{0}{0}{\infty}\simeq0.62$ and $0.63$, but different dynamical dimensions, $\deffdyn(A)\simeq2.33$ and $\deffdyn(B)\simeq16.7$. The smaller $\deffdyn$ of $A$ produces larger and more persistent oscillations: $\deff$ fixes the equilibrium value, whereas $\deffdyn$ controls the fluctuation scale. (b) Peak-and-tail populations on the same $48$ eigenstates of the $N=10$ Ising-like Hamiltonian, with matched second and fourth population moments [Eq.~\eqref{eq:population_moments_app}]. Both states have $\deff\simeq12.2$, $\deffdyn\simeq104.3$, $\leakz\suptiny{0}{0}{\infty}\simeq0.918$, and the same long-time variance, $\avginfty{|\delta\!\leakz|\suptiny{0}{0}{2}}\simeq4.04\times10^{-3}$ (dotted line in the inset), but different spectral descriptors: $\deffspec(A)\simeq2.28$ and $\Hpow(A)\simeq2.26$, whereas $\deffspec(B)\simeq31.1$ and $\Hpow(B)\simeq3.96$. The variance thus does not determine how the fluctuation power is distributed across frequencies.}
\label{fig:figs6_7}
\end{figure*}

\begin{lemma}[Exact long-time variance;
cf.~\textrm{Ref.~\cite{reimann2008foundation}}]
\label{lem:exact_variance}
Let $H=\sum\subtiny{0}{0}{n} E\subtiny{0}{0}{n}\ket{E\subtiny{0}{0}{n}}\!\bra{E\subtiny{0}{0}{n}}$
have distinct energy levels, and let $\rho\subtiny{0}{0}{0}$
be an arbitrary initial density operator
with support projector $P$.
Assuming nondegenerate positive energy gaps,
the long-time variance of the LFF is
\begin{equation}
\avginfty{|\delta\leakP(t)|\suptiny{0}{0}{2}}
=
2S\subtiny{0}{0}{2},
\end{equation}
\end{lemma}
\noindent where $S\subtiny{0}{0}{2}$ is given in Eq.~\eqref{eq:spectral_moments}. For degenerate positive gaps, the exact variance becomes $2\sum\subtiny{0}{0}{k} A\subtiny{0}{0}{k}\suptiny{0}{0}{2}$, with the coherently combined amplitudes defined above. Based on Lemma~\ref{lem:exact_variance}, we can compute the long-time fluctuations of the LFF as a function of $\deffdyn$, defined in Eq.~\eqref{eq:deffspecamp_clean}.

\begin{theorem}[Long-time fluctuations of LFF]
\label{thm:LFF_variance_clean}
Let $\rho\subtiny{0}{0}{0}$ be an arbitrary initial density
operator, with support projector
$P=\operatorname{supp}(\rho\subtiny{0}{0}{0})$.
Under the assumptions of
Lemma~\ref{lem:exact_variance}, and for $S\subtiny{0}{0}{1}>0$,
the long-time mean-square fluctuation of
the LFF satisfies
\begin{equation}
\label{eq:variancedeff_dyn_mixed}
\avginfty{|\delta\leakP(t)|\suptiny{0}{0}{2}}
=
\frac{2S\subtiny{0}{0}{1}\suptiny{0}{2}{2}}{\deffdyn(\rho\subtiny{0}{0}{0},P)}
\geq
\frac{(1-p\subtiny{-1.8}{-0.5}{\mu})\suptiny{0}{2}{2}}{
2\deffdyn(\rho\subtiny{0}{0}{0},P)},
\end{equation}
where $p\subtiny{-1.8}{-0.5}{\mu}=\Tr(P\eqst)$ is the
equilibrium support probability defined
in Eq.~\eqref{eq:alpha0_leak}. 
\end{theorem}
The general proof is presented in
Appendix~\ref{app:thm3_misto}.

\begin{figure}[t]
    \centering
    \includegraphics[width=1\linewidth]{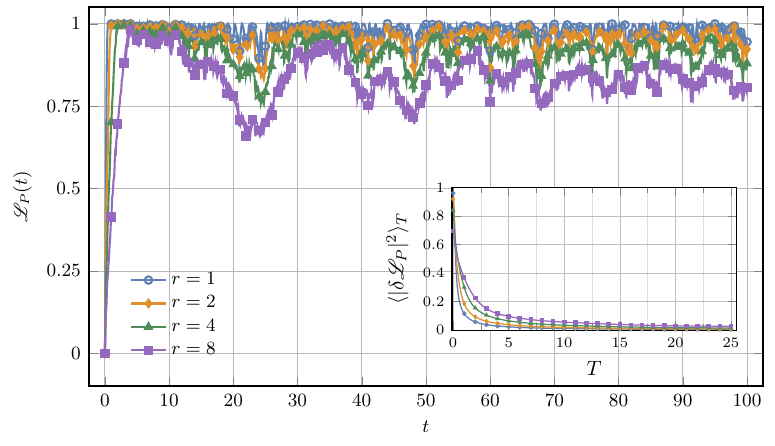}
    \caption{\justifying Time evolution of LFF, $\leakP(t)$,
        for the initial states $\rho_0^{(r)}=P_r/r$,
        with $r=1,2,4,8$, in the $N=10$ spin chain.
        Here, $P_r$ projects onto $r$ orthonormal Fourier
        superpositions of a common set of $K=48$ energy eigenstates,
        and each curve is evaluated with its corresponding initial
        support $P=P_r$. The case $r=1$ provides a pure-state reference.
        All preparations have identical uniform energy populations
        on the selected levels and the same effective dimension
        $\deff=K$, whereas their long-time leakage values satisfy
        $\leakP\suptiny{0}{0}{\infty}=1-r/K$.
        The inset shows the running mean-square deviation, $   \avgT{|\delta\!\leakP|\suptiny{0}{0}{2}}.$         Increasing $r$ lowers the {long-time LFF} and produces
        larger running mean-square deviations over the displayed
        interval, despite the unchanged energy populations.}
    \label{fig:mixed_plots_khub}
\end{figure}

%
The relative root-mean-square fluctuation
of LFF can be related to the
spectral amplitudes. Since
$\leakP(0)=0$ and
$\leakP\suptiny{0}{0}{\infty}
=1-p\subtiny{-1.8}{-0.5}{\mu}$,
the spectral decomposition gives
\begin{equation}
1-p\subtiny{-1.8}{-0.5}{\mu}
=
2\sum\subtiny{0}{0}{n<m}
\operatorname{Re}\{z\subtiny{0}{0}{nm}\}
\leq
2\sum\subtiny{0}{0}{n<m}
|z\subtiny{0}{0}{nm}|
=
2S\subtiny{0}{0}{1},
\end{equation}
where $S\subtiny{0}{0}{1}$ is presented in Eq.~\eqref{eq:spectral_moments}. Consequently, for nondegenerate positive energy gaps and
$\leakP\suptiny{0}{0}{\infty}>0$,
Theorem~\ref{thm:LFF_variance_clean} results
\begin{equation}
\label{eq:relative_fluct}
\frac{
\sqrt{\avginfty{|\delta\leakP|\suptiny{0}{0}{2}}}
}{
\leakP\suptiny{0}{0}{\infty}
}
=
\frac{1}{
\sqrt{2\,\deffdyn(\rho\subtiny{0}{0}{0},P)}
}
\frac{2S\subtiny{0}{0}{1}}{
1-p\subtiny{-1.8}{-0.5}{\mu}}
\geq
\frac{1}{
\sqrt{2\,\deffdyn(\rho\subtiny{0}{0}{0},P)}
}.
\end{equation}
The factor
$2S\subtiny{0}{0}{1}/(1-p\subtiny{-1.8}{-0.5}{\mu})$
quantifies the effect of phase alignment
among the active spectral coefficients.
It equals unity for pure initial states
and for states maximally mixed on their
support, which saturate the lower bound.
For general mixed states, the relative
fluctuation scale depends on both this
factor and $\deffdyn$. Thus, a large
dynamical dimension alone does not
guarantee small relative fluctuations
or the suppression of revivals.

For mixed initial states with a nonuniform
spectrum on their support, the complex
coefficients
$z\subtiny{0}{0}{nm}
=\rho\subtiny{0}{0}{nm}P\subtiny{0}{0}{mn}$
can acquire nontrivial phases.
Writing
$\rho\subtiny{0}{0}{0}=P/r+\Delta$,
we obtain
$z\subtiny{0}{0}{nm}
=|P\subtiny{0}{0}{nm}|\suptiny{0}{0}{2}/r
+\Delta\subtiny{0}{0}{nm}P\subtiny{0}{0}{mn}$.
Thus, deviations from uniform eigenvalue
weights within the initial support can
modify the phases of the spectral coefficients.
For nondegenerate positive gaps, the inequality
in Eq.~\eqref{eq:relative_fluct} is strict
if at least one active coefficient
$z\subtiny{0}{0}{nm}$ is not real and
nonnegative. A nonuniform spectrum alone,
however, does not imply strictness.

The equality in Eq.~\eqref{eq:variancedeff_dyn_mixed} holds, in particular, for states maximally mixed
on their support, $\rho\subtiny{0}{0}{0}=P/r$ and for pure initial states. In the former case, 
$\leakP\suptiny{0}{0}{\infty}=1-r/\deff$, resulting 
\begin{equation}
\label{eq:var_P_over_r}
\avginfty{|\delta\leakP(t)|\suptiny{0}{0}{2}}
=
\frac{\left(1-\frac{r}{\deff}\right)\suptiny{0}{0}{2}}{
2\deffdyn(\rho\subtiny{0}{0}{0},P)}.
\end{equation}

For pure initial states, $P=\rho\subtiny{0}{0}{0}$ and
$b\subtiny{0}{0}{nm}=p\subtiny{0}{0}{n}p\subtiny{0}{0}{m}$, and the general distribution in
Eq.~\eqref{eq:q_amp_mixed} therefore reduces to
\begin{equation}
\label{eq:q_amp_def_clean}
\qnmdyn
=
\frac{2p\subtiny{0}{0}{n}p\subtiny{0}{0}{m}}{
1-\sum\subtiny{0}{0}{k} p\subtiny{0}{0}{k}\suptiny{0}{0}{2}},
\qquad n<m,
\end{equation}
recovering the normalized amplitude weights
of the cosine decomposition in
Eq.~\eqref{eq:cosine_deltaL}, and $\deffdyn$ is the inverse participation ratio of Eq.~\eqref{eq:q_amp_def_clean}. For degenerate gaps, these positive amplitudes are summed over each set $\Gamma\subtiny{0}{0}{k}$ before normalization. The connection between $\deffdyn$ and the
magnitude of temporal fluctuations follows
from the exact long-time variance, adapted
from Ref.~\cite{reimann2008foundation}. In this case, 
$S\subtiny{0}{0}{1}=(1-1/\deff)/2$. Therefore, the bound of Theorem~\ref{thm:LFF_variance_clean} is saturated, and
\begin{equation}\label{eq:variancedeff_dyn}
\avginfty{|\delta\leakz(t)|\suptiny{0}{0}{2}}
=
\frac{\left(1-\frac{1}{\deff}\right)\suptiny{0}{0}{2}}
{2\deffdyn(\ket{\psi_0})},
\end{equation}
where $\deffdyn(\ket{\psi_0})=\big[\sum\subtiny{0}{0}{n<m}q\subtiny{0}{0}{nm}\suptiny{0}{0}{2}\big]^{-1}$, with $q\subtiny{0}{0}{nm}=2p\subtiny{0}{0}{n}p\subtiny{0}{0}{m}/(1-1/\deff)$ the normalized weight of each energy gap. Equation~\eqref{eq:variancedeff_dyn} separates the two features of the LFF signal: $\deff$ sets how close the equilibrium value is to one, whereas $\deffdyn$ sets the magnitude of the temporal fluctuations around it.  

{ Finally, since
$\ket{\varphi_\alpha}\!\bra{\varphi_\alpha}
\preceq P$,
Eq.~\eqref{eq:leakage_mixed_convex}
implies
$\leakP(t)\leq
\sum_\alpha\lambda_\alpha
\leakz^{(\alpha)}(t)$,
where $\leakz^{(\alpha)}$ denotes the
pure-state LFF associated with the
eigencomponent $\ket{\varphi_\alpha}$.
Transitions within the initial support
do not contribute to $\leakP(t)$.
Consequently, the mixed-state LFF
generally differs from the convex
average of the individual pure-state LFFs.}

{Figure~\ref{fig:figs6_7} shows two pairs of controlled pure initial states, $A$ and $B$, constructed directly in the Hamiltonian eigenbasis to separate the roles of the spectral descriptors (the one-peak-plus-flat-tail family presented in Example~1 of Appendix~\ref{secapp:controlled_examples}). In Fig.~\ref{fig:fig6}, the two states have nearly equal effective dimensions, and hence similar equilibrium values of the LFF, but dynamical dimensions that differ by almost an order of magnitude; the state with the smaller $\deffdyn$ displays larger and more persistent oscillations, as predicted by Theorem~\ref{thm:LFF_variance_clean}. In Fig.~\ref{fig:fig7}, the populations allow several plateaus with different heights and multiplicities, and both states are supported on the same $48$ eigenstates of the $N=10$ Ising-like Hamiltonian. For pure states, $\deff=1/\sum\subtiny{0}{0}{n}p\subtiny{0}{0}{n}\suptiny{0}{0}{2}$ and $\deffdyn=\big(1-\sum\subtiny{0}{0}{n}p\subtiny{0}{0}{n}\suptiny{0}{0}{2}\big)\suptiny{0}{0}{2}/\big\{2\big[\big(\sum\subtiny{0}{0}{n}p\subtiny{0}{0}{n}\suptiny{0}{0}{2}\big)\suptiny{0}{0}{2}-\sum\subtiny{0}{0}{n}p\subtiny{0}{0}{n}\suptiny{0}{0}{4}\big]\big\}$ depend only on the second and fourth population moments [Eq.~\eqref{eq:population_moments_app}]. Matching these two moments while keeping the higher ones distinct therefore yields the same equilibrium value and the same long-time variance, Eq.~\eqref{eq:variancedeff_dyn}, but markedly different values of $\deffspec$ and $\Hpow$. Hence, $\deff$ and $\deffdyn$ fix the equilibrium value and the total fluctuation scale, but not how the fluctuation power is distributed among the Bohr frequencies, which is captured by $\deffspec$ and $\Hpow$.}

{In Fig.~\ref{fig:mixed_plots_khub}, we examine the role of the initial support at fixed energy populations. States $\rho_0^{(r)}=P_r/r$ were considered, where $P_r$ projects onto $r$ orthonormal Fourier superpositions of the same $K=48$ energy eigenstates.
Although all preparations share uniform energy populations and the same effective dimension $\deff=K$, their LFF dynamics differ: increasing $r$ lowers the equilibrium value
$\leakP\suptiny{0}{0}{\infty}=1-r/K$ and yields larger finite-time mean-square deviations. This comparison shows that energy populations alone do not characterize mixed-state equilibration, and
the rank and structure of the initial support also affect
the dynamics.}

\subsection{{Robustness to noise}}\label{sec:noise}

The descriptors above characterize the closed-system dynamics, whereas any experimental reconstruction of $\leakP(t)$ is affected by noise. Since the LFF is a population, it is estimated from repeated preparations and measurements, and the measured signal is an average over noise realizations. Here, we consider two distinct sources of noise. In Sec.~\ref{subsec:noise_dynamical}, we analyze the LFF dynamics and its equilibration under energy dephasing, that is, noise acting diagonally in the eigenbasis of the Hamiltonian. In Sec.~\ref{subsub:stability_bounds}, we consider an imperfect preparation of the initial state, which affects both the dynamics and the equilibration descriptors.

\subsubsection{Robustness under decoherence noise}\label{subsec:noise_dynamical}
 {One can consider noise models that act diagonally in the energy eigenbasis, for which the effect on the LFF can be obtained exactly. The system undergoes energy dephasing, described by the Lindblad equation $\dot\rho=-i[H,\rho]-\frac{\gamma}{2}[H,[H,\rho]]$, which arises, e.g., from white-noise fluctuations of the overall energy scale~\cite{xu2019}. In this case the populations $\rho_{nn}$ are unchanged and each coherence acquires a damping factor, $\rho\subtiny{0}{0}{nm}e^{-i\omega\subtiny{0}{0}{nm}t}\to \rho\subtiny{0}{0}{nm}e^{-i\omega\subtiny{0}{0}{nm}t}D\subtiny{0}{0}{nm}(t)$, with
\begin{equation}
\label{eq:noise_damping}
D\subtiny{0}{0}{nm}(t)=e^{-\gamma\,\omega\subtiny{0}{0}{nm}\suptiny{0}{0}{2}t/2},
\end{equation}
where $\gamma$ is the damping factor. The mode decomposition of Eq.~\eqref{eq:leak_mixed_phasor} then becomes
\begin{equation}
\label{eq:delta_LP_noise}
\delta\leakP\suptiny{0}{0}{\,\rm noisy}(t)
=-2\sum\subtiny{0}{0}{n<m}b\subtiny{0}{0}{nm}\,D\subtiny{0}{0}{nm}(t)\cos(\omega\subtiny{0}{0}{nm}t-\phi\subtiny{0}{0}{nm}).
\end{equation}} 
{Three consequences follow from Eq.~\eqref{eq:delta_LP_noise}. First, the equilibrium value $\leakP\suptiny{0}{0}{\infty}=1-\mathrm{Tr}(P\mu)$ is unaffected, because it depends only on the populations; for pure states, $\deff$ can therefore still be read from the long-time mean of the measured signal. Second, every mode with $\omega\subtiny{0}{0}{nm}\neq0$ is damped, so $\delta\leakP^{\,\rm noisy}(t)\to0$ and the running mean-square fluctuation vanishes as $T\to\infty$, instead of saturating at the intrinsic value $2S_2$ of Lemma~\ref{lem:exact_variance}. Noise thus suppresses revivals. Third, since $D\subtiny{0}{0}{nm}$ depends on $\omega\subtiny{0}{0}{nm}\suptiny{0}{0}{2}$, noise acts as a low-pass filter on the spectral power $\mathcal{G}_{\mathscr{L}}(\omega)$: high-frequency channels are suppressed first. For fixed $\rho\subtiny{0}{0}{0}$ the individual amplitudes $b\subtiny{0}{0}{nm}$ remain unchanged. Consequently, the dimensions $\deffdyn$ and $\deffspec$ constructed from these amplitudes are exactly invariant under such shifts, independently of the perturbation magnitude. The Bohr frequencies, however, change according to $\widetilde\omega\subtiny{0}{0}{nm} =\omega\subtiny{0}{0}{nm}+\delta E\subtiny{0}{0}{m}-\delta E\subtiny{0}{0}{n}$. Changes in energy-gap degeneracies can therefore modify the coherent grouping of frequency channels and the corresponding long-time fluctuation formulas. In particular, effective dimensions defined from frequency-aggregated amplitudes need not remain invariant. Therefore, $\deffdyn$, $\deffspec$ and $H_{\rm pow}$ extracted from a noisy signal become biased toward the low-frequency part of the spectrum, while $\leakP\suptiny{0}{0}{\infty}$ is not.}
 
{Noise that does not act diagonally in the energy basis, such as local dephasing, also changes the populations; if it is unital, it generically drives the state toward the maximally mixed state, for which $\leakP\to1-r/d$, with $d$ the dimension of the accessible Hilbert space. In this case, the signature is a drift of the running mean of $\leakP(t)$ away from $\leakP\suptiny{0}{0}{\infty}$. A systematic treatment of open-system effects on spectral diagnostics can be found in Refs.~\cite{matsoukas2023unitarity,MatsoukasRoubeas2024quantumchaos}.}

\begin{figure}[t]
    \centering
    \includegraphics[width=0.96\linewidth]{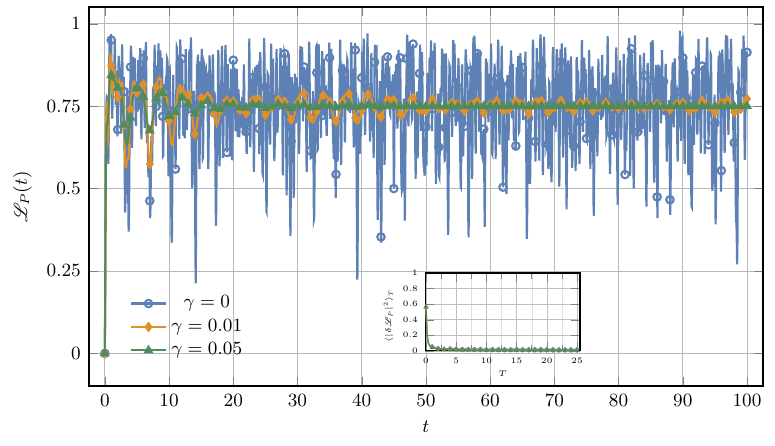}
    \caption{\justifying Leakage dynamics under energy dephasing
        in the $N=10$ spin chain, starting from the same rank-two
        mixed state with uniform populations on $K=8$ energy levels.
        The evolution follows
        \(
        \dot{\rho}=-i[H,\rho]
        -\frac{\gamma}{2}[H,[H,\rho]],
        \)
        with $\gamma=0,0.01,0.05$, while the monitored projector
        $P$ remains fixed to the initial support.
        The case $\gamma=0$ provides the unitary reference.
        All cases have $\deff=8$ and share the same long-time
        leakage value $\leakP\suptiny{0}{0}{\infty}=0.75$.
        The inset shows the running mean-square deviation,
        \(
        \avgT{|\delta\!\leakP|\suptiny{0}{0}{2}}.
        \)
        Energy dephasing suppresses the persistent oscillations
        observed under unitary evolution, with larger $\gamma$
        yielding smaller running mean-square deviations over
        the displayed interval.}
    \label{fig:dephasing}
\end{figure}

{Fig.~\ref{fig:dephasing} extends LFF analysis to nonunitary dynamics by introducing dephasing in the energy eigenbasis. Keeping the initial mixed state and its monitored support fixed, we compare the unitary reference with three dephasing strengths ($\gamma=\{0,0.01,0.05\}$). Since energy dephasing preserves the energy populations, both $\deff$ and the long-time LFF value remain unchanged. Nevertheless, the damping of energy coherences suppresses the persistent unitary oscillations and drives $\leakP(t)$ toward
its equilibrium value, reducing the running mean-square deviation. This comparison illustrates how LFF reveals the dynamical effects of coherence decay even at fixed energy populations.}

\begin{table*}[t]
\begin{minipage}[t]{0.48\textwidth}
\raggedright
\resizebox{0.99\columnwidth}{!}{
\begin{tabular}{c|l|l|l}
\hline\hline
Quantity 
& Large 
& Small 
& Eq. \\
\hline

$\deff$
& Large $\leakz\suptiny{0}{0}{\infty}$
& Small $\leakz\suptiny{0}{0}{\infty}$
& Eq.~\eqref{eq:leak_linearentropy} \\

$\spd$
& Broadband
& Narrowband
& Eq.~\eqref{eq:SPD_definition_degenerate} \\

$\deffspec$
& Large spectral support
& Few dominant modes
& Eq.~\eqref{eq:def_deffspec} \\

$\Hpow$
& High spectral entropy
& Low spectral entropy
& Eq.~\eqref{eq:Shannon_power_entropy} \\

$\deffdyn$
& Small $\avginfty{|\delta\!\leakz|\suptiny{0}{0}{2}}$
& High $\avginfty{|\delta\!\leakz|\suptiny{0}{0}{2}}$
& Eq.~\eqref{eq:deffspecamp_clean} \vspace{2pt}\\
\hline\hline
\end{tabular}}
\caption{\justifying 
Summary of the principal quantities introduced in this work. Each descriptor is grouped by structural level -- state ($\deff$), dynamical ($\deffdyn$), and spectral ($\spd(\omega)$, $\deffspec$, $\Hpow$) -- with references to their defining equations in the last column.} 
\label{tab:figMerit}
\end{minipage}
\hfill
\begin{minipage}[t]{0.48\textwidth}
\raggedright
\resizebox{0.99\columnwidth}{!}{%
\begin{tabular}{ll}
\hline\hline
Quantity & Physical meaning \\
\hline
$\deff$ &
\begin{tabular}[c]{@{}l@{}}
delocalization in energy space; \\
determines the equilibrium value of the LFF \\
for pure initial states
\end{tabular}
\\[2pt]
\hline
$\deffdyn$ &
\begin{tabular}[c]{@{}l@{}}
effective number of active oscillatory modes; \\
controls the magnitude of the temporal fluctuations \\
through dephasing
\end{tabular}
\\[2pt]
\hline
$\deffspec,\,\Hpow$ &
\begin{tabular}[c]{@{}l@{}}
distribution of spectral power across Bohr frequencies; \\
quantify the spectral complexity of the signal \\
and the structure of its temporal fluctuations
\end{tabular}
\\
\hline\hline
\end{tabular}
}
\caption{\justifying Spectral descriptors of LFF equilibration. Each quantity characterizes a distinct aspect of the dynamics: the equilibrium value, the magnitude of the fluctuations, and the spectral complexity of the signal.}
\label{tab:hierarchy_descriptors}
\end{minipage}
\end{table*}

\subsubsection{Stability under initial-state perturbations}
\label{subsub:stability_bounds}

In practice, the target initial state $\rho\subtiny{0}{0}{0}$ may differ
from the experimentally prepared or numerically reconstructed state $\widetilde{\rho}\subtiny{0}{0}{0}$ due to imperfect control, finite numerical precision, or coarse-grained characterization. We investigate how these uncertainties affect $\deffdyn$, $\deffspec$, and $\Hpow$. Previous studies have addressed the stability of quantum
equilibration under Hamiltonian perturbations and their effects on relaxation dynamics \cite{Knipschild18,Dabelow20,Passos_TRO_25}.
Here, we examine how initial-state errors propagate to the effective dimensions and spectral-power entropy.

For arbitrary initial states, the transition amplitudes are
$b\subtiny{0}{0}{nm}=|\rho\subtiny{0}{0}{0,nm}P\subtiny{0}{0}{0,mn}|$, with $n<m$, where $P\subtiny{0}{0}{0}=\operatorname{supp}(\rho\subtiny{0}{0}{0})$.
We keep $P\subtiny{0}{0}{0}$ fixed and define the perturbed amplitudes
by $\widetilde b\subtiny{0}{0}{nm}
=|\widetilde\rho\subtiny{0}{0}{0,nm}P\subtiny{0}{0}{0,mn}|$.
This convention measures the sensitivity of the spectral quantities relative to the support of the target state. It does not require $\widetilde\rho\subtiny{0}{0}{0}$ to have the same support as $\rho\subtiny{0}{0}{0}$.

Since $\deffdyn$ and $\deffspec$ are both inverse participation ratios constructed from the transition amplitudes, we introduce the unified notation $d\subtiny{0}{0}{\ell}=S\subtiny{0}{0}{\ell}\suptiny{0}{0}{2}/S\subtiny{0}{0}{2\ell}$, with $d\subtiny{0}{0}{1}=\deffdyn$ and $d\subtiny{0}{0}{2}=\deffspec$, where the spectral moments $S\subtiny{0}{0}{k}$ are given in Eq.~\eqref{eq:spectral_moments}. To quantify perturbations of these moments, we define the total absolute deviation
\begin{equation}
\varepsilon\subtiny{0}{0}{\ell}
=
\sum\subtiny{0}{0}{n<m}
\left|
\widetilde b\subtiny{0}{0}{nm}\suptiny{0}{0}{\,\ell}-b\subtiny{0}{0}{nm}\suptiny{0}{0}{\,\ell}
\right|,
\qquad \ell=1,2.
\end{equation}
For $S\subtiny{0}{0}{\ell}>0$ and
$\varepsilon\subtiny{0}{0}{\ell}<S\subtiny{0}{0}{\ell}/\sqrt{d\subtiny{0}{0}{\ell}}$,
we obtain, as demonstrated in
Appendix~\ref{app:robustness_bounds},
\begin{equation}
\label{eq:effective_dimension_robustness}
\frac{1}{
(1+\sqrt{d\subtiny{0}{0}{\ell}}\varepsilon\subtiny{0}{0}{\ell}/S\subtiny{0}{0}{\ell})\suptiny{0}{0}{2}}
\leq
\frac{\widetilde d\subtiny{0}{0}{\ell}}{d\subtiny{0}{0}{\ell}}
\leq
\frac{1}{
(1-\sqrt{d\subtiny{0}{0}{\ell}}\varepsilon\subtiny{0}{0}{\ell}/S\subtiny{0}{0}{\ell})\suptiny{0}{0}{2}}.
\end{equation}
For sufficiently small errors, the relative variation is bounded to first order by
$2\sqrt{d\subtiny{0}{0}{\ell}}\varepsilon\subtiny{0}{0}{\ell}/S\subtiny{0}{0}{\ell}$. These bounds accommodate changes in the set of active
transition channels. Their dependence on $\sqrt{d\subtiny{0}{0}{\ell}}$ reflects a worst-case estimate for absolute amplitude
errors, rather than an intrinsic loss of robustness at large effective dimensions.

The amplitude errors can also be related directly to the initial-state perturbation.
Writing $\delta\rho=\widetilde\rho\subtiny{0}{0}{0}-\rho\subtiny{0}{0}{0}$ and
$r=\Tr P\subtiny{0}{0}{0}$, with $\|\cdot\|\subtiny{0}{0}{2}$ denoting the
Hilbert--Schmidt norm, we find
\begin{equation}
\label{eq:initial_state_amplitude_error_dyn}
\varepsilon\subtiny{0}{0}{1}
\leq
\frac{\sqrt r}{2}\|\delta\rho\|\subtiny{0}{0}{2},
\end{equation}
and
\begin{equation}
\label{eq:initial_state_amplitude_error_spec}
\varepsilon\subtiny{0}{0}{2}
\leq
\sqrt{rS\subtiny{0}{0}{2}}\|\delta\rho\|\subtiny{0}{0}{2}
+\frac r4\|\delta\rho\|\subtiny{0}{0}{2}\suptiny{0}{0}{2}.
\end{equation}
Together with Eq.~\eqref{eq:effective_dimension_robustness},
these estimates establish the stability of both dimensions under sufficiently small preparation errors, without requiring the perturbed state to have the same support as the target state. A dimension-independent bound follows when all nonzero amplitudes have uniformly bounded relative errors. Suppose
$\widetilde b\subtiny{0}{0}{nm}=b\subtiny{0}{0}{nm}(1+\xi\subtiny{0}{0}{nm})$,
with $|\xi\subtiny{0}{0}{nm}|\leq\epsilon<1$.
Then
\begin{equation}
\label{eq:relative_amplitude_robustness}
\left(
\frac{1-\epsilon}{1+\epsilon}
\right)\suptiny{0}{0}{2\ell}
\leq
\frac{\widetilde d\subtiny{0}{0}{\ell}}{d\subtiny{0}{0}{\ell}}
\leq
\left(
\frac{1+\epsilon}{1-\epsilon}
\right)\suptiny{0}{0}{2\ell}.
\end{equation}
This estimate preserves the set of active amplitude channels and gives a first-order relative error bounded by $4\ell\epsilon$.
An analogous continuity estimate applies to $\deff$ when the energy eigenbasis is fixed.
Since these estimates concern different normalized distributions and error measures, they do not establish a universal ordering of robustness among the three effective dimensions. The derivations are presented in
Appendix~\ref{app:robustness_bounds}.

Finally, the spectral-power entropy $\Hpow$ admits a continuity bound. Let $K\geq2$ be the number of transition channels with nonzero weight in at least one of the original or perturbed spectral-power distributions. Defining
$\eta=\varepsilon\subtiny{0}{0}{2}/S\subtiny{0}{0}{2}$ and
$\eta\subtiny{0}{0}{*}=\min\{\eta,1-1/K\}$,
we obtain
\begin{equation}
\label{eq:power_entropy_robustness}
|\Delta\Hpow|
\leq
h\subtiny{0}{0}{2}(\eta\subtiny{0}{0}{*})+\eta\subtiny{0}{0}{*}\ln(K-1),
\end{equation}
where
$h\subtiny{0}{0}{2}(x)=-x\ln x-(1-x)\ln(1-x)$
is the binary entropy \cite{Audenaert2007}.
For $\eta\leq1-1/K$, this provides an
error-dependent continuity bound.
For larger errors, it reduces to the general
bound $|\Delta\Hpow|\leq\ln K$.
The threshold $1-1/K$ is the point at which
the continuity function reaches its maximum,
rather than an independent measure of the error. The explicit dependence on $K$ reflects the number of channels over which the spectral power may be redistributed.
The derivation is given in Appendix~\ref{app:robustness_bounds}.

\subsection{Physical interpretation of the effective dynamical dimension and the spectral descriptors}
In Table~\ref{tab:figMerit}, we summarize the principal quantities introduced throughout the work. Each descriptor is grouped by structural level -- state ($\deff$), dynamical ($\deffdyn$), and spectral ($\spd(\omega)$, $\deffspec$, $\Hpow$) -- with references to their defining equations in the last column. 
The spectral quantities introduced here provide complementary information about the equilibration properties of LFF.  The dynamics over equilibrium of LFF can be organized around three independent questions: (i) What is the long-time equilibrium value of the LFF? (ii) How large are the temporal fluctuations around this equilibrium value? (iii) How complex is the temporal structure of the signal? 

{For the pure-state specialization considered here,
the answer to (i) is determined by $\deff$. For general
mixed states, the corresponding equilibrium value is
$\leakP\suptiny{0}{0}{\infty}=1-p\subtiny{-1.8}{-0.5}{\mu}$.} From Eq.~\eqref{eq:leak_linearentropy}, the long-time average of the LFF depends on the initial state energy overlap, and $\deff$ fixes the equilibrium level reached by the dynamics. Physically, $\deff$ quantifies the delocalization of the initial state in the energy eigenbasis: when many energy levels are populated, i.e., for $\deff\to\infty$, destructive interference suppresses recurrences and drives LFF towards its maximal value, $\leakz\suptiny{0}{0}{\infty}\to 1$. {Thus, $\deff$ fixes the long-time LFF for pure initial states, whereas mixed states also depend on the initial support geometry through $\pmu$.}

The answer to (ii) is provided by the dynamical dimension $\deffdyn$. For fixed $\deff$, the size of temporal fluctuations is entirely governed by $\deffdyn$, which measures the effective number of cosine modes contributing to the time-fluctuation signal: $\delta\!\leakz(t)=-2\sum\subtiny{0}{0}{n<m} p\subtiny{0}{0}{n} p\subtiny{0}{0}{m} \cos(\omega\subtiny{0}{0}{nm}t)$. 
When many modes contribute with comparable amplitude, their interference leads to strong dynamical dephasing, suppressing oscillations. As a consequence, for fixed $\deff$, the long-time fluctuation scale
is controlled by $\deffdyn$. In particular, from the variance formula of Theorem~\ref{thm:LFF_variance_clean}, the typical fluctuation magnitude scales as $1/{\sqrt{\deffdyn}}$.

Finally, the answer to (iii) is encoded in the power spectral dimension $\deffspec$ and the associated Shannon entropy $\Hpow$. These quantities characterize the distribution of fluctuation power among the Bohr frequencies. Unlike $\deffdyn$, they do not control the overall magnitude of fluctuations, but rather their spectral organization. When most of the power is concentrated on a few frequencies, the dynamics is dominated by a small number of oscillatory modes and appears quasi-periodic. In contrast, when the spectral power is distributed over many frequencies, interference between modes produces a more irregular temporal behavior. Both quantities, therefore, provide an information-theoretic measure of spectral complexity. 

{In Table~\ref{tab:hierarchy_descriptors}, we summarize the roles of the spectral quantities defined above. The quantities $\deff$ and $\deffdyn$ fix, respectively, the equilibrium value of the LFF and the overall scale of its fluctuations, but they do not determine how the fluctuation power is distributed across the Bohr frequencies. This distribution is characterized, in a complementary way, by the Shannon power entropy $\Hpow$ and the spectral dimension $\deffspec$. In particular, spectra with many degenerate or nearly degenerate gaps concentrate power at a few frequencies, leading to low entropy and nearly periodic dynamics with pronounced revivals. Conversely, spectra with nondegenerate, incommensurate gaps distribute power across many frequencies, resulting in higher entropy and more irregular interference patterns.}

\subsection{{Complementary remark: onset of LFF and the Mandelstam--Tamm bound}}\label{sec:MT_bound_leak}
{The results above concern the long-time statistics of the LFF and do not constrain how early LFF sets in. Within the coarse-graining $\{P,Q\}$, this question reduces to a standard one, namely how fast a two-outcome statistics can change under unitary dynamics, and its answer follows from the Mandelstam--Tamm (MT) relation~\cite{MandelstamTamm1991,DeffnerCampbell2017}. We include it as a complement to the equilibration diagnostics.}

{Since $P$ is a projector, its variance in the state $\rho(t)$ is $p(t)[1-p(t)]$, where $p(t)=\Tr[P\rho(t)]=1-\leakP(t)$. The Robertson uncertainty relation for $H$ and $P$ then gives $|\dot p(t)|\leq2\Delta_H\sqrt{p(t)[1-p(t)]}$~\cite{Ballentine2nd}, where $\Delta_H\suptiny{0}{0}{2}=\Tr(\rho\subtiny{0}{0}{0}H\suptiny{0}{0}{2})-[\Tr(\rho\subtiny{0}{0}{0}H)]\suptiny{0}{0}{2}$ is conserved by the dynamics. Equivalently, the angle $\theta(t)=\arcsin\sqrt{\leakP(t)}$ between $\rho(t)$ and the initial subspace $\M{H}\subtiny{0}{0}{0}$ obeys $|\dot\theta(t)|\leq\Delta_H$, with $\theta(0)=0$, which is the MT bound on the speed of the coarse-grained statistics. For $\Delta_H>0$, and for a leakage threshold $0<\ell\leq1$, and the first-passage time $\tau\subtiny{0}{0}{\ell}=\inf\{t>0:\leakP(t)\geq\ell\}$, integration yields
\begin{equation}\label{eq:qsl_first_passage_revised}
{\tau\subtiny{0}{0}{\ell}\geq\frac{\arcsin\sqrt{\ell}}{\Delta_H}.}
\end{equation}}
{This bound holds for arbitrary mixed initial states. Moreover, whenever $\leakP\suptiny{0}{0}{\infty}>0$, the centered signal $\delta\leakP(t)$ is a nonvanishing almost-periodic function with zero mean, Eq.~\eqref{eq:leak_mixed_phasor}, and therefore takes positive values; the equilibrium level $\ell=\leakP\suptiny{0}{0}{\infty}=1-p\subtiny{-1.8}{-0.5}{\mu}$ is thus always reached, and}
\begin{equation}\label{eq:av_sp_deff}
{\tau\subtiny{0}{0}{\leakP\suptiny{0}{0}{\infty}}\geq\frac{\arccos\sqrt{p\subtiny{-1.8}{-0.5}{\mu}}}{\Delta_H},}
\end{equation}
{which for pure states, $p\subtiny{-1.8}{-0.5}{\mu}=1/\deff$, reads $\tau_{\leakz\suptiny{0}{0}{\infty}}\geq\arccos\sqrt{1/\deff}/\Delta_H$.}

{Two remarks delimit the content of this bound. First, its dependence on the equilibrium level is weak: for $\leakP\suptiny{0}{0}{\infty}\geq1/2$, $\arccos\sqrt{p\subtiny{-1.8}{-0.5}{\mu}}\in[\pi/4,\pi/2)$, so the bound changes by at most a factor of two and approaches the MT orthogonality time $\pi/(2\Delta_H)$ as $p\subtiny{-1.8}{-0.5}{\mu}\to0$. Second, it is set by the energy spread of the initial state, i.e., by the width of its LDOS, and therefore constrains only the onset of leakage. For local Hamiltonians and product initial states, $\Delta_H\suptiny{0}{0}{2}$ grows linearly with the number of constituents $N$, so the bound decreases as $N^{-1/2}$, whereas the time needed to reach the equilibrium regime can be much longer in many-body systems~\cite{Schiulaz2019}. Equation~\eqref{eq:av_sp_deff} is thus neither an estimate of the equilibration time nor a statement about sustained equilibration, which is the content of Theorems~\ref{th:leakage_bounds}--\ref{thm:LFF_variance_clean}. Table~\ref{tab:curvature_states} reports its value for the three pure initial states of Fig.~\ref{fig:fig2}.}

\begin{table}[t]
\centering
\begin{tabular}{lccc}
\hline\hline
Initial state & $\leakz\suptiny{0}{0}{\infty}$ & $\deff$  & ${\tau_{\leakz\suptiny{0}{0}{\infty}}}$\\
\hline
\texttt{Up}   & 0.6606 & 2.9468  &  0.3333 \\
\texttt{Pm}   & 0.9570 & 23.2509  &  0.4761 \\
\texttt{Dw}   & 0.9893 & 93.7412  &  0.5130\\
\hline\hline
\end{tabular}
\caption{\justifying {Asymptotic LFF, effective dimension, and first-passage lower bound for the three pure initial states in Fig.~\ref{fig:fig2}.}}\label{tab:curvature_states}
\end{table}

\section{Discussion and Future Perspectives}\label{sec:conclusions}
Equilibration in isolated quantum systems can be understood, within our framework, as the combined effect of spectral delocalization, destructive phase mixing among Bohr frequencies, and irreversible information leakage between dynamically coupled subspaces. This unified perspective is neither ensemble-based nor perturbative; it does not rely on assumptions of randomness, typicality arguments, or weak-coupling expansions but instead follows directly from the spectral structure of a fixed Hamiltonian and a fixed initial state. In particular, it makes explicit that fluctuation bounds of Reimann-type are fundamentally controlled by a function of the effective dimension $\deff$, thereby tying the suppression of temporal variations to the degree of spectral delocalization of the initial state. 

As an infidelity measure relative to the support of the initial state, the LFF defines a subspace-resolved quantity within the adopted coarse-grained description. Irreversibility, therefore, becomes operational: it is quantified by the information flow across a well-defined Hilbert-space partition. The subspace coarse-graining defined by the support of the initial state elucidates the flux of information from the distinguished sector $\M{H}\subtiny{0}{0}{0}$ to the full Hilbert space, rendering irreversibility as a geometrically controlled redistribution of support. {In contrast with open systems, where irreversibility reflects a flow of information from a subsystem to its environment, here it reflects a flow between two subspaces of a closed system, and equilibration means that, for most times, this information does not return completely to the initial subspace, as quantified by Theorems~\ref{th:leakage_bounds} and~\ref{thm:leakage_concentration_dir}. This is the sense in which the survival probability and the spectral form factor enter our analysis, independently of whether $H$ is chaotic or integrable, and it differs from their use as probes of quantum chaos.}

\subsection{{Relation to previous literature}}
\label{sec:SP_literature}
{We stress at the outset that we do not propose the LFF as a new quantity. Its functional form, the population of the orthogonal complement of a subspace, is standard, and for pure states it reduces to the survival probability. What we do is use it in a specific role: as a diagnostic of equilibration in an isolated quantum system, evaluated for a given pair $(H,\rho\subtiny{0}{0}{0})$, with the monitored subspace fixed by the support of $\rho\subtiny{0}{0}{0}$. In this role, $\leakP(t)$ is the statistics of the two-outcome coarse-grained measurement $\{P,Q\}$, which asks whether the system is still found in the region of Hilbert space where it was prepared. This gives us the notion of information loss and irreversibility in the dynamics of isolated quantum systems. A pair $(H,\rho\subtiny{0}{0}{0})$ for which this statistics becomes stationary on average, with temporal fluctuations that are small and rarely revive, displays the emergence of classical behavior and irreversibility at the coarse-grained level: a time-independent outcome distribution and an effectively irreversible loss of memory of the initial subspace, produced by the unitary dynamics alone, with no environment and no ensemble assumption~\cite{Gogolin2016}. The question addressed in this work is which spectral properties of the pair $(H,\rho\subtiny{0}{0}{0})$ determine whether, and how robustly, this emergence occurs.}

{For a pure initial state, the support projector is $P=\ket{\psi\subtiny{0}{0}{0}}\!\bra{\psi\subtiny{0}{0}{0}}$ and the LFF reduces to $\leakz(t)=1-\mathcal{S}(t)$, where $\mathcal{S}(t)$ is the survival probability (SP), also called return probability, fidelity, or Loschmidt echo (LE) in the quench literature~\cite{Zarate2023,CamposVenuti2010}. Because a substantial part of our pure-state analysis has counterparts in that literature, we summarize here what is already established and delimit what the present diagnostic use of the LFF adds. }
{\paragraph*{Time dependence of the survival probability.}
The SP is the squared modulus of the Fourier transform of the local density of states (LDOS) of the initial state~\cite{Borgonovi2016,Zarate2023}. Its decay is quadratic at short times, with a rate set by the LDOS width; at intermediate times its shape follows the LDOS profile and gives way to a power-law decay; and, when the spectrum displays level repulsion, the SP dips below its saturation value (the correlation hole) before saturating~\cite{SantosTavora2016,TorresHerrera2017,TorresHerrera2018}. The associated Thouless and relaxation time scales have been characterized in many-body models~\cite{Schiulaz2019,Lezama2021}. For pure states the LFF inherits all these features, and we do not claim new results on the time-dependent decay itself.}
{\paragraph*{Long-time statistics.}
The long-time average of the SP is the inverse participation ratio (IPR) $\sum\subtiny{0}{0}{n} p\subtiny{0}{0}{n}\suptiny{0}{0}{2}$ of the initial state in the energy eigenbasis~\cite{CamposVenuti2010,TimeFluct2013}; the relation $\leakz\suptiny{0}{0}{\infty}=1-1/\deff$ of Eq.~\eqref{eq:leak_linearentropy} is a restatement of this fact \cite{CamposVenuti2011}. The temporal fluctuations around this value are also known. For an arbitrary initial density operator and non-degenerate gaps, the exact long-time variance of the expectation value of an observable $A$ is $\sigma\subtiny{0}{0}{A}\suptiny{0}{0}{2}=\sum_{m\neq n}|A\subtiny{0}{0}{mn}|\suptiny{0}{0}{2}|\rho\subtiny{0}{0}{mn}|\suptiny{0}{0}{2}$, as written in the first equality of Eq.~(12) of Ref.~\cite{reimann2008foundation}; the same expression underlies the earlier analysis of Peres~\cite{Peres1984ergodicity}. Setting $A=P$ gives the variance of the LFF (Lemma~\ref{lem:exact_variance}, where we also allow degenerate gaps), and the pure case $A=\ket{\psi\subtiny{0}{0}{0}}\!\bra{\psi\subtiny{0}{0}{0}}$ gives $\avginfty{|\delta\leakz|\suptiny{0}{0}{2}}=\sum\subtiny{0}{0}{n\neq m}p\subtiny{0}{0}{n}\suptiny{0}{0}{2}p\subtiny{0}{0}{m}\suptiny{0}{0}{2}$. In Theorem~~\ref{thm:LFF_variance_clean},
we express a closed form for the variance of LFF in the long-time regime, which converges as a function of $\deffdyn\suptiny{0}{0}{-1}$. More broadly, the decay of temporal fluctuations with system size has been tied to the delocalization of the initial state in an energy shell~\cite{TimeFluct2013}, derived analytically in a regime of many-particle dephasing~\cite{Kiendl2017}, observed in a trapped-ion quantum simulator~\cite{Kaplan2020}, and bounded in terms of the IPR for correlators of systems without gap degeneracies~\cite{Lezama2023}. 
The finite-time bounds of Sec.~\ref{sec:variance_bound} are original contributions following the arguments of the observable equilibration results of Refs.~\cite{reimann2008foundation,linden2009quantum,short2011equilibration,Short_2012}, and the temporal concentration statement in part (A) of Theorem~\ref{thm:leakage_concentration_dir} is obtained based on a Chebyshev argument in the same spirit as Ref.~\cite{reimann2008foundation}.}
{\paragraph*{Spectral form factor.}
When the energy populations are uniform, $p\subtiny{0}{0}{n}=1/d_E$, the SP coincides with the spectral form factor (SFF); equivalently, the SFF is the SP of a thermofield-double (coherent Gibbs) state at infinite temperature~\cite{del2017scrambling,del2018decay}. Our HUB result, in the specific case of an initial Hamiltonian unbiased state, LFF reduces to SFF; $\leak\subtiny{0}{0}{\mathrm{HUB}}(t)=1-\mathcal{K}(t)$ [Eq.~\eqref{eq:HUB_SFF}] expresses this identity and connects with ETH for infinite temperature thermal states \cite{AnzaHuber2018}. The SP of a generic initial state can be read as an SFF filtered by the initial-state populations~\cite{Das2025}.}
{\paragraph*{Mixed states and subspaces.}
The extension of fidelity to mixed initial states is not unique. Peres introduced a normalized Hilbert--Schmidt overlap between two evolved mixed states~\cite{Peres1984}; applied to $\rho\subtiny{0}{0}{0}$ and the unitarily evolved counterpart, it reads $F(t)=\Tr[\rho\subtiny{0}{0}{0}\,\rho(t)]/\Tr\,\rho\subtiny{0}{0}{0}\suptiny{0}{0}{2}$, and alternatives that keep the interference between the components of the mixture, such as the allegiance, have also been proposed~\cite{Sokolov2007}. For an initial state that is maximally mixed on its support, $\rho\subtiny{0}{0}{0}=P/r$, one finds $F(t)=\Tr[P\rho(t)]=1-\leakP(t)$. In this case, the lower bound of Theorem~\ref{thm:LFF_variance_clean} is saturated. For any other mixed state, the $F(t)$ and $\leakP(t)$ differ, because the LFF weights the evolved state by the support projector and not by $\rho\subtiny{0}{0}{0}$ itself. Monitoring the population outside a fixed subspace, called \textit{state leakage}, is also applied to other contexts: for instance, the quantity $\Tr[(\id-P)\rho]$ out of a computational subspace is the basic figure of merit for leakage errors in quantum devices~\cite{Wood2018}, and the quantum Zeno literature studies dynamics confined to, and escaping from symmetric subspaces of the Hilbert space~\cite{Facchi2002}. The LFF has the same functional form as the state leakage. What distinguishes our use of it is that the monitored subspace is not externally prescribed but fixed by the support of the initial state, and that the evolution is the closed unitary dynamics generated by $H$, so transitions within the support are not counted as leakage. The coarse-graining is therefore a property of the pair $(H,\rho\subtiny{0}{0}{0})$ itself, and relies on the information flux between complement subspaces, expressing equilibration and irreversibility in closed systems.}

\subsection{{outlook on future directions}}

The present framework is formulated for finite-dimensional
systems and relies on a coarse-graining defined relative to
the support of the initial state, so its quantitative features naturally depend on the chosen subspace decomposition. Different physically motivated coarse-grainings may emphasize distinct aspects of the dynamics, reflecting the contextual nature of equilibration. For instance, one may coarse-grain the Hilbert space into energy shells, symmetry sectors such as fixed magnetization or particle number, eigenspaces of experimentally accessible few-body observables, spatial or subsystem-resolved sectors, excitation-number sectors, or other subspaces generated dynamically from the initial state. In addition, the sharpness of the derived bounds can exhibit finite-size effects, particularly in regimes where the effective dimension does not scale extensively.

Several concrete directions naturally follow from the present work. A first avenue concerns a deeper investigation of the geometric asymmetry between the initial subspace $\M{H}\subtiny{0}{0}{0}$ and its complement
$\M{H}\subtiny{0}{0}{\perp}$, in particular, whether a
physically meaningful relation can be established between the volume ratio $\dim(\M{H}\subtiny{0}{0}{0})/\dim(\M{H}\subtiny{0}{0}{\perp})$ and spectral quantities such as $\deff$ or $d\subtiny{0}{0}{E}$, thereby elevating the subspace asymmetry to a dynamical control parameter.

Extending the framework to systems exhibiting anomalously
slow or incomplete equilibration due to restricted spectral
participation would clarify how dynamically protected
subspaces constrain LFF behavior and generate persistent
revivals associated with weak ergodicity
breaking~\cite{Turner2018,Serbyn2021}.
Random circuit models provide a complementary testbed,
enabling controlled interpolation between chaotic and
nonergodic regimes while probing universal features of
spectral power distributions and dynamical
spreading~\cite{Nahum2017,vonKeyserlingk2018}.
On the experimental side, programmable quantum platforms
such as cold atoms, neutral Rydberg
arrays~\cite{Bernien2017,Labuhn2016}, and superconducting
qubits offer realistic settings in which subspace-resolved
measurements can access LFF and related spectral diagnostics.

Although the present robustness analysis establishes
stability bounds for the spectral descriptors and examines
energy dephasing, extending these results to more general
noise mechanisms and finite measurement statistics remains
an important direction. In particular, perturbations that
modify energy populations or couple different energy
eigenstates may affect both the long-time LFF values and
its fluctuation spectrum. Quantifying these effects would
help distinguish intrinsic spectral complexity from
experimentally induced fluctuations and assess the
reliability of the proposed descriptors in quantum
simulators.

An additional open question concerns the relation between
the spectral descriptors and the timescale required for
sustained equilibration. The quantum speed limit considered
here constrains the first passage population from the initial state subspace to its complement, but does not establish when fluctuations become
persistently small. Investigating whether $\deffdyn$,
$\deffspec$, and $\Hpow$, together with the distribution
of Bohr gaps, can constrain the onset and persistence
of equilibration would complement the present
first-passage analysis.
%


\section*{Data Availability}
The Julia source code used for the numerical simulations is available in the GitHub repository cited in Ref.~\cite{git2026}.


\begin{acknowledgments}
The authors acknowledge Krissia Zawadzki, Thiago R. de Oliveira, Fernando de Melo, Pedro S. Correia, Gabriel Dias Carvalho, and Raul O. Vallejos for the fruitful discussions. 
The authors also acknowledge financial support from CNPq, through process No. 445150/2024-6; 402786/2023-8, and the National Institute of Science and Technology for Applied Quantum Computing (INCT-CQA) through process No. 408884/2024-0.  
ATC acknowledges RAU No.~12-2016-AY-UNA, Per\'{u}, and also CNPq through process No.~150643/2026-8.
TD acknowledges CNPq through process No.~310854/2026-1 and 408990/2025-2.
\end{acknowledgments}


%


\hypertarget{sec:appendix}{}
\onecolumngrid
\appendix
\setcounter{theorem}{0}


\section*{Appendices}\label{Appendix}

The appendices collect the technical derivations supporting the analysis of the Leakage Fidelity Function. Appendix~\ref{appsec:theorem_leak_variance} proves Theorem~\ref{th:leakage_bounds}. Appendix~\ref{appsec:General_moments_dirichlet} derives the Dirichlet moments used to characterize Haar-typical effective dimensions and the LFF concentration, to prove Theorem~\ref{thm:leakage_concentration_dir}. Appendix~\ref{appsec:deffdyn} develops the fundamental properties of the dynamical and spectral dimensions, namely $\deffdyn$ and $\deffspec$, as a preparatory step toward the proof of Theorem~\ref{thm:LFF_variance_clean}, while Appendix~\ref{appsec:appE} provides numerical examples showing that these descriptors refine the usual effective dimension by capturing fluctuation and spectral structure beyond $\deff$.

\section{Finite-time leakage bounds for arbitrary mixed states}
\label{appsec:theorem_leak_variance}

We prove Theorem~\ref{th:leakage_bounds} and establish its support-sensitive and smallest-eigenvalue consequences. The proof retains the dependence on the energy-basis matrix elements of the initial support projector and therefore applies to arbitrary initial
density operators, including states with degenerate nonzero eigenvalues.

\begin{theorem}[Leakage Fidelity Function variance]
\label{app:leakage_bounds}
Let $H$ have $d_E$ distinct energy levels and let
$\rho\subtiny{0}{0}{0}$ be an arbitrary initial density operator. Denote
by $P=\operatorname{supp}(\rho\subtiny{0}{0}{0})$ its support projector,
$p\subtiny{0}{0}{n}=\bra{E\subtiny{0}{0}{n}}\rho\subtiny{0}{0}{0}\ket{E\subtiny{0}{0}{n}}$, and
$P\subtiny{0}{0}{nn}=\bra{E\subtiny{0}{0}{n}}P\ket{E\subtiny{0}{0}{n}}$. Define the
\emph{equilibrium support probability}
\begin{equation}\label{app:alpha0_leak}
p\subtiny{-1.8}{-0.5}{\mu}:=\Tr(P\eqst)=\sum\subtiny{0}{0}{n} p\subtiny{0}{0}{n}P\subtiny{0}{0}{nn}.
\end{equation}
For every $\varepsilon,T>0$, the finite-time mean-square deviation of the LFF from its long-time average obeys
\begin{equation}
\label{app:leak_bound}
\avgT{|\delta\leakP(t)|\suptiny{0}{0}{2}}
\leq f(\varepsilon,T)\left(1-\frac1{d_E}\right)p\subtiny{-1.8}{-0.5}{\mu}\suptiny{0}{0}{2}.
\end{equation}
Furthermore, independently of the rank and the nonzero
eigenvalues of $\rho\subtiny{0}{0}{0}$,
\begin{equation}\label{eq_app:leak_unconditional_deff_bound}
\avgT{|\delta\leakP(t)|\suptiny{0}{0}{2}}
\leq f(\varepsilon,T)\sum\subtiny{0}{0}{n} p\subtiny{0}{0}{n}\suptiny{0}{0}{2} P\subtiny{0}{0}{nn}(1-P\subtiny{0}{0}{nn})
\leq\frac{f(\varepsilon,T)}{4\deff}.
\end{equation}
The spectral factor $f(\varepsilon,T)$ is defined in
Eq.~\eqref{eq:spectral_factor}.
\end{theorem}

\begin{proof}
Let $H=\sum_{n=0}^{d_E-1}E\subtiny{0}{0}{n}\ketbra{E\subtiny{0}{0}{n}}{E\subtiny{0}{0}{n}}$ have distinct energy eigenvalues, and let $\rho\subtiny{0}{0}{0}$ be an arbitrary density operator of rank $r$. Its spectral decomposition and support projector are $\rho\subtiny{0}{0}{0}=\sum_{\alpha=0}^{r-1}
\lambda_\alpha\ketbra{\varphi_\alpha}{\varphi_\alpha},$ and $P =\sum_{\alpha=0}^{r-1}
\ketbra{\varphi_\alpha}{\varphi_\alpha}$, and
$P\rho\subtiny{0}{0}{0}P=\rho\subtiny{0}{0}{0}$. Here $\lambda_\alpha>0$ and $\sum_\alpha\lambda_\alpha=1$. If some positive eigenvalues coincide, $P$ remains the projector onto the entire support; its definition is independent of the orthonormal basis chosen within the degenerate eigenspaces. In the energy eigenbasis, define
\begin{equation}
\label{eq:app_matrix_elements_new}
\rho\subtiny{0}{0}{nm}:=\bra{E\subtiny{0}{0}{n}}\rho\subtiny{0}{0}{0}\ket{E\subtiny{0}{0}{m}},\quad
P\subtiny{0}{0}{mn}:=\bra{E\subtiny{0}{0}{m}}P\ket{E\subtiny{0}{0}{n}},\quad
p\subtiny{0}{0}{n}:=\rho_{nn}.
\end{equation}
The unitary evolution $\rho(t)=e^{-iHt}\rho\subtiny{0}{0}{0}e^{iHt}$ implies,
for $\omega\subtiny{0}{0}{nm}:=E\subtiny{0}{0}{n}-E\subtiny{0}{0}{m}$, $\bra{E\subtiny{0}{0}{n}}\rho(t)\ket{E\subtiny{0}{0}{m}}
=\rho\subtiny{0}{0}{nm}e^{-i\omega\subtiny{0}{0}{nm}t}.$ Using $\Tr\rho(t)=1$ in the definition of the LFF, and inserting two resolutions of the identity in the energy basis, we obtain
\begin{equation}
\label{eq:app_full_leak_expansion_new}
\leakP(t) =1-\Tr[P\rho(t)]=1-\sum\subtiny{0}{0}{n,m}P\subtiny{0}{0}{mn}\rho\subtiny{0}{0}{nm}
 e^{-i\omega\subtiny{0}{0}{nm}t}=1-\sum\subtiny{0}{0}{n} p\subtiny{0}{0}{n}P\subtiny{0}{0}{nn}
 -\sum\subtiny{0}{0}{n\ne m}\rho\subtiny{0}{0}{nm}P\subtiny{0}{0}{mn}e^{-i\omega\subtiny{0}{0}{nm}t}.
\end{equation}
Because the energy levels are distinct, the long-time average of each oscillatory term with $n\ne m$ vanishes. Accordingly,
the diagonal ensemble and the equilibrium LFF are
\begin{align}
\eqst&=\sum\subtiny{0}{0}{n} p\subtiny{0}{0}{n}\ketbra{E\subtiny{0}{0}{n}}{E\subtiny{0}{0}{n}},
\label{eq:app_mu_mixed_new}\\
\leakP\suptiny{0}{0}{\infty}
&=1-\Tr(P\eqst)=1-\sum\subtiny{0}{0}{n} p\subtiny{0}{0}{n}P\subtiny{0}{0}{nn}.
\label{eq:app_Linfty_mixed_new}
\end{align}
Introduce the support--equilibrium overlap
\begin{equation}
\label{eq:app_alpha_new}
p\subtiny{-1.8}{-0.5}{\mu}:=\Tr(P\eqst)=\sum\subtiny{0}{0}{n} p\subtiny{0}{0}{n}P\subtiny{0}{0}{nn}.
\end{equation}
Thus, with the ordered-pair coefficients
$c\subtiny{0}{0}{nm}:=\rho\subtiny{0}{0}{nm}P\subtiny{0}{0}{mn}$, the centered signal reads
\begin{equation}
\label{eq:app_mixed_centered}
\delta\leakP(t)
=\leakP(t)-\leakP\suptiny{0}{0}{\infty}
=-\sum\subtiny{0}{0}{n\ne m}c\subtiny{0}{0}{nm}e^{-i\omega\subtiny{0}{0}{nm}t}.
\end{equation}
Hermiticity of $\rho\subtiny{0}{0}{0}$ and $P$ implies
$c\subtiny{0}{0}{mn}=c\subtiny{0}{0}{nm}^{*}$. If $c\subtiny{0}{0}{nm}=b\subtiny{0}{0}{nm}e^{i\phi\subtiny{0}{0}{nm}}$
with $b\subtiny{0}{0}{nm}=|c\subtiny{0}{0}{nm}|$, pairing the two ordered contributions
for each $n<m$ yields the equivalent real representation
\begin{equation}
\label{eq:app_mixed_realphasor_new}
\delta\leakP(t)
=-2\sum\subtiny{0}{0}{n<m}b\subtiny{0}{0}{nm}
\cos(\omega\subtiny{0}{0}{nm}t-\phi\subtiny{0}{0}{nm}).
\end{equation}
The ordered-pair form in Eq.~\eqref{eq:app_mixed_centered}
is used for the finite-time estimate; the unordered-pair form in Eq.~\eqref{eq:app_mixed_realphasor_new} is used to define spectral moments.

\medskip
\noindent\emph{(a) Finite-time spectral estimate.} Define the index set $\mathcal A:=\{(n,m):n\ne m\}$ and regard $c=(c\subtiny{0}{0}{nm})_{(n,m)\in\mathcal A}$ as a complex column vector. Squaring Eq.~\eqref{eq:app_mixed_centered}, multiplying
out the two finite sums, and using the linearity of the time average gives
\begin{align}
\label{eq:app_squared_expanded_new}
\avgT{|\delta\leakP(t)|\suptiny{0}{0}{2}}
&=\sum_{\substack{n\ne m\\ k\ne\ell}}
 c\subtiny{0}{0}{nm}^{*}c_{k\ell}
 \avgT{e^{i(\omega\subtiny{0}{0}{nm}-\omega_{k\ell})t}}.
\end{align}
Introduce the finite-time gap-correlation matrix on
$\mathcal A\times\mathcal A$,
\begin{equation}
\label{eq:app_gap_matrix}
\Phi^{(T)}_{(nm),(k\ell)}
:=\frac{1}{T}\int_0^T
 e^{i(\omega\subtiny{0}{0}{nm}-\omega_{k\ell})t}\,dt.
\end{equation}
Equation~\eqref{eq:app_squared_expanded_new} is therefore
the quadratic form $c^\dagger\Phi^{(T)}c$. For any complex
vector $v$ indexed by $\mathcal A$,
\begin{equation}
\label{eq:app_gram_psd_new}
v^\dagger\Phi^{(T)}v
=\frac1T\int_0^T
 \left|\sum\subtiny{0}{0}{n\ne m}v\subtiny{0}{0}{nm}e^{-i\omega\subtiny{0}{0}{nm}t}\right|\suptiny{0}{0}{2} dt
\geq0.
\end{equation}
Thus, $\Phi^{(T)}$ is Hermitian and positive semidefinite. The spectral theorem and the Rayleigh--Ritz inequality give
\begin{equation}
\label{eq:app_matrix_norm_new}
\avgT{|\delta\leakP(t)|\suptiny{0}{0}{2}}=c^\dagger\Phi^{(T)}c\leq\|\Phi^{(T)}\|_{\mathrm{op}}\|c\|_2\suptiny{0}{0}{2}=\|\Phi^{(T)}\|_{\mathrm{op}}
\sum\subtiny{0}{0}{n\ne m}|\rho\subtiny{0}{0}{nm}P\subtiny{0}{0}{mn}|\suptiny{0}{0}{2}.
\end{equation}
The finite-time dephasing estimate of
Refs.~\cite{reimann2008foundation,short2011equilibration}, applied to the ordered multiset of nonzero Bohr gaps (with multiplicities included in $\Ne(\varepsilon)$), states that
\begin{equation}
\label{eq:app_phi_bound_new}
\|\Phi^{(T)}\|_{\mathrm{op}}
\leq\Ne(\varepsilon)
\left(1+\frac{8\log d_E}{\varepsilon T}\right)
=f(\varepsilon,T).
\end{equation}
No nondegeneracy assumption on the gaps is needed for this inequality: repeated gaps are included in the gap-counting factor. We invoke the cited finite-time estimate at this step; its independent proof is not repeated here. Define the second moment of the unordered-pair amplitudes by
\begin{equation}
\label{eq:app_mixed_S2}
S\subtiny{0}{0}{2}^{(\rho\subtiny{0}{0}{0})}
:=\sum\subtiny{0}{0}{n<m}b\subtiny{0}{0}{nm}\suptiny{0}{0}{2}
=\sum\subtiny{0}{0}{n<m}|\rho\subtiny{0}{0}{nm}P\subtiny{0}{0}{mn}|\suptiny{0}{0}{2}.
\end{equation}
Since $|c\subtiny{0}{0}{mn}|\suptiny{0}{0}{2}=|c\subtiny{0}{0}{nm}|\suptiny{0}{0}{2}$, the sum over ordered pairs is twice the sum over unordered pairs, and
Eqs.~\eqref{eq:app_matrix_norm_new}--\eqref{eq:app_phi_bound_new}
produce
\begin{equation}
\label{eq:app_mixed_a}
\avgT{|\delta\leakP(t)|\suptiny{0}{0}{2}}
\leq 2f(\varepsilon,T)S\subtiny{0}{0}{2}^{(\rho\subtiny{0}{0}{0})}.
\end{equation}

\medskip
\noindent\emph{(b) Positivity estimate.} Since $\rho\subtiny{0}{0}{0}\succeq0$ and $P\succeq0$, their principal $2\times2$ submatrices in the energy basis are positive semidefinite. Therefore, for each $n\ne m$, $|\rho\subtiny{0}{0}{nm}|\suptiny{0}{0}{2}\leq p\subtiny{0}{0}{n}p\subtiny{0}{0}{m},$ and $|P\subtiny{0}{0}{mn}|\suptiny{0}{0}{2}\leq P\subtiny{0}{0}{nn}P\subtiny{0}{0}{mm}$. Multiplying the two nonnegative inequalities gives
\begin{equation}
\label{eq:app_product_bound_new}
|\rho\subtiny{0}{0}{nm}P\subtiny{0}{0}{mn}|\suptiny{0}{0}{2}
\leq(p\subtiny{0}{0}{n}P\subtiny{0}{0}{nn})(p\subtiny{0}{0}{m}P\subtiny{0}{0}{mm}).
\end{equation}
Set $x\subtiny{0}{0}{n}:=p\subtiny{0}{0}{n}P\subtiny{0}{0}{nn}\geq0$, so that
$\sum\subtiny{0}{0}{n}x\subtiny{0}{0}{n}=p\subtiny{-1.8}{-0.5}{\mu}$. Summation over $n\ne m$ yields
\begin{equation}
\label{eq:app_mixed_b}
2S\subtiny{0}{0}{2}^{(\rho\subtiny{0}{0}{0})}
=\sum\subtiny{0}{0}{n\ne m}|\rho\subtiny{0}{0}{nm}P\subtiny{0}{0}{mn}|\suptiny{0}{0}{2}
\leq\sum\subtiny{0}{0}{n\ne m}x\subtiny{0}{0}{n}x\subtiny{0}{0}{m}
=p\subtiny{-1.8}{-0.5}{\mu}\suptiny{0}{0}{2}-\sum\subtiny{0}{0}{n}(p\subtiny{0}{0}{n}P\subtiny{0}{0}{nn})\suptiny{0}{0}{2}.
\end{equation}
Substituting Eq.~\eqref{eq:app_mixed_b} into
Eq.~\eqref{eq:app_mixed_a} establishes the sharper bound
\begin{equation}
\label{eq:app_mixed_refined}
\avgT{|\delta\leakP(t)|\suptiny{0}{0}{2}}
\leq f(\varepsilon,T)
\left[p\subtiny{-1.8}{-0.5}{\mu}\suptiny{0}{0}{2}-\sum\subtiny{0}{0}{n}(p\subtiny{0}{0}{n}P\subtiny{0}{0}{nn})\suptiny{0}{0}{2}\right].
\end{equation}

\medskip
\noindent\emph{(c) General bound in terms of $p\subtiny{-1.8}{-0.5}{\mu}$.} Apply the Cauchy--Schwarz inequality to the vectors
$(x\subtiny{0}{0}{0},\ldots,x_{d_E-1})$ and $(1,\ldots,1)$. We have $p\subtiny{-1.8}{-0.5}{\mu}\suptiny{0}{0}{2}=\left(\sum\subtiny{0}{0}{n}x\subtiny{0}{0}{n}\cdot1\right)\suptiny{0}{0}{2}\leq\left(\sum\subtiny{0}{0}{n}x\subtiny{0}{0}{n}\suptiny{0}{0}{2}\right)\!\left(\sum\subtiny{0}{0}{n}1\suptiny{0}{0}{2}\right)=d_E\sum\subtiny{0}{0}{n}(p\subtiny{0}{0}{n}P\subtiny{0}{0}{nn})\suptiny{0}{0}{2}$. Consequently,
\begin{equation}
\label{eq:app_mixed_cs1}
\sum\subtiny{0}{0}{n}(p\subtiny{0}{0}{n}P\subtiny{0}{0}{nn})\suptiny{0}{0}{2}\geq\frac{p\subtiny{-1.8}{-0.5}{\mu}\suptiny{0}{0}{2}}{d_E}.
\end{equation}
Inserting this lower bound into the negative term
in Eq.~\eqref{eq:app_mixed_refined}, we conclude that
\begin{align}
\label{eq:app_mixed_c}
\avgT{|\delta\leakP(t)|\suptiny{0}{0}{2}}
&\leq f(\varepsilon,T)
\left(1-\frac1{d_E}\right)p\subtiny{-1.8}{-0.5}{\mu}\suptiny{0}{0}{2}.
\end{align}
This proves the finite-time bound in
Theorem~\ref{th:leakage_bounds}.

\medskip
\noindent\emph{(d) Support-sensitive consequence.} Because $P$ is an orthogonal projector of rank $r$,
$\Tr P=r=\sum\subtiny{0}{0}{n}P\subtiny{0}{0}{nn}$. Define the effective dimension of
its diagonal profile in the energy basis by
\begin{equation}
\label{eq:dsupp}
d_{\mathrm{supp}}:=\frac{r\suptiny{0}{0}{2}}{\sum\subtiny{0}{0}{n}P\subtiny{0}{0}{nn}\suptiny{0}{0}{2}}.
\end{equation}
Applying Cauchy--Schwarz to the vectors $(p\subtiny{0}{0}{n})_n$
and $(P\subtiny{0}{0}{nn})_n$, and recalling that
$\sum\subtiny{0}{0}{n}p\subtiny{0}{0}{n}\suptiny{0}{0}{2}=1/\deff$, gives
\begin{align}
\label{eq:app_mixed_d_alpha}
p\subtiny{-1.8}{-0.5}{\mu}\suptiny{0}{0}{2}
&=\left(\sum\subtiny{0}{0}{n}p\subtiny{0}{0}{n}P\subtiny{0}{0}{nn}\right)\suptiny{0}{0}{2} \leq\left(\sum\subtiny{0}{0}{n}p\subtiny{0}{0}{n}\suptiny{0}{0}{2}\right)\!\left(\sum\subtiny{0}{0}{n}P\subtiny{0}{0}{nn}\suptiny{0}{0}{2}\right)=\frac1{\deff}\frac{r\suptiny{0}{0}{2}}{d_{\mathrm{supp}}}.
\end{align}
Using this upper bound for $p\subtiny{-1.8}{-0.5}{\mu}\suptiny{0}{0}{2}$ in
Eq.~\eqref{eq:app_mixed_c} yields
\begin{equation}
\label{eq:mixed_lff_support_bound}
\avgT{|\delta\leakP(t)|\suptiny{0}{0}{2}}
\leq f(\varepsilon,T)
\left(1-\frac1{d_E}\right)
\frac{r\suptiny{0}{0}{2}}{d_{\mathrm{supp}}\deff}.
\end{equation}

\medskip
\noindent\emph{(e) Smallest-eigenvalue bound.} Let $\lambda\subtiny{0}{0}{\min}>0$ be the smallest nonzero eigenvalue
of $\rho\subtiny{0}{0}{0}$, thus,
\begin{equation}
\label{eq:app_lmin_operator_new}
\rho\subtiny{0}{0}{0}-\lambda\subtiny{0}{0}{\min}P =\sum_{\alpha=0}^{r-1}
(\lambda_\alpha-\lambda\subtiny{0}{0}{\min})
\ketbra{\varphi_\alpha}{\varphi_\alpha}\succeq0.
\end{equation}
Thus, $\rho\subtiny{0}{0}{0}\succeq\lambda\subtiny{0}{0}{\min}P$. Taking its expectation value in each energy eigenstate gives
\begin{equation}
\label{eq:app_lmin_diagonal_new}
p\subtiny{0}{0}{n}=\bra{E\subtiny{0}{0}{n}}\rho\subtiny{0}{0}{0}\ket{E\subtiny{0}{0}{n}}
\geq\lambda\subtiny{0}{0}{\min}\bra{E\subtiny{0}{0}{n}}P\ket{E\subtiny{0}{0}{n}}
=\lambda\subtiny{0}{0}{\min}P\subtiny{0}{0}{nn}.
\end{equation}
In particular, $P\subtiny{0}{0}{nn}\leq p\subtiny{0}{0}{n}/\lambda\subtiny{0}{0}{\min}$.
Multiplication by $p\subtiny{0}{0}{n}\geq0$ and summation over $n$ yield
\begin{align}
\label{eq:app_mixed_e_alpha}
p\subtiny{-1.8}{-0.5}{\mu}
&=\sum\subtiny{0}{0}{n}p\subtiny{0}{0}{n}P\subtiny{0}{0}{nn}\leq\frac1{\lambda\subtiny{0}{0}{\min}}\sum\subtiny{0}{0}{n}p\subtiny{0}{0}{n}\suptiny{0}{0}{2}=\frac1{\lambda\subtiny{0}{0}{\min}\deff}.
\end{align}
Substituting its square into Eq.~\eqref{eq:app_mixed_c}
gives the original smallest-eigenvalue estimate, now as a
corollary of the general theorem,
\begin{equation}
\label{eq:mixed_lff_lambda_min_bound}
\avgT{|\delta\leakP(t)|\suptiny{0}{0}{2}}
\leq f(\varepsilon,T)
\left(1-\frac1{d_E}\right)
\frac1{\lambda\subtiny{0}{0}{\min}\suptiny{0}{0}{2}\deff\suptiny{0}{0}{2}}.
\end{equation}
For comparison, the same conclusion follows from the support-sensitive bound. Indeed,
Eq.~\eqref{eq:app_lmin_diagonal_new} implies
$p\subtiny{0}{0}{n}\suptiny{0}{0}{2}\geq\lambda\subtiny{0}{0}{\min}\suptiny{0}{0}{2}P\subtiny{0}{0}{nn}\suptiny{0}{0}{2}$, and therefore
\begin{equation}
\label{eq:app_lmin_dsupp_new}
\frac1{\deff}=\sum\subtiny{0}{0}{n}p\subtiny{0}{0}{n}\suptiny{0}{0}{2}\geq\lambda\subtiny{0}{0}{\min}\suptiny{0}{0}{2}\sum\subtiny{0}{0}{n}P\subtiny{0}{0}{nn}\suptiny{0}{0}{2}=\lambda\subtiny{0}{0}{\min}\suptiny{0}{0}{2}
\frac{r\suptiny{0}{0}{2}}{d_{\mathrm{supp}}}.
\end{equation}
Rearranging this inequality,
\begin{equation}
\label{eq:app_lmin_dsupp_lower_new}
d_{\mathrm{supp}}
\geq r\suptiny{0}{0}{2}\lambda\subtiny{0}{0}{\min}\suptiny{0}{0}{2}\deff.
\end{equation}
Insertion into Eq.~\eqref{eq:mixed_lff_support_bound}
again yields Eq.~\eqref{eq:mixed_lff_lambda_min_bound}.

\medskip
\noindent\emph{(f) Rank-independent bound.}

The finite-time spectral estimate in
Eq.~\eqref{eq:app_mixed_a} also allows us
to derive a variance bound independent
of the rank of the initial state.
Since $\rho\subtiny{0}{0}{0}$ is positive semidefinite,
its matrix elements satisfy
$|\rho\subtiny{0}{0}{nm}|\suptiny{0}{0}{2}\leq p\subtiny{0}{0}{n}p\subtiny{0}{0}{m}$.
Furthermore, using
$2p\subtiny{0}{0}{n}p\subtiny{0}{0}{m}\leq p\subtiny{0}{0}{n}\suptiny{0}{0}{2}+p\subtiny{0}{0}{m}\suptiny{0}{0}{2}$
and the symmetry under
$n\leftrightarrow m$, we obtain
\begin{equation}
\label{eq:app_rank_independent_step}
\sum\subtiny{0}{0}{n\ne m}|\rho\subtiny{0}{0}{nm}P\subtiny{0}{0}{mn}|\suptiny{0}{0}{2}
\leq
\sum\subtiny{0}{0}{n} p\subtiny{0}{0}{n}\suptiny{0}{0}{2}
\sum\subtiny{0}{0}{m\ne n}|P\subtiny{0}{0}{nm}|\suptiny{0}{0}{2}.
\end{equation}
To evaluate the inner sum, we use
the projector identities $P^\dagger=P$
and $P\suptiny{0}{0}{2}=P$. Inserting the completeness relation of the energy eigenbasis between the two operators gives
\begin{equation}
\label{eq:app_projector_matrix}
(P\suptiny{0}{0}{2})_{nn}
=\sum\subtiny{0}{0}{m} P\subtiny{0}{0}{nm}P\subtiny{0}{0}{mn}
=\sum\subtiny{0}{0}{m}|P\subtiny{0}{0}{nm}|\suptiny{0}{0}{2}=P\subtiny{0}{0}{nn}.
\end{equation}
Here, $(P\suptiny{0}{0}{2})_{nn}$ denotes the
diagonal matrix element of the operator
$P\suptiny{0}{0}{2}$. Separating the term $m=n$, we obtain
\begin{equation}
\label{eq:app_projector_identity}
\sum\subtiny{0}{0}{m\ne n}|P\subtiny{0}{0}{nm}|\suptiny{0}{0}{2}
=P\subtiny{0}{0}{nn}-P\subtiny{0}{0}{nn}\suptiny{0}{0}{2}
=P\subtiny{0}{0}{nn}(1-P\subtiny{0}{0}{nn}).
\end{equation}
Since $P$ is an orthogonal projector,
$0\leq P\subtiny{0}{0}{nn}\leq1$, and therefore $P\subtiny{0}{0}{nn}(1-P\subtiny{0}{0}{nn})\leq1/4$.
Substituting
Eq.~\eqref{eq:app_projector_identity}
into Eq.~\eqref{eq:app_rank_independent_step}
yields
\begin{equation}
\label{eq:app_rank_independent_final}
\sum\subtiny{0}{0}{n\ne m}|\rho\subtiny{0}{0}{nm}P\subtiny{0}{0}{mn}|\suptiny{0}{0}{2}
\leq
\sum\subtiny{0}{0}{n} p\subtiny{0}{0}{n}\suptiny{0}{0}{2}P\subtiny{0}{0}{nn}(1-P\subtiny{0}{0}{nn}).
\end{equation}
The right-hand side satisfies
$\sum\subtiny{0}{0}{n} p\subtiny{0}{0}{n}\suptiny{0}{0}{2}P\subtiny{0}{0}{nn}(1-P\subtiny{0}{0}{nn})
\leq \frac14\sum\subtiny{0}{0}{n} p\subtiny{0}{0}{n}\suptiny{0}{0}{2}
=1/(4\deff)$.
Combining this result with
Eq.~\eqref{eq:app_mixed_a}, we conclude that
\begin{equation}
\label{eq:app_unconditional_variance}
\avgT{|\delta\leakP(t)|\suptiny{0}{0}{2}}
\leq\frac{f(\varepsilon,T)}{4\deff}.
\end{equation}
This inequality holds for arbitrary
mixed initial states, without imposing
additional conditions on their rank
or smallest nonzero eigenvalue.

\medskip
\noindent\emph{(g) Pure-state specialization.} If $\rho\subtiny{0}{0}{0}=\ketbra{\psi\subtiny{0}{0}{0}}{\psi\subtiny{0}{0}{0}}$, then $P=\rho\subtiny{0}{0}{0}$,
$P\subtiny{0}{0}{nn}=p\subtiny{0}{0}{n}$, and
\begin{equation}
\label{eq:app_pure_alpha_new}
p\subtiny{-1.8}{-0.5}{\mu}=\sum\subtiny{0}{0}{n}p\subtiny{0}{0}{n}\suptiny{0}{0}{2}=\frac1{\deff}.
\end{equation}
Furthermore,
$c\subtiny{0}{0}{nm}=\rho\subtiny{0}{0}{nm}P\subtiny{0}{0}{mn}=p\subtiny{0}{0}{n}p\subtiny{0}{0}{m}\geq0$; Equation~\eqref{eq:app_mixed_c} then reduces to the pure-state bound in Eq.~\eqref{eq:leakZ_bound}.
\end{proof}

\section{General Moments of the Uniform Dirichlet Distribution}
\label{appsec:General_moments_dirichlet}
In this appendix, we derive the moments of the population vector $\vec{p} = (p\subtiny{0}{0}{1},\ldots,p_{d\subtiny{0}{0}{E}})$ that is Haar-typical within an energy shell of dimension $d\subtiny{0}{0}{E}$, and satisfying, $\frac{1}{\deff}=\sum_{i=1}^{d\subtiny{0}{0}{E}}p\subtiny{0}{0}{i}\suptiny{0}{0}{2}.$ Let $\ket{\psi}$ be Haar-distributed on the unit sphere of a $d\subtiny{0}{0}{E}$-dimensional subspace and expanded in the energy eigenbasis as $\ket{\psi}=\sum_{i=1}^{d\subtiny{0}{0}{E}} c\subtiny{0}{0}{i} \ket{E\subtiny{0}{0}{i}},$ $p\subtiny{0}{0}{i} := |c\subtiny{0}{0}{i}|\suptiny{0}{0}{2}.$ A convenient construction of a Haar-random state consists of sampling independent complex Gaussian variables $z\subtiny{0}{0}{i}\sim\mathcal{N}_{\mathbb{C}}(0,1)$ and normalizing,
\begin{equation}
\ket{\psi}
=
\frac{1}{\sqrt{\sum_{j=1}^{d\subtiny{0}{0}{E}}|z\subtiny{0}{0}{j}|\suptiny{0}{0}{2}}}
\sum_{i=1}^{d\subtiny{0}{0}{E}} z\subtiny{0}{0}{i} \ket{E\subtiny{0}{0}{i}},
\quad
p\subtiny{0}{0}{i}=\frac{|z\subtiny{0}{0}{i}|\suptiny{0}{0}{2}}{\sum_{j=1}^{d\subtiny{0}{0}{E}}|z\subtiny{0}{0}{j}|\suptiny{0}{0}{2}}.
\end{equation}
Writing $X\subtiny{0}{0}{i} := |z\subtiny{0}{0}{i}|\suptiny{0}{0}{2}$ and $S := \sum\subtiny{0}{0}{j} X\subtiny{0}{0}{j}$, the populations take the form $p\subtiny{0}{0}{i} = X\subtiny{0}{0}{i}/S$. {Since the $X\subtiny{0}{0}{i}$ are independent, identically distributed exponential random variables, the normalized populations follow a symmetric Dirichlet distribution with unit parameters, namely}
\begin{equation}
\vec p \sim \mathrm{Dirichlet}_{d\subtiny{0}{0}{E}}(a\subtiny{0}{0}{1},\ldots,a_{d\subtiny{0}{0}{E}}),
\,\text{ with:  }
a\subtiny{0}{0}{1}=\cdots=a_{d\subtiny{0}{0}{E}}=1,
\end{equation}
that is,
\begin{equation}
\vec p \sim \mathrm{Dirichlet}_{d\subtiny{0}{0}{E}}(1,\ldots,1).
\end{equation}
That is, the squared amplitudes of a Haar-random pure state are uniformly distributed over the simplex
\begin{equation}
\Delta_{d\subtiny{0}{0}{E}-1}
=
\Big\{(p\subtiny{0}{0}{1},\dots,p_{d\subtiny{0}{0}{E}})\in\mathbb{R}^{d\subtiny{0}{0}{E}}:
p\subtiny{0}{0}{i}\ge 0,\;
\sum_{i=1}^{d\subtiny{0}{0}{E}} p\subtiny{0}{0}{i}=1
\Big\}.
\end{equation}
Equivalently, the Haar measure on the unit sphere induces the natural uniform measure on populations after projection to $|c\subtiny{0}{0}{i}|\suptiny{0}{0}{2}$ \cite{BengtssonZyczkowski2017,NgTianTang2011}. More generally, if $\vec p\sim\mathrm{Dirichlet}_{d\subtiny{0}{0}{E}}(a\subtiny{0}{0}{1},\ldots,a_{d\subtiny{0}{0}{E}})$ and $r\subtiny{0}{0}{1},\ldots,r_{d\subtiny{0}{0}{E}}$ are nonnegative integers, then
\begin{equation}
\mathbb{E}\!\left[\prod_{i=1}^{d\subtiny{0}{0}{E}} p\subtiny{0}{0}{i}^{r\subtiny{0}{0}{i}}\right]
=
\frac{\prod_{i=1}^{d\subtiny{0}{0}{E}}(a\subtiny{0}{0}{i})^{(r\subtiny{0}{0}{i})}}
{(a_+)^{(\sum\subtiny{0}{0}{i} r\subtiny{0}{0}{i})}},
\quad
a_+ := \sum_{i=1}^{d\subtiny{0}{0}{E}} a\subtiny{0}{0}{i},
\label{eq:dirichlet-moments}
\end{equation}
where $(x)^{(m)}$ denotes the rising factorial (Pochhammer symbol),
\begin{equation}
(x)^{(m)} := x(x+1)\cdots(x+m-1),
\quad
(x)^{(0)} := 1,
\end{equation}
equivalently $(x)^{(m)}=\Gamma(x+m)/\Gamma(x)$.

\subsection*{Uniform Case \texorpdfstring{$a\subtiny{0}{0}{i}=1$}{ai=1} for all \texorpdfstring{$i=1,\ldots,d\subtiny{0}{0}{E}$}{i=1,\ldots,d\subtiny{0}{0}{E}}}

\noindent For $a\subtiny{0}{0}{i} = 1$ for all $i=1,\dots,d\subtiny{0}{0}{E}$, and $a_+ := \sum_{i=1}^{d\subtiny{0}{0}{E}} a\subtiny{0}{0}{i} = d\subtiny{0}{0}{E},$ we have
\begin{equation}
\mathbb{E}(p\subtiny{0}{0}{i}\suptiny{0}{0}{m})
=
\frac{m!}{d\subtiny{0}{0}{E}(d\subtiny{0}{0}{E}+1)\cdots(d\subtiny{0}{0}{E}+m-1)}.
\end{equation}

\subsection*{First Moment of \texorpdfstring{$M\subtiny{0}{0}{1}=\sum\subtiny{0}{0}{i} p\subtiny{0}{0}{i}\suptiny{0}{0}{2}$}{S}}

\noindent Using Eq.~\eqref{eq:dirichlet-moments}, $\mathbb{E}(p\subtiny{0}{0}{i}\suptiny{0}{0}{2}) = \frac{2}{d\subtiny{0}{0}{E}(d\subtiny{0}{0}{E}+1)}.$ By symmetry we obtain
\begin{equation}
\mathbb{E}\!\left[\sum_{i=1}^{d\subtiny{0}{0}{E}} p\subtiny{0}{0}{i}\suptiny{0}{0}{2}\right]
=
\frac{2}{d\subtiny{0}{0}{E}+1}.
\label{eq:resultado_lemma_comentado}
\end{equation}

\subsection*{Second Moment \texorpdfstring{$M\subtiny{0}{0}{1}\suptiny{0}{0}{2}$}{M1^2}}

\noindent The second moment of $M\subtiny{0}{0}{1}$ is defined as $M\subtiny{0}{0}{1}\suptiny{0}{0}{2},$ and it is written as $M\subtiny{0}{0}{1}\suptiny{0}{0}{2}
=
\sum_{i=1}^{d\subtiny{0}{0}{E}} p\subtiny{0}{0}{i}\suptiny{0}{1}{4} + 2\sum_{1\le i<j\le d\subtiny{0}{0}{E}} p\subtiny{0}{0}{i}\suptiny{0}{0}{2} p\subtiny{0}{0}{j}\suptiny{0}{0}{2}.$ To calculate its expectation value $\mathbb{E}(M\subtiny{0}{0}{1}\suptiny{0}{0}{2})$, by linearity of expectation, we start by deriving the following expected values
\begin{equation}
\mathbb{E}(p\subtiny{0}{0}{i}\suptiny{0}{1}{4})
=
\frac{24}{d\subtiny{0}{0}{E}(d\subtiny{0}{0}{E}+1)(d\subtiny{0}{0}{E}+2)(d\subtiny{0}{0}{E}+3)},
\quad
\mathbb{E}(p\subtiny{0}{0}{i}\suptiny{0}{0}{2} p\subtiny{0}{0}{j}\suptiny{0}{0}{2})
=
\frac{4}{d\subtiny{0}{0}{E}(d\subtiny{0}{0}{E}+1)(d\subtiny{0}{0}{E}+2)(d\subtiny{0}{0}{E}+3)}.
\end{equation}
Hence,
\begin{equation}
\label{eq:ES2}
\mathbb{E}(M\subtiny{0}{0}{1}\suptiny{0}{0}{2})
=
\frac{4(d\subtiny{0}{0}{E}+5)}{(d\subtiny{0}{0}{E}+1)(d\subtiny{0}{0}{E}+2)(d\subtiny{0}{0}{E}+3)}.
\end{equation}

\subsection*{Variance and concentration}

\noindent From the above,
\begin{equation}
\mathrm{Var}\!\Big(\frac{1}{\deff}\Big)
=
\frac{4(d\subtiny{0}{0}{E}-1)}{(d\subtiny{0}{0}{E}+1)\suptiny{0}{0}{2}(d\subtiny{0}{0}{E}+2)(d\subtiny{0}{0}{E}+3)}.
\end{equation}
For large $d\subtiny{0}{0}{E}$,
\begin{equation}
\mathbb{E}\!\Big(\frac{1}{\deff}\Big)
\sim \frac{2}{d\subtiny{0}{0}{E}},
\quad
\mathrm{Var}\!\Big(\frac{1}{\deff}\Big)
\sim \frac{4}{d\subtiny{0}{0}{E}\suptiny{0}{1}{3}}.
\end{equation}
Thus the standard deviation scales as $d\subtiny{0}{0}{E}^{-3/2}$ while the mean scales as $d\subtiny{0}{0}{E}\suptiny{0}{0}{-1}$, implying relative fluctuations of order $d\subtiny{0}{0}{E}^{-1/2}$. Consequently, for large energy shells, $1/\deff$ is sharply concentrated near ${2/d\subtiny{0}{0}{E}}$, and $\deff$ is typically of order $d\subtiny{0}{0}{E}$.

\subsection{Proof of Leakage concentration bound presented in Theorem~\ref{thm:leakage_concentration_dir}}
\label{app:leakage_concentration}

\begin{theorem}[Leakage concentration from the variance bound]\label{thmapp:leakage_concentration_dir}
Assuming that, for a given time window $[0,T]$,
for any fixed initial state $\ket{\psi\subtiny{0}{0}{0}}$ drawn from the complex Haar ensemble (so that the populations $p\subtiny{0}{0}{i}=|c\subtiny{0}{0}{i}|\suptiny{0}{0}{2}$ are $\mathrm{Dirichlet}(1,\ldots,1)$). Then:\\
\noindent\emph{(A) Time probability (fixed state).}
For any fixed state $\ket{\psi\subtiny{0}{0}{0}}$, for any $\epsilon>0$ and $t\sim\mathrm{Unif}[0,T]$,
\begin{equation}\label{eq:A-timeAPP}
\Pr_t\!\Big(\,|\leakz(t)-\leakz\suptiny{0}{0}{\infty}|\ \ge\ \epsilon\,\Big)
\ \le\
\frac{f(\varepsilon,T)}{\epsilon\suptiny{0}{0}{2}}\,
\Bigl(1-\frac{1}{d\subtiny{0}{0}{E}}\Bigr)\,
\frac{1}{\deff\suptiny{0}{1}{2}}.
\end{equation}\\
\noindent\emph{(B) Ensemble probability (random state).}
Let $\ket{\psi\subtiny{0}{0}{0}}$ be distributed according to the Haar measure. For any fixed averaging time $T$ and any $\eta>0$, 
\begin{align}\label{eq:B-ensembleAPP}
&\Pr_{\ket{\psi\subtiny{0}{0}{0}}}\!\Big(\,\Big\langle\big|\leakz(t)-\leakz\suptiny{0}{0}{\infty}\big|\suptiny{0}{0}{2}\Big\rangle\subtiny{0}{0}{T}\ \ge\ \eta\,\Big)
\le F(\eta,T,d\subtiny{0}{0}{E}),
\end{align}
where $F(\eta,T,d\subtiny{0}{0}{E}) = \frac{f(\varepsilon,T)}{\eta}\,
\Bigl(1-\frac{1}{d\subtiny{0}{0}{E}}\Bigr)g(d\subtiny{0}{0}{E})$, and 
$g(d_E)=\frac{4(d_E+5)}{(d_E+1)(d_E+2)(d_E+3)}$.
\end{theorem}

\begin{proof}
(A) Define the non-negative random variable
\begin{equation}
X(t)\ :=\ \big|\mathrm{\leakz}(t)-\mathrm{\leakz}^\infty\big|\suptiny{0}{0}{2} \ \ge\ 0.
\end{equation}
Markov's inequality for $X$ with threshold $\varepsilon\suptiny{0}{0}{2}>0$ gives
\begin{equation}
\Pr_t\!\big(X\ge \varepsilon\suptiny{0}{0}{2}\big)\ \le\ \frac{\langle X\rangle_T}{\varepsilon\suptiny{0}{0}{2}}
\ =\
\frac{\Big\langle\big|\mathrm{\leakz}(t)-\mathrm{\leakz}^\infty\big|\suptiny{0}{0}{2}\Big\rangle_T}{\varepsilon\suptiny{0}{0}{2}}.
\end{equation}
Using the variance bound given in Eq.~\eqref{eq:leakZ_bound} we obtain
\begin{equation}
\Pr_t\!\Big(\,|\mathrm{\leakz}(t)-\mathrm{\leakz}^\infty|\ \ge\ \varepsilon\,\Big)
\ \le\
\frac{f(\varepsilon,T)}{\varepsilon\suptiny{0}{0}{2}}\,
\Bigl(1-\frac{1}{d\subtiny{0}{0}{E}}\Bigr)\,
\frac{1}{\deff\suptiny{0}{1}{2}},
\end{equation}
which is precisely Eq.~\eqref{eq:A-time}.
\end{proof}
\begin{proof}
(B) By definition $X(t)\ge 0$, its time average
\begin{equation}
Z(\ket{\psi\subtiny{0}{0}{0}})\ :=\ \Big\langle\big|\mathrm{\leakz}(t)-\mathrm{\leakz}^\infty\big|\suptiny{0}{0}{2}\Big\rangle_T
\end{equation}
is also non-negative for each fixed $\ket{\psi\subtiny{0}{0}{0}}$. {Averaging the mean-square bound in Eq.~\eqref{eq:leakZ_bound} over the Haar ensemble, we obtain}
\begin{equation}
\mathbb{E}_{\ket{\psi\subtiny{0}{0}{0}}}\!\big[Z(\ket{\psi\subtiny{0}{0}{0}})\big]\le
f(\varepsilon,T)\,\Bigl(1-\tfrac{1}{d\subtiny{0}{0}{E}}\Bigr)\,
\mathbb{E}_{\ket{\psi\subtiny{0}{0}{0}}}\!\Bigg[\frac{1}{\deff\suptiny{0}{0}{2}}\Bigg].
\end{equation}
Now apply Markov's inequality to the non-negative random variable $Z(\ket{\psi\subtiny{0}{0}{0}})$ with threshold $\eta>0$:
\begin{eqnarray}
\Pr_{\ket{\psi\subtiny{0}{0}{0}}}\big(Z(\ket{\psi\subtiny{0}{0}{0}})\ge\eta\big)
&\le& \frac{\mathbb{E}_{\ket{\psi\subtiny{0}{0}{0}}}[Z(\ket{\psi\subtiny{0}{0}{0}})]}{\eta},\\
&\le& \frac{f(\varepsilon,T)}{\eta}\,\Bigl(1-\tfrac{1}{d\subtiny{0}{0}{E}}\Bigr)\,
\mathbb{E}_{\ket{\psi\subtiny{0}{0}{0}}}\!\Bigg[\frac{1}{\deff\suptiny{0}{0}{2}}\Bigg].
\end{eqnarray}
Finally, using Eq.~\eqref{eq:ES2}, one obtains the explicit Haar value
\begin{equation}
\mathbb{E}_{\ket{\psi\subtiny{0}{0}{0}}}\!\Bigg[\frac{1}{\deff\suptiny{0}{0}{2}}\Bigg]
=
\frac{4(d\subtiny{0}{0}{E}+5)}{(d\subtiny{0}{0}{E}+1)(d\subtiny{0}{0}{E}+2)(d\subtiny{0}{0}{E}+3)},
\end{equation}
from which Eq.~\eqref{eq:B-ensemble} follows.
\end{proof}

\subsection{Extension of Theorem~\ref{thm:leakage_concentration_dir}
to mixed initial states}
\label{app:leakage_mixed_concentration}

\noindent\emph{(A) Temporal concentration.}
The temporal concentration bound extends directly to arbitrary mixed initial states. For a fixed $\rho\subtiny{0}{0}{0}$
of rank $r$, with $P$ denoting its support projector, consider the non-negative random variable $X(t)=|\delta\leakP(t)|\suptiny{0}{0}{2}.$ For $t\sim\mathrm{Unif}[0,T]$ and any $\delta>0$, Markov's inequality gives
\begin{equation}
\Pr_t\left(
|\delta\leakP(t)|\geq\delta
\right)=\Pr_t\left(X(t)\geq\delta\suptiny{0}{0}{2}\right) \leq\frac{\mathbb E_t[X(t)]}{\delta\suptiny{0}{0}{2}}.
\end{equation}
Since $\mathbb E_t[X(t)] = \avgT{|\delta\leakP(t)|\suptiny{0}{0}{2}},$ the variance bound established for Theorem~\ref{th:leakage_bounds} and proved in Appendix~\ref{appsec:theorem_leak_variance}  (in Eq.~\eqref{eq:mixed_lff_lambda_min_bound}), yields
\begin{equation}
\Pr_t\left(
|\delta\leakP(t)|\geq\delta
\right)\leq
\frac{f(\varepsilon,T)}{\delta\suptiny{0}{0}{2}}
\left(1-\frac{1}{d_E}\right)
\frac{1}{\lambda\subtiny{0}{0}{\min}\suptiny{0}{0}{2}\deff\suptiny{0}{0}{2}}.
\label{eq:mixed_concentration_A}
\end{equation}
Here, $\lambda\subtiny{0}{0}{\min}$ is the smallest nonzero
eigenvalue of $\rho\subtiny{0}{0}{0}$. For pure initial states, $r=1$ and $\lambda\subtiny{0}{0}{\min}=1$, recovering
Eq.~\eqref{eq:A-time}.

\medskip
\noindent\emph{(B) Ensemble concentration.}
Extending part (B) to mixed initial states requires specifying a probability measure $\nu$ over density operators. Unlike the pure-state case, no unique unitarily invariant measure exists on the full
space of mixed quantum states
\cite{ZyczkowskiSommers2001,BengtssonZyczkowski2017}. For a chosen ensemble $\nu$, consider
\begin{equation}
Z_T(\rho\subtiny{0}{0}{0})
=
\avgT{|\delta\leakP(t)|\suptiny{0}{0}{2}},
\end{equation}
where $P$ is the support projector associated
with each realization of $\rho\subtiny{0}{0}{0}$. Applying Markov's inequality to the ensemble gives
\begin{equation}
\Pr_{\rho\subtiny{0}{0}{0}\sim\nu}
\left(Z_T(\rho\subtiny{0}{0}{0})\geq\eta\right)
\leq
\frac{\mathbb E_\nu[Z_T(\rho\subtiny{0}{0}{0})]}{\eta}.
\end{equation}
An explicit concentration bound therefore requires evaluating the ensemble average of the temporal variance. For example, a fixed-spectrum Haar ensemble may be constructed by sampling unitary conjugations of a density operator with prescribed eigenvalues. The corresponding averages can be evaluated using
unitary Haar integration techniques \cite{CollinsSniady2006}. Other possibilities include induced measures obtained by partial tracing over an auxiliary system  \cite{ZyczkowskiSommers2001}. The resulting concentration estimates depend on the chosen ensemble and its support structure.
In particular, full-rank states yield $P=\id$
and an identically vanishing LFF. We therefore
restrict the present extension to temporal
concentration for arbitrary fixed mixed states. Investigating ensemble-dependent concentration bounds for mixed states of different ranks and statistical distributions constitutes an interesting direction for future research, which lies beyond the scope of the present work.

\section{Dynamical and spectral structure of equilibration}
\label{appsec:deffdyn}

In this appendix, we establish the mixed-state extension of Theorem~\ref{thm:LFF_variance_clean} and derive the corresponding dynamical and spectral descriptors. The construction applies to arbitrary initial density
operators and accounts for possible degeneracies among positive Bohr frequencies. We subsequently recover the pure-state results for $\deffdyn$, $\deffspec$, and
$\Hpow$, including the expression of $\deffspec$ in terms of $\deff$ and higher-order moments of the initial energy distribution.

\subsection{Gap-resolved leakage amplitudes}

Consider an arbitrary initial density operator
$\rho\subtiny{0}{0}{0}$, with support projector $P=\operatorname{supp}(\rho\subtiny{0}{0}{0})$, evolving under a
time-independent Hamiltonian with nondegenerate energy levels. In the energy eigenbasis, let
\begin{equation}
\rho\subtiny{0}{0}{nm}:=\bra{E\subtiny{0}{0}{n}}\rho\subtiny{0}{0}{0}\ket{E\subtiny{0}{0}{m}}, \qquad P\subtiny{0}{0}{mn}:=\bra{E\subtiny{0}{0}{m}}P\ket{E\subtiny{0}{0}{n}}.
\end{equation}
Since $\rho(t)=e^{-iHt}\rho\subtiny{0}{0}{0}e^{iHt}$, the leakage fidelity function is
\begin{equation}
\label{eq:app_mixed_leakage}
\leakP(t) = 1-\sum\subtiny{0}{0}{n,m}
\rho\subtiny{0}{0}{nm}P\subtiny{0}{0}{mn}e^{-i\omega\subtiny{0}{0}{nm}t},
\end{equation}
where $\omega\subtiny{0}{0}{nm}=E\subtiny{0}{0}{n}-E\subtiny{0}{0}{m}$. Its long-time average is
\begin{equation}
\label{eq:app_mixed_leakage_inf}
\leakP\suptiny{0}{0}{\infty} = 1-\sum\subtiny{0}{0}{n}\rho_{nn}P\subtiny{0}{0}{nn} = 1-\Tr(P\eqst).
\end{equation}
Subtracting the long-time average yields
\begin{equation}
\label{eq:app_mixed_delta}
\delta\leakP(t) = -\sum\subtiny{0}{0}{n\neq m}
\rho\subtiny{0}{0}{nm}P\subtiny{0}{0}{mn}e^{-i\omega\subtiny{0}{0}{nm}t}.
\end{equation}
Introducing the complex coefficients
\begin{equation}
\label{eq:app_mixed_coefficients}
z\subtiny{0}{0}{nm}
:=
\rho\subtiny{0}{0}{nm}P\subtiny{0}{0}{mn}
=
b\subtiny{0}{0}{nm}e^{i\phi\subtiny{0}{0}{nm}},
\qquad
b\subtiny{0}{0}{nm}:=|z\subtiny{0}{0}{nm}|,
\end{equation}
and using $z\subtiny{0}{0}{mn}=z\subtiny{0}{0}{nm}^*$, we obtain the real-phasor decomposition
\begin{equation}
\label{eq:app_mixed_phasor}
\delta\leakP(t) = -2\sum\subtiny{0}{0}{n<m}
b\subtiny{0}{0}{nm}\cos(\omega\subtiny{0}{0}{nm}t-\phi\subtiny{0}{0}{nm}).
\end{equation}
Unlike the pure-state coefficients $p\subtiny{0}{0}{n}p\subtiny{0}{0}{m}$,
the quantities $z\subtiny{0}{0}{nm}$ generally contain nontrivial phases. These phases become relevant when different
energy-level pairs contribute to the same Bohr
frequency. Let the energies be ordered increasingly and denote by $\omega\subtiny{0}{0}{k}>0$ the distinct positive gaps. For each frequency, the set $\Gamma\subtiny{0}{0}{k}$ is defined in Eq.~\eqref{eq:Gamma_alpha}. We introduce the corresponding complex amplitude
\begin{equation}
\label{eq:mixed_gap_amplitude}
C\subtiny{0}{0}{k} := \sum_{(n,m)\in\Gamma\subtiny{0}{0}{k}}z\subtiny{0}{0}{nm}, \qquad
A\subtiny{0}{0}{k}:=|C\subtiny{0}{0}{k}|.
\end{equation}
The centered leakage signal can then be written as
\begin{equation}
\label{eq:mixed_gap_signal}
\delta\leakP(t)
=
-2\operatorname{Re}
\left[
\sum_{k=1}^{N\subtiny{0}{0}{\omega}}
C\subtiny{0}{0}{k} e^{i\omega\subtiny{0}{0}{k} t}
\right],
\end{equation}
where $N\subtiny{0}{0}{\omega}$ is the number of distinct positive Bohr frequencies. The positive exponent follows from $\omega\subtiny{0}{0}{k}=E\subtiny{0}{0}{m}-E\subtiny{0}{0}{n}>0$ for $n<m$, consistently with Eq.~\eqref{eq:app_mixed_phasor}.
Equation~\eqref{eq:mixed_gap_amplitude} shows
that contributions sharing the same frequency must be added coherently. For nondegenerate positive gaps, each set $\Gamma\subtiny{0}{0}{k}$ contains a single pair, so that
$C\subtiny{0}{0}{k}=z\subtiny{0}{0}{nm}$ and $A\subtiny{0}{0}{k}=b\subtiny{0}{0}{nm}$. With the Fourier convention used in the main text, the spectral power associated with each distinct positive frequency is therefore proportional to $|C\subtiny{0}{0}{k}|\suptiny{0}{0}{2}$. The discrete spectral-power distribution is constructed from these amplitudes, without squaring the Dirac distributions.
\subsection{Dynamical dimension and exact variance}

For degenerate positive Bohr gaps, the
contributions sharing the same frequency are
combined coherently into the amplitudes
$C\subtiny{0}{0}{k}$, defined in
Eq.~\eqref{eq:mixed_gap_amplitude}.
Their moduli $A\subtiny{0}{0}{k}=|C\subtiny{0}{0}{k}|$ define, for $\sum_k A\subtiny{0}{0}{k}>0$, the normalized distribution over the
active frequency channels,
\begin{equation}
\label{eq:mixed_dynamic_weights}
q\subtiny{0}{0}{k}^{\mathrm{dyn}}
(\rho\subtiny{0}{0}{0},P)
:=
\frac{A\subtiny{0}{0}{k}}{
\displaystyle\sum_j A\subtiny{0}{0}{j}},
\qquad
\sum_k q\subtiny{0}{0}{k}^{\mathrm{dyn}}
(\rho\subtiny{0}{0}{0},P)=1.
\end{equation}
The dynamical dimension is the inverse
participation ratio of these amplitude weights,
\begin{equation}
\label{eq:mixed_dynamic_dimension}
\deffdyn(\rho\subtiny{0}{0}{0},P)
:=
\frac{1}{
\displaystyle\sum_k
[q\subtiny{0}{0}{k}^{\mathrm{dyn}}
(\rho\subtiny{0}{0}{0},P)]\suptiny{0}{0}{2}}
=
\frac{
\left(\displaystyle\sum_k A\subtiny{0}{0}{k}\right)
\suptiny{0}{0}{2}
}{
\displaystyle\sum_k
A\subtiny{0}{0}{k}\suptiny{0}{0}{2}
}.
\end{equation}
We now establish its exact relation to the
long-time variance of the LFF.
\begin{proof}[Proof of
Theorem~\ref{thm:LFF_variance_clean}]
\label{app:thm3_misto}
From Eq.~\eqref{eq:mixed_gap_amplitude} and from Eq.~\eqref{eq:mixed_gap_signal}, since
the frequencies $\omega\subtiny{0}{0}{k}$
are distinct and strictly positive,
their Fourier modes satisfy
$\avginfty{e^{i(\omega\subtiny{0}{0}{k}
-\omega\subtiny{0}{0}{j})t}} =\delta_{kj}$,
whereas averages involving $\omega\subtiny{0}{0}{k} +\omega\subtiny{0}{0}{j}$ vanish. 
Consequently, all cross terms between distinct
frequencies disappear under long-time averaging,
yielding
\begin{equation}
\label{eq:mixed_gap_variance}
\avginfty{|\delta\leakP(t)|\suptiny{0}{0}{2}}
= 2\sum_k|C\subtiny{0}{0}{k}|\suptiny{0}{0}{2}
= 2\sum_k A\subtiny{0}{0}{k}\suptiny{0}{0}{2}.
\end{equation}
Substituting Eq.~\eqref{eq:mixed_dynamic_dimension},
we obtain
\begin{equation}
\label{eq:app_mixed_variance_deffdyn}
\avginfty{|\delta\leakP(t)|\suptiny{0}{0}{2}}
=
\frac{
2\left(\displaystyle\sum_k
A\subtiny{0}{0}{k}\right)\suptiny{0}{0}{2}
}{
\deffdyn(\rho\subtiny{0}{0}{0},P)
}.
\end{equation}
This identity holds for arbitrary mixed initial states and allows degenerate positive-energy gaps, provided that contributions sharing the same frequency are coherently grouped. For nondegenerate positive gaps, $A\subtiny{0}{0}{k}=b\subtiny{0}{0}{nm}$, and Eq.~\eqref{eq:mixed_gap_variance} reduces to
\begin{equation}
\label{eq:app_mixed_variance}
\avginfty{|\delta\leakP(t)|\suptiny{0}{0}{2}}
=
2\sum\subtiny{0}{0}{n<m}
b\subtiny{0}{0}{nm}\suptiny{0}{0}{2}
=
2S\subtiny{0}{0}{2},
\end{equation}
recovering Lemma~\ref{lem:exact_variance}.
In this case,
\begin{equation}
\label{eq:app_mixed_deffdyn}
\deffdyn(\rho\subtiny{0}{0}{0},P)
=
\frac{S\subtiny{0}{0}{1}\suptiny{0}{0}{2}}
{S\subtiny{0}{0}{2}}.
\end{equation}
Since $\leakP(0)=0$,
Eq.~\eqref{eq:mixed_gap_signal} implies
\begin{equation}
\label{eq:mixed_leakage_mean_phasors}
1-p\subtiny{-1.8}{-0.5}{\mu}
=
2\sum_k\operatorname{Re}\{C\subtiny{0}{0}{k}\}
\leq
2\sum_k A\subtiny{0}{0}{k}.
\end{equation}
For nondegenerate positive gaps,
$\sum_k A\subtiny{0}{0}{k}=S\subtiny{0}{0}{1}$. Since $1-p\subtiny{-1.8}{-0.5}{\mu}\geq0$,
we obtain
\begin{equation}
\avginfty{|\delta\leakP(t)|\suptiny{0}{0}{2}}
= \frac{2S\subtiny{0}{0}{1}\suptiny{0}{0}{2}
}{\deffdyn(\rho\subtiny{0}{0}{0},P)
}\geq \frac{ (1-p\subtiny{-1.8}{-0.5}{\mu})\suptiny{0}{0}{2}}{2\deffdyn(\rho\subtiny{0}{0}{0},P)},
\end{equation}
which proves Eq.~\eqref{eq:variancedeff_dyn_mixed}.
Equality holds when all nonzero coefficients
$\rho\subtiny{0}{0}{nm}P\subtiny{0}{0}{mn}$
are real and nonnegative.
In particular, this condition is satisfied
for pure initial states, since
$\rho\subtiny{0}{0}{nm}P\subtiny{0}{0}{mn}
=p\subtiny{0}{0}{n}p\subtiny{0}{0}{m}\geq0$,
and for states maximally mixed on their support,
$\rho\subtiny{0}{0}{0}=P/r$, since
$\rho\subtiny{0}{0}{nm}P\subtiny{0}{0}{mn}
=|P\subtiny{0}{0}{nm}|\suptiny{0}{0}{2}/r
\geq0$.
In the latter case,
$p\subtiny{0}{0}{n}=P\subtiny{0}{0}{nn}/r$,
and therefore
\begin{equation}
p\subtiny{-1.8}{-0.5}{\mu}
=
\frac{1}{r}\sum_n
P\subtiny{0}{0}{nn}\suptiny{0}{0}{2}
=
\frac{r}{\deff}.
\end{equation}
Substituting this relation into the equality
case recovers Eq.~\eqref{eq:var_P_over_r}.
\end{proof}
At fixed support and spectral weights $\displaystyle\sum_k
A\subtiny{0}{0}{k}$, a larger $\deffdyn$ implies smaller long-time fluctuations. The amplitude normalization and the dynamical dimension therefore provide complementary information about the LFF signal's magnitude.
\subsection{Spectral dimension and Shannon power entropy}

The frequency-resolved spectral power is
defined by
\begin{equation}
\label{eq:mixed_spectral_power}
\spd(\omega\subtiny{0}{0}{k};\rho\subtiny{0}{0}{0},P)
:=
|C\subtiny{0}{0}{k}|\suptiny{0}{0}{2}
=
A\subtiny{0}{0}{k}\suptiny{0}{0}{2}.
\end{equation}
Its normalized distribution is
\begin{equation}
\label{eq:mixed_normalized_power}
w\subtiny{0}{0}{k}(\rho\subtiny{0}{0}{0},P)
:=
\frac{A\subtiny{0}{0}{k}\suptiny{0}{0}{2}}{
\displaystyle\sum\subtiny{0}{0}{j} A\subtiny{0}{0}{j}\suptiny{0}{0}{2}},
\qquad
\sum\subtiny{0}{0}{k} w\subtiny{0}{0}{k}(\rho\subtiny{0}{0}{0},P)=1.
\end{equation}
The corresponding spectral dimension is
\begin{equation}
\label{eq:mixed_spectral_dimension}
\deffspec(\rho\subtiny{0}{0}{0},P)
:=
\frac{1}{
\displaystyle\sum\subtiny{0}{0}{k} w\subtiny{0}{0}{k}(\rho\subtiny{0}{0}{0},P)\suptiny{0}{0}{2}}
=
\frac{
\left(\displaystyle\sum\subtiny{0}{0}{k} A\subtiny{0}{0}{k}\suptiny{0}{0}{2}\right)\suptiny{0}{0}{2}
}{
\displaystyle\sum\subtiny{0}{0}{k} A\subtiny{0}{0}{k}\suptiny{0}{0}{4}
}.
\end{equation}
Unlike $\deffdyn$, which is constructed
from the amplitudes $A\subtiny{0}{0}{k}$, the spectral dimension
quantifies the participation of their squared
amplitudes. For nondegenerate positive gaps,
Eq.~\eqref{eq:mixed_spectral_dimension} reduces to
\begin{equation}
\label{eq:mixed_spectral_dimension_pairs}
\deffspec(\rho\subtiny{0}{0}{0},P)
=
\frac{
\left(\displaystyle\sum\subtiny{0}{0}{n<m}b\subtiny{0}{0}{nm}\suptiny{0}{0}{2}\right)\suptiny{0}{0}{2}
}{
\displaystyle\sum\subtiny{0}{0}{n<m}b\subtiny{0}{0}{nm}\suptiny{0}{0}{4}
}.
\end{equation}
The Shannon power entropy is obtained
from the same normalized spectral distribution,
\begin{equation}
\label{eq:mixed_power_entropy}
\Hpow(\rho\subtiny{0}{0}{0},P)
:=
-\sum\subtiny{0}{0}{k} w\subtiny{0}{0}{k}(\rho\subtiny{0}{0}{0},P)\ln w\subtiny{0}{0}{k}(\rho\subtiny{0}{0}{0},P),
\end{equation}
where $0\ln 0=0$ by continuity. Since the order-two R\'enyi entropy of
the distribution $w\subtiny{0}{0}{k}$ equals
$\ln\deffspec(\rho\subtiny{0}{0}{0},P)$, the ordering of R\'enyi entropies gives
\begin{equation}
\label{eq:mixed_power_entropy_bounds}
\ln\deffspec(\rho\subtiny{0}{0}{0},P)
\leq
\Hpow(\rho\subtiny{0}{0}{0},P)
\leq
\ln N_{\mathrm{act}},
\end{equation}
where $N_{\mathrm{act}}$ denotes the number
of frequencies with nonzero spectral power. The three descriptors thus characterize
different aspects of the same signal:
$\deffdyn$ measures amplitude participation,
$\deffspec$ measures power participation,
and $\Hpow$ measures the Shannon entropy
of the normalized spectral power.

\subsection{Spectral moments for pure initial states}

{For pure initial states, $P=\rho\subtiny{0}{0}{0}$,
the spectral amplitudes reduce to
$b\subtiny{0}{0}{nm}=p\subtiny{0}{0}{n}p\subtiny{0}{0}{m}$,
where $p\subtiny{0}{0}{n}=|c\subtiny{0}{0}{n}|\suptiny{0}{0}{2}$.
We now recover the expression for $\deffspec$
in terms of the higher moments of the
energy-population distribution.}

\begin{prop}[Spectral dimension for pure initial states]
\label{prop:deffspec}

Assuming nondegenerate positive energy gaps,
the inverse participation ratio of the normalized
spectral-power distribution defines
\begin{equation}
\label{eq:def_deffspec_LFF}
\deffspec
:=
\frac{1}{
\displaystyle\sum\subtiny{0}{0}{n<m}
[p\subtiny{0}{0}{\leak}(\omega\subtiny{0}{0}{nm})]\suptiny{0}{0}{2}
}.
\end{equation}

Equivalently, using
Eq.~\eqref{eq:pow_distribution_LFF},
\begin{equation}
\label{eq:deffspec_exact_general_clean}
\deffspec
=
\frac{
\left(
\displaystyle\sum\subtiny{0}{0}{n<m}p\subtiny{0}{0}{n}\suptiny{0}{0}{2}p\subtiny{0}{0}{m}\suptiny{0}{0}{2}
\right)\suptiny{0}{0}{2}
}{
\displaystyle\sum\subtiny{0}{0}{n<m}p\subtiny{0}{0}{n}\suptiny{0}{0}{4}p\subtiny{0}{0}{m}\suptiny{0}{0}{4}
}.
\end{equation}

The spectral dimension can also be expressed as
\begin{equation}
\label{eq:deffspec_exact_app}
\deffspec
=
\frac{
\left(
\frac{1}{\deff\suptiny{0}{0}{2}}-\displaystyle\sum\subtiny{0}{0}{n} p\subtiny{0}{0}{n}\suptiny{0}{0}{4}
\right)\suptiny{0}{0}{2}
}{
2\left[
\left(\displaystyle\sum\subtiny{0}{0}{n} p\subtiny{0}{0}{n}\suptiny{0}{0}{4}\right)\suptiny{0}{0}{2}
-\displaystyle\sum\subtiny{0}{0}{n} p\subtiny{0}{0}{n}\suptiny{0}{0}{8}
\right]
}.
\end{equation}

\end{prop}

\begin{proof}

From Eq.~\eqref{eq:pow_distribution_LFF},
the normalized spectral-power distribution is
\begin{equation}
p\subtiny{0}{0}{\leak}(\omega\subtiny{0}{0}{nm})
=
\frac{
p\subtiny{0}{0}{n}\suptiny{0}{0}{2}p\subtiny{0}{0}{m}\suptiny{0}{0}{2}
}{
\displaystyle\sum_{i<j}p_i\suptiny{0}{0}{2}p\subtiny{0}{0}{j}\suptiny{0}{0}{2}
}.
\end{equation}

Consequently,
\begin{equation}
\sum\subtiny{0}{0}{n<m}
[p\subtiny{0}{0}{\leak}(\omega\subtiny{0}{0}{nm})]\suptiny{0}{0}{2}
=
\frac{
\displaystyle\sum\subtiny{0}{0}{n<m}p\subtiny{0}{0}{n}\suptiny{0}{0}{4}p\subtiny{0}{0}{m}\suptiny{0}{0}{4}
}{
\left(
\displaystyle\sum_{i<j}p\subtiny{0}{0}{i}\suptiny{0}{0}{2}p\subtiny{0}{0}{j}\suptiny{0}{0}{2}
\right)\suptiny{0}{0}{2}
}.
\end{equation}
Taking the reciprocal proves
Eq.~\eqref{eq:deffspec_exact_general_clean}. To obtain the moment representation, we use
\begin{equation}
\sum\subtiny{0}{0}{n<m}p\subtiny{0}{0}{n}\suptiny{0}{0}{2}p\subtiny{0}{0}{m}\suptiny{0}{0}{2}
=
\frac12\left[
\left(\sum\subtiny{0}{0}{n} p\subtiny{0}{0}{n}\suptiny{0}{0}{2}\right)\suptiny{0}{0}{2}
-\sum\subtiny{0}{0}{n} p\subtiny{0}{0}{n}\suptiny{0}{0}{4}
\right],
\end{equation}
and
\begin{equation}
\sum\subtiny{0}{0}{n<m}p\subtiny{0}{0}{n}\suptiny{0}{0}{4}p\subtiny{0}{0}{m}\suptiny{0}{0}{4}
=
\frac12\left[
\left(\sum\subtiny{0}{0}{n} p\subtiny{0}{0}{n}\suptiny{0}{0}{4}\right)\suptiny{0}{0}{2}
-\sum\subtiny{0}{0}{n} p\subtiny{0}{0}{n}\suptiny{0}{0}{8}
\right].
\end{equation}

Substituting these identities into
Eq.~\eqref{eq:deffspec_exact_general_clean},
and using $\sum\subtiny{0}{0}{n} p\subtiny{0}{0}{n}\suptiny{0}{0}{2}=1/\deff$,
gives Eq.~\eqref{eq:deffspec_exact_app}.
\end{proof}

{For pure initial states, the Shannon power entropy is recovered by substituting
$p\subtiny{0}{0}{\leak}(\omega\subtiny{0}{0}{nm})$ into
Eq.~\eqref{eq:mixed_power_entropy}.
Thus, the general constructions reduce exactly
to the original definitions of $\deffdyn$,
$\deffspec$, and $\Hpow$.}

{The moment representation also makes
explicit that $\deffspec$ is not determined
by $\deff$ alone, but depends on the fourth
and eighth moments of the initial energy
distribution. Finally, if
$\sum_{n<m}b\subtiny{0}{0}{nm}=0$,
the LFF vanishes identically.
In this case, the normalized amplitude
and power distributions are undefined,
as are $\deffdyn$, $\deffspec$, and $\Hpow$.
This includes full-support initial states
and pure energy eigenstates.}

\subsection{Stability bounds for effective dimensions and $\Hpow$}
\label{app:robustness_bounds}

We derive the stability bounds for the dynamical and spectral effective dimensions under perturbations of their transition amplitudes. We then relate these amplitude errors to perturbations of the initial
density operator and establish analogous continuity bounds for $\deff$ and $\Hpow$. Let $\alpha=(n,m)$, with $n<m$, label the transition channels. For nonnegative amplitudes $b\subtiny{0}{0}{\alpha}$, we define the moments and normalized weights
\begin{equation}
S\subtiny{0}{0}{k}=\sum\subtiny{0}{0}{\alpha} b\subtiny{0}{0}{\alpha}\suptiny{0}{0}{k},
\qquad
q\subtiny{0}{0}{\alpha}\suptiny{0}{0}{(\ell)}
=
\frac{b\subtiny{0}{0}{\alpha}\suptiny{0}{0}{\ell}}{S\subtiny{0}{0}{\ell}},
\qquad \ell=1,2,
\end{equation}
where $S\subtiny{0}{0}{\ell}>0$ is assumed. These weights
form probability distributions,
$\sum\subtiny{0}{0}{\alpha} q\subtiny{0}{0}{\alpha}\suptiny{0}{0}{(\ell)}=1$.
The corresponding inverse participation ratios are
\begin{equation}
\label{eq:robustness_ipr}
d\subtiny{0}{0}{\ell}
=
\frac{S\subtiny{0}{0}{\ell}\suptiny{0}{0}{2}}{S\subtiny{0}{0}{2\ell}}
=
\left[
\sum\subtiny{0}{0}{\alpha}
(q\subtiny{0}{0}{\alpha}\suptiny{0}{0}{(\ell)})\suptiny{0}{0}{2}
\right]\suptiny{0}{0}{-1},
\end{equation}
with $d\subtiny{0}{0}{1}=\deffdyn$ and $d\subtiny{0}{0}{2}=\deffspec$. For arbitrary mixed initial states, the amplitudes are $b\subtiny{0}{0}{nm}=|\rho\subtiny{0}{0}{0,nm}P\subtiny{0}{0}{0,mn}|$,
where $P\subtiny{0}{0}{0}$ is the support projector of $\rho\subtiny{0}{0}{0}$. We keep $P\subtiny{0}{0}{0}$ fixed when perturbing the initial state, so that
$\widetilde b\subtiny{0}{0}{nm} =|\widetilde\rho\subtiny{0}{0}{0,nm}P\subtiny{0}{0}{0,mn}|$.
This convention permits changes in the support of the prepared state while retaining a common reference projector for the comparison. Consider arbitrary perturbed amplitudes $\widetilde b\subtiny{0}{0}{\alpha}\geq0$, with $S\subtiny{0}{0}{\ell},\widetilde S\subtiny{0}{0}{\ell}>0$, and define the total absolute deviation $\varepsilon\subtiny{0}{0}{\ell}$
\begin{equation}
\varepsilon\subtiny{0}{0}{\ell}
=
\sum\subtiny{0}{0}{\alpha}
\left|
\widetilde b\subtiny{0}{0}{\alpha}\suptiny{0}{0}{\ell}-b\subtiny{0}{0}{\alpha}\suptiny{0}{0}{\ell}
\right|.
\end{equation}
Since
$\widetilde S\subtiny{0}{0}{\ell}
=\sum\subtiny{0}{0}{\alpha}\widetilde b\subtiny{0}{0}{\alpha}\suptiny{0}{0}{\ell}$,
the triangle inequality immediately gives
\begin{equation}
|\widetilde S\subtiny{0}{0}{\ell}-S\subtiny{0}{0}{\ell}|
=
\left|
\sum\subtiny{0}{0}{\alpha}
(\widetilde b\subtiny{0}{0}{\alpha}\suptiny{0}{0}{\ell}-b\subtiny{0}{0}{\alpha}\suptiny{0}{0}{\ell})
\right|
\leq
\varepsilon\subtiny{0}{0}{\ell}.
\end{equation}

We first determine how this moment error affects
the normalized probability distributions.
Adding and subtracting
$\widetilde b\subtiny{0}{0}{\alpha}\suptiny{0}{0}{\ell}/S\subtiny{0}{0}{\ell}$, we obtain
$\widetilde q\subtiny{0}{0}{\alpha}\suptiny{0}{0}{(\ell)}
-q\subtiny{0}{0}{\alpha}\suptiny{0}{0}{(\ell)}=
\frac{\widetilde b\subtiny{0}{0}{\alpha}\suptiny{0}{0}{\ell}-b\subtiny{0}{0}{\alpha}\suptiny{0}{0}{\ell}}
{S\subtiny{0}{0}{\ell}}
+
\widetilde b\subtiny{0}{0}{\alpha}\suptiny{0}{0}{\ell}
\left(
\frac1{\widetilde S\subtiny{0}{0}{\ell}}
-\frac1{S\subtiny{0}{0}{\ell}}
\right).$ Taking absolute values, summing over all channels, and using the triangle inequality yields
\begin{equation}
\sum\subtiny{0}{0}{\alpha}
\left|
\widetilde q\subtiny{0}{0}{\alpha}\suptiny{0}{0}{(\ell)}
-q\subtiny{0}{0}{\alpha}\suptiny{0}{0}{(\ell)}
\right|\leq
\frac{\varepsilon\subtiny{0}{0}{\ell}}{S\subtiny{0}{0}{\ell}}
+
\sum\subtiny{0}{0}{\alpha}\widetilde b\subtiny{0}{0}{\alpha}\suptiny{0}{0}{\ell}
\left|
\frac1{\widetilde S\subtiny{0}{0}{\ell}}
-\frac1{S\subtiny{0}{0}{\ell}}
\right|=
\frac{\varepsilon\subtiny{0}{0}{\ell}}{S\subtiny{0}{0}{\ell}}
+
\widetilde S\subtiny{0}{0}{\ell}
\frac{|\widetilde S\subtiny{0}{0}{\ell}-S\subtiny{0}{0}{\ell}|}
{S\subtiny{0}{0}{\ell}\widetilde S\subtiny{0}{0}{\ell}}=
\frac{\varepsilon\subtiny{0}{0}{\ell}}{S\subtiny{0}{0}{\ell}}
+
\frac{|\widetilde S\subtiny{0}{0}{\ell}-S\subtiny{0}{0}{\ell}|}{S\subtiny{0}{0}{\ell}}\leq
\frac{2\varepsilon\subtiny{0}{0}{\ell}}{S\subtiny{0}{0}{\ell}}.
\label{eq:robustness_normalized}
\end{equation}
Thus, the total variation distance between
the normalized distributions satisfies $D\subtiny{0}{0}{TV}
\left(
\widetilde{\boldsymbol q}\suptiny{0}{0}{(\ell)},
\boldsymbol q\suptiny{0}{0}{(\ell)}
\right)
\leq
\frac{\varepsilon\subtiny{0}{0}{\ell}}{S\subtiny{0}{0}{\ell}}.$ To translate this result into a bound on the
effective dimensions, we use an inequality
valid for arbitrary normalized probability
distributions $\boldsymbol q$ and
$\widetilde{\boldsymbol q}$.
Since $\|\widetilde{\boldsymbol q}\|\subtiny{0}{0}{2}\suptiny{0}{0}{2}
-\|\boldsymbol q\|\subtiny{0}{0}{2}\suptiny{0}{0}{2}
=
\sum\subtiny{0}{0}{\alpha}
(\widetilde q\subtiny{0}{0}{\alpha}-q\subtiny{0}{0}{\alpha})
(\widetilde q\subtiny{0}{0}{\alpha}+q\subtiny{0}{0}{\alpha}),$ factorization of the difference of squares gives $\|\widetilde{\boldsymbol q}\|\subtiny{0}{0}{2}
-\|\boldsymbol q\|\subtiny{0}{0}{2}
=
\sum\subtiny{0}{0}{\alpha}
c\subtiny{0}{0}{\alpha}(\widetilde q\subtiny{0}{0}{\alpha}-q\subtiny{0}{0}{\alpha}),$ 
where $c\subtiny{0}{0}{\alpha}
=
\frac{\widetilde q\subtiny{0}{0}{\alpha}+q\subtiny{0}{0}{\alpha}}{
\|\widetilde{\boldsymbol q}\|\subtiny{0}{0}{2}+
\|\boldsymbol q\|\subtiny{0}{0}{2}}.$ Every probability component is nonnegative and
satisfies
$q\subtiny{0}{0}{\alpha}\leq\|\boldsymbol q\|\subtiny{0}{0}{2}$ and
$\widetilde q\subtiny{0}{0}{\alpha}
\leq\|\widetilde{\boldsymbol q}\|\subtiny{0}{0}{2}$.
Consequently, $0\leq c\subtiny{0}{0}{\alpha}\leq1$. Moreover, both distributions are normalized, so $\sum\subtiny{0}{0}{\alpha}
(\widetilde q\subtiny{0}{0}{\alpha}-q\subtiny{0}{0}{\alpha})=0.$ Separating the positive and negative differences,
we therefore have
\begin{align}
\sum\subtiny{0}{0}{\widetilde q\subtiny{0}{0}{\alpha}>q\subtiny{0}{0}{\alpha}}
(\widetilde q\subtiny{0}{0}{\alpha}-q\subtiny{0}{0}{\alpha})
&=
\sum\subtiny{0}{0}{\widetilde q\subtiny{0}{0}{\alpha}<q\subtiny{0}{0}{\alpha}}
(q\subtiny{0}{0}{\alpha}-\widetilde q\subtiny{0}{0}{\alpha})=
\frac12
\sum\subtiny{0}{0}{\alpha}
|\widetilde q\subtiny{0}{0}{\alpha}-q\subtiny{0}{0}{\alpha}|.
\end{align}
Because each coefficient $c\subtiny{0}{0}{\alpha}$ belongs to
$[0,1]$, the weighted sum of these differences
cannot exceed the total positive contribution
or fall below the total negative contribution.
It follows that
\begin{equation}
\label{eq:robustness_probability}
\left|
\|\widetilde{\boldsymbol q}\|\subtiny{0}{0}{2}
-\|\boldsymbol q\|\subtiny{0}{0}{2}
\right|
\leq
\frac12
\|\widetilde{\boldsymbol q}
-\boldsymbol q\|\subtiny{0}{0}{1}.
\end{equation}

Applying this inequality to
$q\subtiny{0}{0}{\alpha}\suptiny{0}{0}{(\ell)}$ and
$\widetilde q\subtiny{0}{0}{\alpha}\suptiny{0}{0}{(\ell)}$,
and using
\begin{equation}
\|\boldsymbol q\suptiny{0}{0}{(\ell)}\|\subtiny{0}{0}{2}
=
\frac1{\sqrt{d\subtiny{0}{0}{\ell}}},
\qquad
\|\widetilde{\boldsymbol q}\suptiny{0}{0}{(\ell)}\|\subtiny{0}{0}{2}
=
\frac1{\sqrt{\widetilde d\subtiny{0}{0}{\ell}}},
\end{equation}
we obtain, from
Eq.~\eqref{eq:robustness_normalized},
\begin{equation}
\label{eq:robustness_sqrt_bound}
\left|
\frac1{\sqrt{\widetilde d\subtiny{0}{0}{\ell}}}
-
\frac1{\sqrt{d\subtiny{0}{0}{\ell}}}
\right|
\leq
\frac{\varepsilon\subtiny{0}{0}{\ell}}{S\subtiny{0}{0}{\ell}}.
\end{equation}
Equivalently,
\begin{equation}
\frac1{\sqrt{d\subtiny{0}{0}{\ell}}}
-\frac{\varepsilon\subtiny{0}{0}{\ell}}{S\subtiny{0}{0}{\ell}}
\leq
\frac1{\sqrt{\widetilde d\subtiny{0}{0}{\ell}}}
\leq
\frac1{\sqrt{d\subtiny{0}{0}{\ell}}}
+\frac{\varepsilon\subtiny{0}{0}{\ell}}{S\subtiny{0}{0}{\ell}}.
\end{equation}
Introducing the dimensionless error
\begin{equation}
x\subtiny{0}{0}{\ell}
:=
\sqrt{d\subtiny{0}{0}{\ell}}
\frac{\varepsilon\subtiny{0}{0}{\ell}}{S\subtiny{0}{0}{\ell}},
\end{equation}
this inequality becomes
\begin{equation}
1-x\subtiny{0}{0}{\ell}
\leq
\sqrt{\frac{d\subtiny{0}{0}{\ell}}{\widetilde d\subtiny{0}{0}{\ell}}}
\leq
1+x\subtiny{0}{0}{\ell}.
\end{equation}
For $x\subtiny{0}{0}{\ell}<1$, all terms are positive.
Taking reciprocals reverses the inequalities,
and squaring gives
\begin{equation}
\frac1{(1+x\subtiny{0}{0}{\ell})\suptiny{0}{0}{2}}
\leq
\frac{\widetilde d\subtiny{0}{0}{\ell}}{d\subtiny{0}{0}{\ell}}
\leq
\frac1{(1-x\subtiny{0}{0}{\ell})\suptiny{0}{0}{2}},
\end{equation}
which proves
Eq.~\eqref{eq:effective_dimension_robustness}. For $x\subtiny{0}{0}{\ell}\ll1$, the expansions $\frac1{(1\pm x\subtiny{0}{0}{\ell})\suptiny{0}{0}{2}}
=
1\mp2x\subtiny{0}{0}{\ell}+O(x\subtiny{0}{0}{\ell}\suptiny{0}{0}{2})$ show that
\begin{equation}
\frac{|\widetilde d\subtiny{0}{0}{\ell}-d\subtiny{0}{0}{\ell}|}{d\subtiny{0}{0}{\ell}}
\leq
2\sqrt{d\subtiny{0}{0}{\ell}}
\frac{\varepsilon\subtiny{0}{0}{\ell}}{S\subtiny{0}{0}{\ell}}
+
O\left(
\frac{d\subtiny{0}{0}{\ell}\varepsilon\subtiny{0}{0}{\ell}\suptiny{0}{0}{2}}{S\subtiny{0}{0}{\ell}\suptiny{0}{0}{2}}
\right).
\end{equation}
This establishes the first-order stability estimate. We next relate the errors associated with
$d\subtiny{0}{0}{1}$ and $d\subtiny{0}{0}{2}$. Writing
$\delta b\subtiny{0}{0}{\alpha}=\widetilde b\subtiny{0}{0}{\alpha}-b\subtiny{0}{0}{\alpha}$,
we have $\widetilde b\subtiny{0}{0}{\alpha}\suptiny{0}{0}{2}-b\subtiny{0}{0}{\alpha}\suptiny{0}{0}{2}
= 2b\subtiny{0}{0}{\alpha}\delta b\subtiny{0}{0}{\alpha} +(\delta b\subtiny{0}{0}{\alpha})\suptiny{0}{0}{2}.$ Therefore,
\begin{equation}
\varepsilon\subtiny{0}{0}{2}
=
\sum\subtiny{0}{0}{\alpha}
\left|
2b\subtiny{0}{0}{\alpha}\delta b\subtiny{0}{0}{\alpha}
+(\delta b\subtiny{0}{0}{\alpha})\suptiny{0}{0}{2}
\right|\leq
2\sum\subtiny{0}{0}{\alpha}
b\subtiny{0}{0}{\alpha}|\delta b\subtiny{0}{0}{\alpha}|
+
\sum\subtiny{0}{0}{\alpha}|\delta b\subtiny{0}{0}{\alpha}|\suptiny{0}{0}{2}.
\end{equation}
By Cauchy--Schwarz, $\sum\subtiny{0}{0}{\alpha}
b\subtiny{0}{0}{\alpha}|\delta b\subtiny{0}{0}{\alpha}|
\leq
\left(\sum\subtiny{0}{0}{\alpha} b\subtiny{0}{0}{\alpha}\suptiny{0}{0}{2}\right)\suptiny{0}{0}{1/2}
\left(\sum\subtiny{0}{0}{\alpha}|\delta b\subtiny{0}{0}{\alpha}|\suptiny{0}{0}{2}\right)\suptiny{0}{0}{1/2}.$ Since $\left(
\sum\subtiny{0}{0}{\alpha}|\delta b\subtiny{0}{0}{\alpha}|\suptiny{0}{0}{2}
\right)\suptiny{0}{0}{1/2}
\leq
\sum\subtiny{0}{0}{\alpha}|\delta b\subtiny{0}{0}{\alpha}|
=
\varepsilon\subtiny{0}{0}{1},$ we obtain
\begin{equation}
\label{eq:robustness_moment_errors}
\varepsilon\subtiny{0}{0}{2}
\leq
2\sqrt{S\subtiny{0}{0}{2}}\,\varepsilon\subtiny{0}{0}{1}
+\varepsilon\subtiny{0}{0}{1}\suptiny{0}{0}{2}.
\end{equation}
Thus, a bound on the absolute amplitude error
also controls the error in the spectral-power weights.

A dimension-independent estimate is obtained
under a different assumption.
Suppose that the nonzero amplitudes satisfy
$\widetilde b\subtiny{0}{0}{\alpha}
=
b\subtiny{0}{0}{\alpha}(1+\xi\subtiny{0}{0}{\alpha}),$ for $
|\xi\subtiny{0}{0}{\alpha}|\leq\epsilon<1.$ For every $k>0$,
\begin{equation}
(1-\epsilon)\suptiny{0}{0}{k} b\subtiny{0}{0}{\alpha}\suptiny{0}{0}{k}
\leq
\widetilde b\subtiny{0}{0}{\alpha}\suptiny{0}{0}{k}
\leq
(1+\epsilon)\suptiny{0}{0}{k} b\subtiny{0}{0}{\alpha}\suptiny{0}{0}{k}.
\end{equation}
Summing over the transition channels gives
\begin{equation}
(1-\epsilon)\suptiny{0}{0}{k} S\subtiny{0}{0}{k}
\leq
\widetilde S\subtiny{0}{0}{k}
\leq
(1+\epsilon)\suptiny{0}{0}{k} S\subtiny{0}{0}{k}.
\end{equation}
Using the definition of $d\subtiny{0}{0}{\ell}$, $\frac{\widetilde d\subtiny{0}{0}{\ell}}{d\subtiny{0}{0}{\ell}}
=
\left(
\frac{\widetilde S\subtiny{0}{0}{\ell}}{S\subtiny{0}{0}{\ell}}
\right)\suptiny{0}{0}{2}
\frac{S\subtiny{0}{0}{2\ell}}{\widetilde S\subtiny{0}{0}{2\ell}},$
we obtain
\begin{equation}
\left(
\frac{1-\epsilon}{1+\epsilon}
\right)\suptiny{0}{0}{2\ell}
\leq
\frac{\widetilde d\subtiny{0}{0}{\ell}}{d\subtiny{0}{0}{\ell}}
\leq
\left(
\frac{1+\epsilon}{1-\epsilon}
\right)\suptiny{0}{0}{2\ell},
\end{equation}
proving Eq.~\eqref{eq:relative_amplitude_robustness}.
For $\epsilon\ll1$, the first-order relative
variation is bounded by $4\ell\epsilon$.
Unlike the absolute-error bound, this estimate requires the set of active amplitude channels to remain unchanged.

We now express the amplitude errors in terms
of the initial-state perturbation.
Let $\delta\rho = \widetilde\rho\subtiny{0}{0}{0}-\rho\subtiny{0}{0}{0}.$ With the reference projector $P\subtiny{0}{0}{0}$ fixed, the reverse triangle inequality yields
\begin{align}
|\widetilde b\subtiny{0}{0}{nm}-b\subtiny{0}{0}{nm}|
&=
\Big|
|\widetilde\rho\subtiny{0}{0}{0,nm}P\subtiny{0}{0}{0,mn}|
-
|\rho\subtiny{0}{0}{0,nm}P\subtiny{0}{0}{0,mn}|
\Big|\leq
|\delta\rho\subtiny{0}{0}{nm}|\,|P\subtiny{0}{0}{0,mn}|.
\end{align}
Consequently,
\begin{align}
\varepsilon\subtiny{0}{0}{1}
&\leq
\sum\subtiny{0}{0}{n<m}
|\delta\rho\subtiny{0}{0}{nm}|\,|P\subtiny{0}{0}{0,mn}|\leq
\left(
\sum\subtiny{0}{0}{n<m}|\delta\rho\subtiny{0}{0}{nm}|\suptiny{0}{0}{2}
\right)\suptiny{0}{0}{1/2}\!
\left(
\sum\subtiny{0}{0}{n<m}|P\subtiny{0}{0}{0,mn}|\suptiny{0}{0}{2}
\right)\suptiny{0}{0}{1/2}.
\end{align}
For any Hermitian operator $X$, its off-diagonal matrix elements satisfy $\sum\subtiny{0}{0}{n<m}|X\subtiny{0}{0}{nm}|\suptiny{0}{0}{2}
=
\frac12
\left(
\|X\|\subtiny{0}{0}{2}\suptiny{0}{0}{2}-\sum\subtiny{0}{0}{n}|X\subtiny{0}{0}{nn}|\suptiny{0}{0}{2}
\right)
\leq
\frac12\|X\|\subtiny{0}{0}{2}\suptiny{0}{0}{2}.$ Applying this identity to $\delta\rho$ and
$P\subtiny{0}{0}{0}$, we find
\begin{equation}
\varepsilon\subtiny{0}{0}{1}
\leq
\frac12
\|\delta\rho\|\subtiny{0}{0}{2}\|P\subtiny{0}{0}{0}\|\subtiny{0}{0}{2}.
\end{equation}
Since $P\subtiny{0}{0}{0}\suptiny{0}{0}{2}=P\subtiny{0}{0}{0}$,
\begin{equation}
\|P\subtiny{0}{0}{0}\|\subtiny{0}{0}{2}\suptiny{0}{0}{2}
=
\Tr(P\subtiny{0}{0}{0}\suptiny{0}{0}{2})
=
\Tr P\subtiny{0}{0}{0}
=
r,
\end{equation}
and hence
\begin{equation}
\varepsilon\subtiny{0}{0}{1}
\leq
\frac{\sqrt r}{2}\|\delta\rho\|\subtiny{0}{0}{2}.
\end{equation}
This proves
Eq.~\eqref{eq:initial_state_amplitude_error_dyn}. Substituting this result into
Eq.~\eqref{eq:robustness_moment_errors}
gives
\begin{align}
\varepsilon\subtiny{0}{0}{2}
&\leq
2\sqrt{S\subtiny{0}{0}{2}}
\left(
\frac{\sqrt r}{2}\|\delta\rho\|\subtiny{0}{0}{2}
\right)
+
\left(
\frac{\sqrt r}{2}\|\delta\rho\|\subtiny{0}{0}{2}
\right)\suptiny{0}{0}{2}
\nonumber\\
&=
\sqrt{rS\subtiny{0}{0}{2}}\|\delta\rho\|\subtiny{0}{0}{2}
+
\frac r4\|\delta\rho\|\subtiny{0}{0}{2}\suptiny{0}{0}{2},
\end{align}
establishing
Eq.~\eqref{eq:initial_state_amplitude_error_spec}.
These estimates hold without assuming that
$\widetilde\rho\subtiny{0}{0}{0}$ and $\rho\subtiny{0}{0}{0}$ have identical
supports, because the same reference projector
$P\subtiny{0}{0}{0}$ is used in both sets of amplitudes.

The energy effective dimension obeys an
analogous continuity bound.
For a fixed Hamiltonian with distinct
energy levels, let
\begin{equation}
p\subtiny{0}{0}{n}
=
\Tr\left(
\ketbra{E\subtiny{0}{0}{n}}{E\subtiny{0}{0}{n}}\rho\subtiny{0}{0}{0}
\right),
\qquad
\deff
=
\frac1{\sum\subtiny{0}{0}{n} p\subtiny{0}{0}{n}\suptiny{0}{0}{2}}.
\end{equation}
Under an initial-state perturbation,
the energy eigenstates remain fixed and
\begin{equation}
\widetilde p\subtiny{0}{0}{n}
=
\Tr\left(
\ketbra{E\subtiny{0}{0}{n}}{E\subtiny{0}{0}{n}}\widetilde\rho\subtiny{0}{0}{0}
\right).
\end{equation}
Their difference is therefore
\begin{equation}
\widetilde p\subtiny{0}{0}{n}-p\subtiny{0}{0}{n}
=
\Tr\left(
\ketbra{E\subtiny{0}{0}{n}}{E\subtiny{0}{0}{n}}\delta\rho
\right).
\end{equation}
The populations $\{p\subtiny{0}{0}{n}\}$ and
$\{\widetilde p\subtiny{0}{0}{n}\}$ are normalized probability
distributions. Applying
Eq.~\eqref{eq:robustness_probability},
we obtain
\begin{equation}
\left|
\frac1{\sqrt{\widetilde d\subtiny{0}{0}{\mathrm{eff}}}}
-
\frac1{\sqrt{\deff}}
\right|
\leq
\frac12
\sum\subtiny{0}{0}{n}|\widetilde p\subtiny{0}{0}{n}-p\subtiny{0}{0}{n}|.
\end{equation}
To bound the population differences in terms
of the density operators, write the Jordan
decomposition of the Hermitian perturbation as $\delta\rho=\delta\rho\subtiny{0}{0}{+}-\delta\rho\subtiny{0}{0}{-},$ for $\delta\rho\subtiny{0}{0}{\pm}\geq0.$ For each energy eigenstate,
\begin{equation}
\left|
\Tr\left(
\ketbra{E\subtiny{0}{0}{n}}{E\subtiny{0}{0}{n}}\delta\rho
\right)
\right|\leq
\Tr\left(
\ketbra{E\subtiny{0}{0}{n}}{E\subtiny{0}{0}{n}}\delta\rho\subtiny{0}{0}{+}
\right)+
\Tr\left(
\ketbra{E\subtiny{0}{0}{n}}{E\subtiny{0}{0}{n}}\delta\rho\subtiny{0}{0}{-}
\right).
\end{equation}
Summing over the complete energy eigenbasis,
\begin{align}
\sum\subtiny{0}{0}{n}|\widetilde p\subtiny{0}{0}{n}-p\subtiny{0}{0}{n}|\leq
\Tr(\delta\rho\subtiny{0}{0}{+})+\Tr(\delta\rho\subtiny{0}{0}{-})=
\|\delta\rho\|\subtiny{0}{0}{1}.
\end{align}
Thus,
\begin{equation}
\label{eq:energy_deff_robustness}
\left|
\frac1{\sqrt{\widetilde d\subtiny{0}{0}{\mathrm{eff}}}}
-
\frac1{\sqrt{\deff}}
\right|
\leq
\frac12\sum\subtiny{0}{0}{n}|\widetilde p\subtiny{0}{0}{n}-p\subtiny{0}{0}{n}|
\leq
\frac12\|\delta\rho\|\subtiny{0}{0}{1}.
\end{equation}
This is the corresponding continuity bound
for the standard energy effective dimension.
As in the previous derivation, sufficiently
small errors also yield two-sided bounds
for $\widetilde d\subtiny{0}{0}{\mathrm{eff}}/\deff$ by inversion and squaring. Finally, consider the normalized spectral-power distribution
\begin{equation}
r\subtiny{0}{0}{nm}
=
\frac{b\subtiny{0}{0}{nm}\suptiny{0}{0}{2}}{S\subtiny{0}{0}{2}},
\qquad
\Hpow
=
-\sum\subtiny{0}{0}{n<m}r\subtiny{0}{0}{nm}\ln r\subtiny{0}{0}{nm},
\end{equation}
with the convention $0\ln0=0$.
The perturbed distribution is
$\widetilde r\subtiny{0}{0}{nm}
=\widetilde b\subtiny{0}{0}{nm}\suptiny{0}{0}{2}/\widetilde S\subtiny{0}{0}{2}$.
From Eq.~\eqref{eq:robustness_normalized}
with $\ell=2$, its total variation distance
from the original distribution obeys
\begin{equation}
D\subtiny{0}{0}{TV}
=
\frac12\sum\subtiny{0}{0}{n<m}
|\widetilde r\subtiny{0}{0}{nm}-r\subtiny{0}{0}{nm}|
\leq
\frac{\varepsilon\subtiny{0}{0}{2}}{S\subtiny{0}{0}{2}}.
\end{equation}
We denote the upper bound by
\begin{equation}
\eta=\frac{\varepsilon\subtiny{0}{0}{2}}{S\subtiny{0}{0}{2}}.
\end{equation}

Let $K\geq2$ be the number of transition
channels with nonzero weight in at least
one of these distributions.
Both distributions can then be regarded
as probability vectors over the same
$K$-element set, assigning zero weight
to channels absent from either distribution. For $D\subtiny{0}{0}{TV}\leq1-1/K$, the Shannon entropy
continuity inequality gives
\cite{Audenaert2007}
\begin{equation}
|\Delta\Hpow|
\leq
D\subtiny{0}{0}{TV}\ln(K-1)+h\subtiny{0}{0}{2}(D\subtiny{0}{0}{TV}),
\end{equation}
where $h\subtiny{0}{0}{2}(x) = -x\ln x-(1-x)\ln(1-x).$ To express this result in terms of the
amplitude-error bound $\eta$, define
\begin{equation}
f\subtiny{0}{0}{K}(x)
=
h\subtiny{0}{0}{2}(x)+x\ln(K-1).
\end{equation}
Its derivative is $f\subtiny{0}{0}{K}'(x) = \ln\left(
\frac{(K-1)(1-x)}{x}\right).$ Thus, $f\subtiny{0}{0}{K}'(x)\geq0$ $\Longleftrightarrow$ $x\leq1-\frac1K.$ The function is therefore increasing on
$0\leq x\leq1-1/K$ and reaches $f\subtiny{0}{0}{K}\left(1-\frac1K\right) = \ln K$ at the upper endpoint of this interval.

If $\eta\leq1-1/K$, then
$D\subtiny{0}{0}{TV}\leq\eta$ lies entirely within
the interval on which $f\subtiny{0}{0}{K}$ is increasing.
Consequently,
\begin{equation}
|\Delta\Hpow|
\leq
f\subtiny{0}{0}{K}(D\subtiny{0}{0}{TV})
\leq
f\subtiny{0}{0}{K}(\eta).
\end{equation}

If $\eta>1-1/K$, the bound
$D\subtiny{0}{0}{TV}\leq\eta$ no longer guarantees
that $D\subtiny{0}{0}{TV}$ belongs to the increasing
interval of $f\subtiny{0}{0}{K}$.
Nevertheless, both entropies lie between
zero and $\ln K$, so the universal bound
\begin{equation}
|\Delta\Hpow|\leq\ln K,
\end{equation}
always holds.
Combining the two regimes gives
\begin{equation}
|\Delta\Hpow|
\leq
\begin{cases}
f\subtiny{0}{0}{K}(\eta),
&0\leq\eta\leq1-1/K,\\
\ln K,
&\eta>1-1/K.
\end{cases}
\end{equation}
Introducing
\begin{equation}
\eta\subtiny{0}{0}{*}
=
\min\left\{
\eta,1-\frac1K
\right\},
\end{equation}
both cases can be expressed as
\begin{equation}
|\Delta\Hpow|
\leq
h\subtiny{0}{0}{2}(\eta\subtiny{0}{0}{*})+\eta\subtiny{0}{0}{*}\ln(K-1),
\end{equation}
which establishes
Eq.~\eqref{eq:power_entropy_robustness}. The quantity $\eta$ bounds the total
variation distance between the normalized
spectral-power distributions, whereas
$1-1/K$ is the point at which the
continuity function reaches its maximum
value $\ln K$.
The minimum in the definition of $\eta\subtiny{0}{0}{*}$
therefore ensures that the error-dependent
estimate is used only within the
monotonicity interval, while the universal
entropy bound applies otherwise.

A simpler expression follows in the
small-error regime.
Using $-(1-\eta)\ln(1-\eta)\leq\eta$ for $0\leq\eta\leq1$, we find
\begin{align}
f\subtiny{0}{0}{K}(\eta)\leq
-\eta\ln\eta+\eta+\eta\ln(K-1)=
\eta\ln\left(
\frac{e(K-1)}{\eta}
\right),\quad
0<\eta\leq1-\frac1K.
\end{align}
For fixed $K$, the right-hand side
vanishes as $\eta\to0$, establishing
continuity of $\Hpow$ under amplitude
perturbations. The dependence on $\ln(K-1)$
also shows that this estimate does not
provide a dimension-independent modulus
of continuity when the number of available
transition channels grows without bound.
%


\section{Numerical methods and spectral descriptors}
\label{appsec:appE}

In this Appendix, we present the numerical procedures used to generate the data presented in the main text and to compute the spectral
descriptors associated with the leakage fidelity function. The purpose is twofold. First, we specify the finite-dimensional spin model, the exact-diagonalization protocol, the initial states, and the construction of the time-dependent LFF signal. Second, we explain how the quantities $\deff$, $\deffdyn$, $\deffspec$, and $\Hpow$ are computed from either the Hamiltonian spectrum or controlled families of energy-population distributions.

The central point of the numerical analysis is that the effective dimension $\deff$ does not exhaust the information relevant for equilibration. While $\deff$ fixes the LFF equilibration value,
$\deffdyn$ and $\deffspec$ resolve two different aspects of the fluctuation signal: the effective number of dynamical cosine modes and the distribution of spectral power over Bohr frequencies, respectively. The Shannon power entropy $\Hpow$ provides a complementary measure of the overall spread of this spectral-power distribution.
Consequently, states with similar values of $\deff$ may exhibit different fluctuation amplitudes, while states with the same $\deff$ and $\deffdyn$ may still possess different spectral behavior.

\subsection{Spin-chain Hamiltonian}
\label{secapp:ising_model}

The numerical simulations in the main text are performed for a finite spin-$1/2$ Ising-like chain with both transverse and longitudinal fields. The Hamiltonian is
\begin{equation}
\label{eq:H_non_integrable_app}
H =
g \sum_{i=1}^{N} \sigma\subtiny{0}{0}{i}^x
+
h \sum_{i=2}^{N-1} \sigma\subtiny{0}{0}{i}^z
+
J \sum_{i=1}^{N-1} \sigma\subtiny{0}{0}{i}^z \sigma_{i+1}^z
+
(h-J)\left(\sigma\subtiny{0}{0}{1}^z+\sigma\subtiny{0}{0}{N}^z\right),
\end{equation}
where $\sigma\subtiny{0}{0}{i}^\alpha$ denotes the Pauli operator acting on site $i$ in direction $\alpha=x,y,z$. The first term is a transverse field, the
second term is a longitudinal field acting on the bulk, the third term is the nearest-neighbor Ising interaction, and the final term introduces a boundary-field correction. Throughout the numerical examples, we use {$g=\frac{5+\sqrt{5}}{8},$ $h=\frac{1+\sqrt{5}}{4},$ and $J=1$.} These parameters are chosen to avoid the fine-tuned symmetries and degeneracies that occur in simpler transverse-field Ising limits. The resulting model is non-integrable for the finite chains considered here and provides a convenient setting in which to study dephasing,
spectral spreading, and LFF equilibration. The Hamiltonian is diagonalized exactly, {$H=\sum\subtiny{0}{0}{n=0}\suptiny{0}{0}{d\subtiny{0}{0}{\mathrm{E}}-1}E\subtiny{0}{0}{n}\ketbra{E\subtiny{0}{0}{n}}{E\subtiny{0}{0}{n}},$} where $d\subtiny{0}{0}{\mathrm{E}}=2^N$ is the dimension of the full Hilbert space. In the simulations shown throughout the main text, we use $N=10$, hence $d\subtiny{0}{0}{\mathrm{E}}=1024$. Exact diagonalization gives access to the complete set of eigenvalues $E\subtiny{0}{0}{n}$ and eigenvectors $\ket{E\subtiny{0}{0}{n}}$, allowing the LFF signal to be computed without Trotter or time-discretization errors in the unitary
evolution. All the codes required to perform the numerical investigations presented in this paper, along with implementations for many related studies, are available in the Julia scientific programming language at \cite{git2026}.

\subsection{Population-distribution families}
\label{secapp:synthetic_families}

To separate the roles of $\deff$, $\deffdyn$, and $\deffspec$ in a controlled way, we also consider engineered energy-population distributions. These distributions are not intended to reproduce the spectral occupation of a particular many-body initial state. Instead, they provide transparent test cases demonstrating that the descriptors introduced in the main text probe distinct layers of equilibration. For later convenience, we define the population moments 
\begin{equation}
\label{eq:population_moments_app}
S_q[p]\equiv\sum\subtiny{0}{0}{n}p\subtiny{0}{0}{n}^q.
\end{equation}
For pure initial states, $P=\rho\subtiny{0}{0}{0}$ and
$b\subtiny{0}{0}{nm}=p\subtiny{0}{0}{n}p\subtiny{0}{0}{m}$.
Hence, the spectral-amplitude moments $S_q$
introduced in the main text satisfy{
\begin{equation}
\label{eq:population_spectral_moments_app}
S_q
=
\sum\subtiny{0}{0}{n<m}
b\subtiny{0}{0}{nm}^q
=
\frac{1}{2}
\left(S_q[p]^2-S_{2q}[p]\right).
\end{equation}}
In particular, the effective dimension and the equilibrium value of the LFF are
\begin{equation}
\label{eq:deff_moment_app}
\deff=\frac{1}{S_2[p]},
\qquad
\leakz\suptiny{0}{0}{\infty}
=1-S_2[p]
=1-\frac{1}{\deff}.
\end{equation}
The main family used in the static numerical exploration is the one-peak plus flat-tail distribution
\begin{equation}
\label{eq:peak_tail_family_app}
p^{(a,M)}
=
\left(
a,
\underbrace{
\frac{1-a}{M},\ldots,\frac{1-a}{M}
}_{M\ \mathrm{terms}}
\right),
\qquad
0<a<1,
\qquad
M\in\mathbb{N}.
\end{equation}
One energy level carries population $a$, while the remaining probability weight is uniformly distributed over $M$ additional levels. This family interpolates between a localized distribution dominated by a single component and a broadly delocalized dilute tail. For this family, the effective dimension is
\begin{equation}
\label{eq:deff_peak_tail_app}
\deff(a,M)
=
\frac{1}{
a\suptiny{0}{0}{2}+\frac{(1-a)\suptiny{0}{0}{2}}{M}
},
\end{equation}
and, therefore,
\begin{equation}
\label{eq:deff_peak_tail_limit_app}
\deff(a,M)
\xrightarrow[M\to\infty]{}
\frac{1}{a\suptiny{0}{0}{2}}.
\end{equation}
Thus, for fixed $a$, the effective dimension saturates as the tail size $M$ increases. This shows that $\deff$ becomes insensitive to the further spreading of very small populations once the dominant
component fixes the leading contribution to $S\subtiny{0}{0}{2}$.

By contrast, the number of pairwise dynamical channels continues to grow with $M$. The pair set contains $M$ peak--tail pairs and $M(M-1)/2$ tail--tail pairs. Consequently, $\deffdyn$, $\deffspec$, and
$\Hpow$ may continue to grow even when $\deff$ is already close to its saturation value. The family in
Eq.~\eqref{eq:peak_tail_family_app}, therefore, provides a simple counterexample to the idea that the effective dimension alone fully characterizes equilibration-relevant spectral structure. In the numerical scan, we fix
\begin{equation}
a\in\{0.08,0.10,0.12,0.14,0.18,0.22\},
\end{equation}
and vary the tail size $M$. For each pair $(a,M)$, we compute
$\deff$, $\deffdyn$, $\deffspec$, and $\Hpow$.

\begin{figure*}[t]
\centering
\begin{subfigure}{0.49\textwidth}
    \centering
    \includegraphics[width=\linewidth]
        {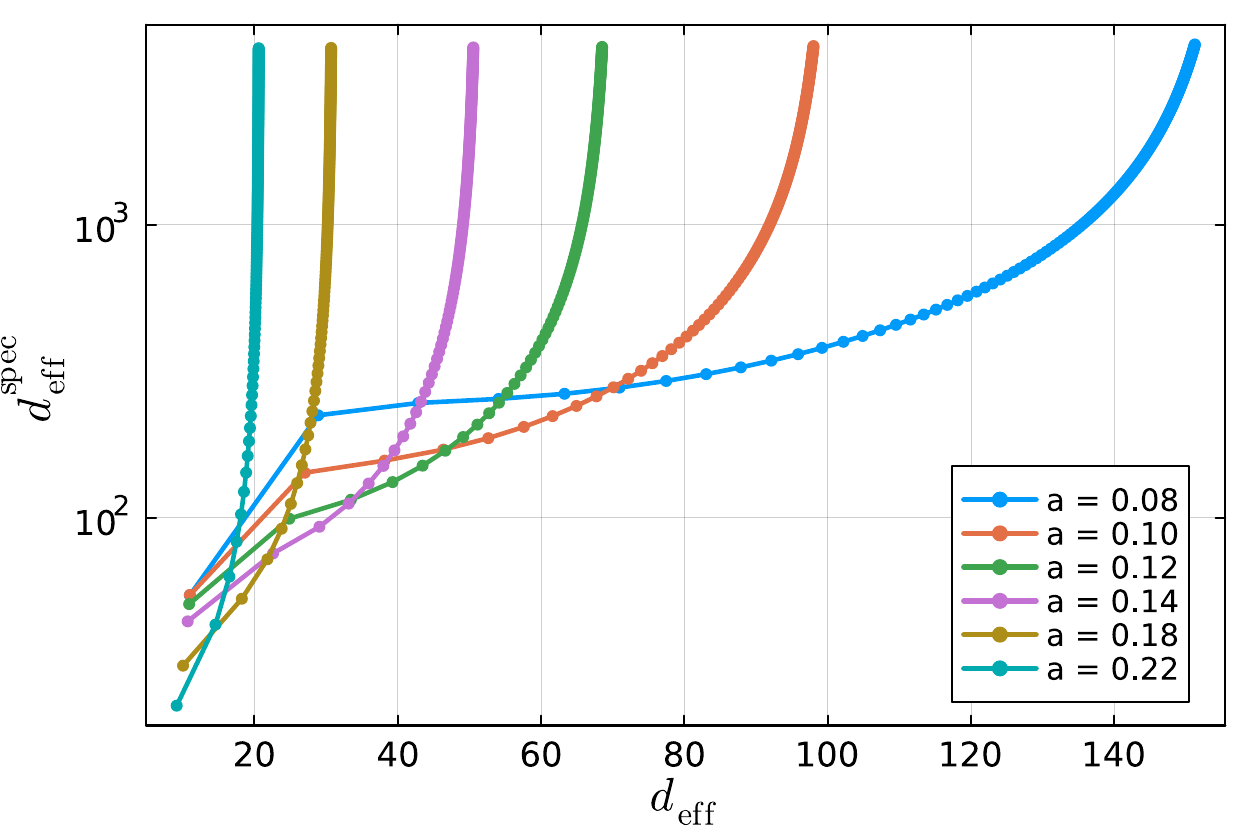}
    \caption{$\deffspec$ as a function of $\deff$.}
    \label{fig:fig8a}
\end{subfigure}
\hfill
\begin{subfigure}{0.49\textwidth}
    \centering
    \includegraphics[width=\linewidth]
        {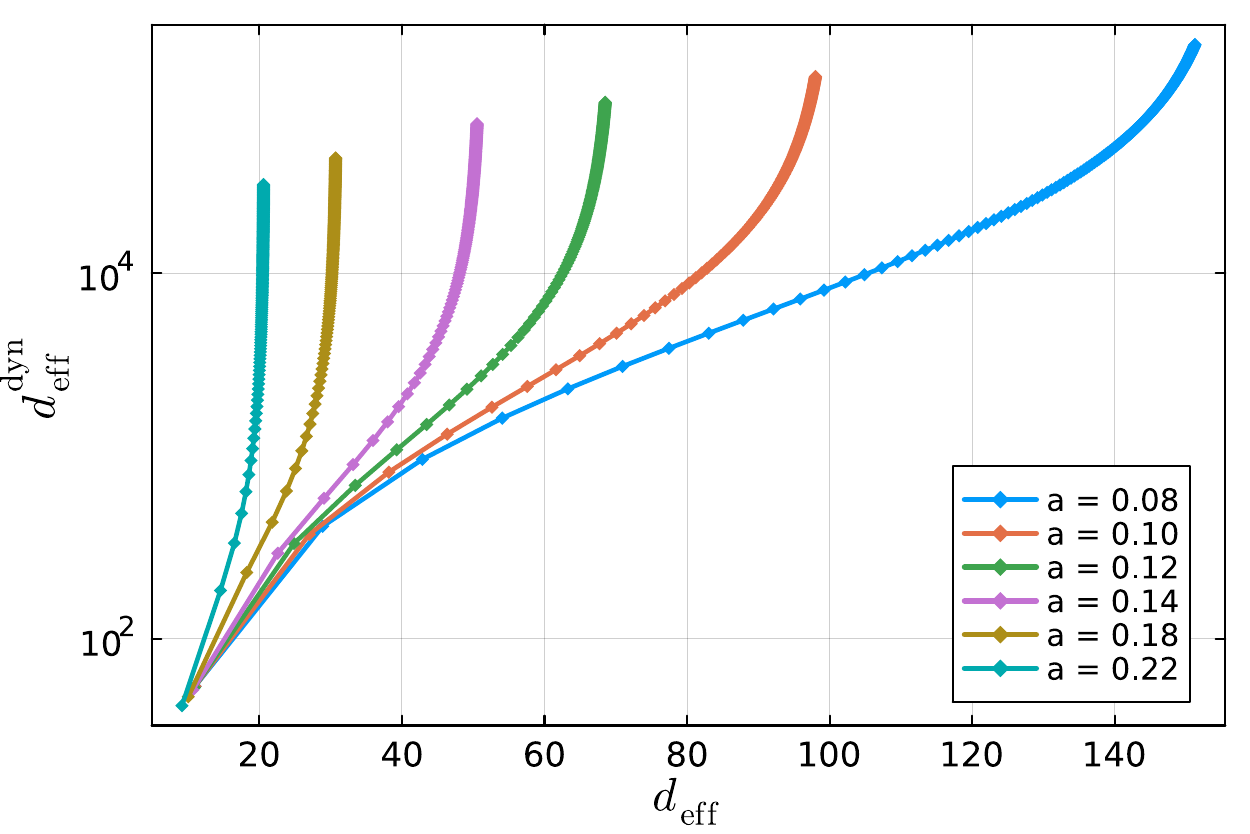}
    \caption{$\deffdyn$ as a function of $\deff$.}
    \label{fig:fig8b}
\end{subfigure}
\caption{\justifying
Scan of the one-peak plus flat-tail family in
Eq.~\eqref{eq:peak_tail_family_app}. Each curve corresponds to a fixed value of
$a\in\{0.08,0.10,0.12,0.14,0.18,0.22\}$,
while the tail size $M$ is varied. The effective dimension
$\deff$ saturates for fixed $a$, whereas $\deffdyn$ and $\deffspec$ continue to grow as additional weakly populated levels create new pairwise interference channels. Thus, states with nearly identical LFF equilibration values may nevertheless possess distinct fluctuation stability and distinct spectral-power organization.}
\label{fig:fig8}
\end{figure*}

Fig.~\ref{fig:fig8} makes the separation among the
descriptors explicit. The left panel shows that $\deffspec$ may grow even when $\deff$ is already close to saturation, because the spectral-power distribution keeps acquiring additional pairwise frequency channels. The right panel shows the analogous behavior for $\deffdyn$, which counts the effective number of amplitude-weighted cosine modes contributing to the LFF fluctuation signal. Therefore,
the scan provides a static population-level validation of the hierarchy proposed in the main text: $\deff$ controls the LFF equilibration value, $\deffdyn$ controls the size of long-time fluctuations, and $\deffspec$ and $\Hpow$ characterize the frequency-domain organization of the fluctuation power.

\subsubsection*{Generalized peak-and-tail family}
\label{app:generalized_peak_tail}

To construct the controlled states used in Example~2 below, we introduce a block-structured extension of
Eq.~\eqref{eq:peak_tail_family_app}. {Let
$\boldsymbol{m}=(m\subtiny{0}{0}{1},\ldots,m_R),$ and $\boldsymbol{r}=(r\subtiny{0}{0}{1},\ldots,r_R),$} where $m\subtiny{0}{0}{j}\in\mathbb{N}$ is a block multiplicity and $r\subtiny{0}{0}{j}\geq0$ is its
unnormalized population height. We define
\begin{equation}
\label{eq:generalized_peak_tail_family_app}
p^{(\boldsymbol{m},\boldsymbol{r})}
=
\frac{1}{Z(\boldsymbol{m},\boldsymbol{r})}
\left(
\underbrace{r\subtiny{0}{0}{1},\ldots,r\subtiny{0}{0}{1}}_{m\subtiny{0}{0}{1}\ \mathrm{terms}},
\underbrace{r\subtiny{0}{0}{2},\ldots,r\subtiny{0}{0}{2}}_{m\subtiny{0}{0}{2}\ \mathrm{terms}},
\ldots,
\underbrace{r_R,\ldots,r_R}_{m_R\ \mathrm{terms}}
\right),
\end{equation}
{with $Z(\boldsymbol{m},\boldsymbol{r})
= \sum_{j=1}^{R}m\subtiny{0}{0}{j}r\subtiny{0}{0}{j},$ and support size given by $K=\sum_{j:\,r_j>0}^{R}m\subtiny{0}{0}{j}$.} At least one $r\subtiny{0}{0}{j}$ must be strictly positive. The ordinary one-peak plus flat-tail family is recovered as the two-block case
\begin{equation}
\boldsymbol{m}=(1,M),
\qquad
\boldsymbol{r}
=
\left(
a,\frac{1-a}{M}
\right),
\end{equation}
for which $Z=1$. The generalized family allows several population plateaus with different multiplicities and heights. This additional freedom makes it possible to match selected low-order population moments while changing higher-order moments and, consequently, the organization of
pair-resolved spectral power.

Using the population moments defined in
Eq.~\eqref{eq:population_moments_app}, the three
effective dimensions for pure initial states
can be expressed as
\begin{equation}
\label{eq:descriptor_moments_app}
\deff
=
\frac{1}{S_2[p]},
\qquad
\deffdyn
=
\frac{(1-S_2[p])\suptiny{0}{0}{2}}{
2\left(S_2[p]\suptiny{0}{0}{2}-S_4[p]\right)
},
\qquad
\deffspec
=
\frac{
\left(S_2[p]\suptiny{0}{0}{2}-S_4[p]\right)\suptiny{0}{0}{2}
}{
2\left(S_4[p]\suptiny{0}{0}{2}-S_8[p]\right)
}.
\end{equation}
Moreover, the total pair-resolved power is
\begin{equation}
\label{eq:total_pair_power_app}
\sum\subtiny{0}{0}{n<m}
\spd(\omega\subtiny{0}{0}{nm})
=
\sum\subtiny{0}{0}{n<m}
p\subtiny{0}{0}{n}\suptiny{0}{0}{2}
p\subtiny{0}{0}{m}\suptiny{0}{0}{2}
=
\frac{S_2[p]\suptiny{0}{0}{2}-S_4[p]}{2} = S\subtiny{0}{0}{2}.
\end{equation}
It follows that matching $S_2[p]$ and $S_4[p]$
fixes $\deff$, $\deffdyn$,
$\leakz\suptiny{0}{0}{\infty}$,
and the total fluctuation power.
Different values of $S_8[p]$, however,
produce different values of $\deffspec$.
The complete population profile also determines
$\Hpow$, which characterizes the overall
spread of the normalized spectral power distribution.

\subsection{Controlled numerical examples}
\label{secapp:controlled_examples}

We now complement the static population-level analysis with two explicit dynamical examples. These examples illustrate, directly at the level of $\delta\!\leakz(t)$, how $\deff$, $\deffdyn$, and $\deffspec$ control different features of the LFF dynamics.

In Example~1, two states are chosen with nearly identical effective dimensions $\deff$ but distinct dynamical dimensions $\deffdyn$. This isolates the role of $\deffdyn$ in controlling the magnitude of temporal fluctuations. In Example~2, two states are constructed with identical $\deff$ and $\deffdyn$, but different $\deffspec$ and $\Hpow$. This isolates the frequency-domain organization of the fluctuation power while keeping both the equilibration value and the asymptotic fluctuation scale fixed.

\subsubsection*{Example 1: similar \texorpdfstring{$\deff$}{deff} but distinct \texorpdfstring{$\deffdyn$}{deffdyn}}
\label{app:example1}

To illustrate the independent role of the dynamical dimension, we consider two initial states, denoted by $A$ and $B$, constructed so that their effective dimensions are nearly identical while their pair-amplitude distributions differ
substantially. For these states, we obtain
\begin{align}
\deff(A)&\approx 2.606,
&
\deffdyn(A)&\approx 2.330,
&
\leakz\suptiny{0}{0}{\infty}(A)&\approx 0.616,
\\
\deff(B)&\approx 2.700,
&
\deffdyn(B)&\approx 16.716,
&
\leakz\suptiny{0}{0}{\infty}(B)&\approx 0.630.
\end{align}
Thus, the two states have almost the same equilibrium LFF value, since
$\leakz\suptiny{0}{0}{\infty}
= 1-\frac{1}{\deff}.$ Their dynamical dimensions, however, differ by more than a factor of
seven. By Theorem~\ref{thm:LFF_variance_clean}, the typical long-time
fluctuation amplitude is $\sqrt{\avginfty{|\delta\!\leakz(t)|\suptiny{0}{0}{2}}}
= (1-\frac1\deff)/\sqrt{2\deffdyn}$. Therefore, for two states labeled $1$ and $2$, the ratio of their typical fluctuation amplitudes is
\begin{equation}
\frac{
\sqrt{
\avginfty{
|\delta\!\leak(A)|\suptiny{0}{0}{2}
}
}
}{
\sqrt{
\avginfty{
|\delta\!\leak(B)|\suptiny{0}{0}{2}
}
}
}
\approx
2.62.
\end{equation}
This predicts that state $A$ should display substantially larger oscillations than state $B$, at long times, despite the two states having nearly the same value of $\deff$.

{Panel~(a) of Fig.~\ref{fig:figs6_7} confirms this prediction. The two states exhibit similar LFF equilibrium values, but their fluctuation amplitudes differ markedly. This behavior is precisely what $\deffdyn$ is designed to capture. A larger $\deffdyn$ indicates that the oscillatory contributions are distributed over a larger effective number of cosine channels, resulting in a smaller long-time mean-square fluctuation for comparable long-time equilibrium values. Conversely, a smaller $\deffdyn$ indicates a greater concentration of the oscillatory weights among fewer channels, resulting in larger fluctuations.
}

{
This example elucidates the complementary roles of $\deff$ and $\deffdyn$: the former determines the equilibrium value of the LFF, whereas the latter provides the additional information required to characterize the long-time fluctuation scale around that value.}

\subsubsection*{Example 2: nearly identical \texorpdfstring{$\deff$}{deff} and \texorpdfstring{$\deffdyn$}{deffdyn} but distinct spectral organization}
\label{app:example2}

The second example probes the independent role of the spectral dimension. We construct two controlled pure initial states, again denoted by $A$ and $B$, from the generalized peak-and-tail family in Eq.~\eqref{eq:generalized_peak_tail_family_app}. Both states are supported on the same set of $48$ energy eigenstates of the $N=10$ Hamiltonian. Consequently, they share the same set of available pairwise Bohr frequencies; only the populations assigned to those levels are changed.

State $A$ is specified by the following block multiplicities and raw heights $\boldsymbol{m}\subtiny{0}{0}{A}=(34,2,10,2),$ and $\boldsymbol{r}\subtiny{0}{0}{A}=
\left(
0.05696,\,
0.85627,\,
0.08207,\,
0.00469
\right).$ 
State $B$ is specified by $\boldsymbol{m}\subtiny{0}{0}{B}=(1,19,6,22),$ and $
\boldsymbol{r}\subtiny{0}{0}{B}=
\left(
1,\,
0.177681,\,
0.003862,\,
0.000893
\right).$ In each case, the raw heights are normalized according to $Z(\boldsymbol{m},\boldsymbol{r})
= \sum_{j=1}^{R}m\subtiny{0}{0}{j}r\subtiny{0}{0}{j}$. The multiplicities sum to $48$ for both states.

{The two distributions were constructed
to have nearly identical second and fourth population moments:} $S\subtiny{0}{0}{2}(A)=S\subtiny{0}{0}{2}(B)\approx 0.081940$, and $S\subtiny{0}{0}{4}(A)=S\subtiny{0}{0}{4}(B) \approx 0.002672.$ Their eighth moments are nevertheless different, $S\subtiny{0}{0}{8}(A)\approx 3.566\times10^{-6},$ and $S\subtiny{0}{0}{8}(B)\approx 6.879\times10^{-6}.$ According to Eq.~\eqref{eq:descriptor_moments_app}, this construction
fixes $\deff$ and $\deffdyn$ while allowing $\deffspec$ to change. The resulting descriptors are
\begin{align}
\deff(A){\simeq}\deff(B)
&\approx 12.204,
\\
\deffdyn(A){\simeq}\deffdyn(B)
&\approx 104.267,
\\
\leakz\suptiny{0}{0}{\infty}(A)=\leakz\suptiny{0}{0}{\infty}(B)
&\approx 0.918.
\end{align}
The common theoretical long-time variance given by Theorem~\ref{thm:LFF_variance_clean} is
\begin{equation}
\label{eq:example2_common_variance_app}
\avginfty{|\delta\!\leakz(t)|\suptiny{0}{0}{2}}
=
\frac{
\left(1-\frac{1}{\deff}\right)\suptiny{0}{0}{2}
}{
2\deffdyn
}
\approx
4.042\times10^{-3}.
\end{equation}
{Despite these closely matched values,} the frequency-domain descriptors differ
substantially:
\begin{align}
\deffspec(A)&\approx 2.283,
&
\Hpow(A)&\approx 2.257,
\\
\deffspec(B)&\approx 31.066,
&
\Hpow(B)&\approx 3.961.
\end{align}
Thus, the effective number of spectral-power channels for state $B$ is more than an order of magnitude larger than for state $A$, even though the two states have the same equilibration value and the same asymptotic fluctuation variance.

{To illustrate the complementary information provided by $\deffspec$ and $\Hpow$, panel~(b) of Fig.~\ref{fig:figs6_7} compares two controlled pure initial states, $A$ and $B$, constructed as described in Example~2. Their energy populations exhibit several plateaus with different heights and multiplicities, and both states are supported on the same $48$ eigenstates of the $N=10$ Ising-like Hamiltonian. For pure states, $\deff=1/\sum\subtiny{0}{0}{n}p\subtiny{0}{0}{n}\suptiny{0}{0}{2}$ and $\deffdyn=\big(1-\sum\subtiny{0}{0}{n}p\subtiny{0}{0}{n}\suptiny{0}{0}{2}\big)\suptiny{0}{0}{2}/\big\{2\big[\big(\sum\subtiny{0}{0}{n}p\subtiny{0}{0}{n}\suptiny{0}{0}{2}\big)\suptiny{0}{0}{2}-\sum\subtiny{0}{0}{n}p\subtiny{0}{0}{n}\suptiny{0}{0}{4}\big]\big\}$ depend only on the second and fourth population moments [Eq.~\eqref{eq:population_moments_app}]. Matching these two moments while keeping higher moments distinct yields the same equilibrium value and the same long-time variance, Eq.~\eqref{eq:variancedeff_dyn}, but markedly different values of $\deffspec$ and $\Hpow$. This example demonstrates that even when the equilibrium value and the total fluctuation power are fixed, the spectral organization of the LFF fluctuations can differ substantially, as quantified by $\deffspec$ and $\Hpow$.}

\end{document}